\documentclass{article}
\usepackage{graphicx} 
\usepackage{tabularx,colortbl}

\usepackage[numbers]{natbib}

\usepackage{comment}
\usepackage[utf8]{inputenc}
\usepackage{amsmath}
\usepackage{amsthm}
\usepackage[colorlinks,linkcolor=blue,citecolor=red]{hyperref}
\usepackage{color}
\usepackage{graphicx}
\usepackage[dvipsnames]{xcolor}
\usepackage{tikz}
\usepackage{amsfonts}
\usepackage{soul}
\usepackage{cancel}
\usepackage[normalem]{ulem}
\usepackage{subcaption}
\usepackage{authblk}
\usepackage[export]{adjustbox}
\usepackage{float}
\usepackage{cleveref}

\usepackage[paperwidth=199.8mm,paperheight=287mm,centering,hmargin=25mm,vmargin=2cm]{geometry}
\usepackage{amsmath,bm}

\usepackage{algpseudocode}

\usepackage{ulem}

\makeatletter
\newcommand{\linethrough}{\mathpalette\@thickbar}
\newcommand{\@thickbar}[2]{{#1\mkern0mu\vbox{
    \sbox\z@{$#1#2\mkern-1.5mu$}%
    \dimen@=\dimexpr\ht\tw@-\ht\z@+2\p@\relax 
    \hrule\@height0.5\p@ 
    \vskip\dimen@
    \box\z@}}
}
\makeatother

\newcommand{\red}[1]{{\color{red}#1}}

\newcommand{\blue}[1]{{\color{blue}#1}}

\newcommand{\ket}[1]{\left\vert #1 \right\rangle}
\newcommand{\bra}[1]{\left\langle #1 \right\vert}

\newcommand{\ketbra}[2]{\vert#1\rangle\!\langle#2\vert}
\newcommand{\braket}[2]{\langle #1 
\vert#2\rangle}

\newcommand{\new}[1]{\blue{#1}}
\newcommand{\old}[1]{\red{\sout{#1}}}

\newcommand{\cR}[0]{\mathcal{R}}
\newcommand{\cB}[0]{\mathcal{B}}
\newcommand{\cD}[0]{\mathcal{D}}
\newcommand{\cW}[0]{\mathcal{W}}

\newcommand{\cO}[0]{\mathcal{O}}

\newcommand{\cP}[0]{\mathcal{P}}
\newcommand{\cH}[0]{\mathcal{H}}

\newcommand{\cA}[0]{\mathcal{A}}
\newcommand{\cL}[0]{\mathcal{L}}
\newcommand{\cI}[0]{\mathcal{I}}

\newcommand{\cN}[0]{\mathcal{N}}

\newcommand{\bR}[0]{\mathbb{R}}
\newcommand{\bN}[0]{\mathbb{N}}

\newcommand{\bj}[0]{{\bm{j}}}
\newcommand{\bb}[0]{{\bm{b}}}
\newcommand{\bl}[0]{{\bm{\lambda}}}

\newcommand{\BQP}{\textsf{BQP}}
\newcommand{\NP}{\textsf{NP}}
\newcommand{\BPP}{\textsf{BPP}}

\newcommand{\QMA}{\textsf{QMA}}

\usepackage{thmtools}
\usepackage{thm-restate}

\newtheorem{definition}{Definition}

\newtheorem{lemma}[definition]{Lemma}
\newtheorem{algorithm}[definition]{Algorithm}

\newtheorem{remark}[definition]{Remark}
\newtheorem{theorem}[definition]{Theorem}
\newtheorem{assumption}[definition]{Assumption}
\newtheorem{corollary}[definition]{Corollary}

\newtheorem{claim}[definition]{Claim}

\usepackage{algorithm,algpseudocode}

\newcommand{\TCtime}[0]{\mc O(n^ 2poly(\log n))}

\usepackage{enumitem}
\usepackage{ifthen}
\usepackage{intcalc}
\usetikzlibrary{decorations.fractals, spy}
\usetikzlibrary{positioning, calc, fit, patterns}
\pgfdeclarelayer{bg}    
\pgfsetlayers{bg,main}

\newcommand{\mc}{\mathcal}
\newcommand{\alg}[1]{\langle #1 \rangle}
\newcommand{\Id}{\mathbb I}
\newcommand{\bs}[1]{\boldsymbol{#1}}
\renewcommand{\tilde}{\widetilde}

\makeatletter
\tikzset{
  fitting node/.style={
    inner sep=0pt,
    fill=none,
    draw=none,
    reset transform,
    fit={(\pgf@pathminx,\pgf@pathminy) (\pgf@pathmaxx,\pgf@pathmaxy)}
  },
  reset transform/.code={\pgftransformreset}
}
\makeatother

\begin{document}
\title{Gibbs state  preparation for commuting Hamiltonian: Mapping to classical Gibbs sampling}
\date{}

\author[1]{Yeongwoo Hwang$^\ddagger$~\thanks{yeongwoohwang@g.harvard.edu}}
\affil[1]{School of Engineering and Applied Sciences, Harvard University}

\author[2]{Jiaqing Jiang$^\ddagger$~\thanks{jiaqingjiang95@gmail.com}}
\affil[2]{Department of Computing and Mathematical Sciences, California Institute of Technology}
\affil[2]{Simons Institute for the Theory of Computing, University of California, Berkeley
}

\def\thefootnote{$\ddagger$}\footnotetext{Both authors contributed equally (listed in alphabetical order).}\def\thefootnote{\arabic{footnote}}

\maketitle

\begin{abstract}

Gibbs state preparation, also known as Gibbs sampling, is a key computational technique extensively used in physics, statistics, and  other scientific fields. 
 Recent efforts~\cite{capel2020modified,bardet2023rapid,kastoryano2016quantum} for designing fast mixing Gibbs samplers for quantum Hamiltonians have  largely focused on  commuting local Hamiltonians (CLHs),  a non-trivial subclass of Hamiltonians which include highly entangled systems such as the Toric code and quantum double model. Most previous Gibbs samplers relied on simulating the Davies generator, which is a Lindbladian associated with the thermalization process in nature. 

Instead of using the Davies generator, we design a different Gibbs sampler for various CLHs by giving a reduction to
classical Hamiltonians. More precisely, we show that one can efficiently prepare the Gibbs state for a CLH $H$ on a
quantum computer as long as one can efficiently perform classical Gibbs sampling for a corresponding classical
Hamiltonian $H^{(c)}$. Combining our results with existing fast mixing classical Gibbs samplers, we are able to replicate state-of-the-art results \cite{capel2020modified,bardet2023rapid,kastoryano2016quantum} as well as prepare the Gibbs state in regimes which were previously unknown, such as the low temperature region (as long as there exists fast mixing Gibbs samplers for the corresponding classical Hamiltonians).

Our reductions are as follows. 
\begin{itemize}
    \item If $H$ is a 2-local qudit CLH, then $H^{(c)}$ is a 2-local qudit classical Hamiltonian.
    \item 
   If $H$ is a 4-local qubit CLH on 2D lattice and there are no classical qubits, then  $H^{(c)}$ is a 2-local qudit classical Hamiltonian on a planar graph. As an example, our algorithm can prepare the Gibbs state  for the (defected) Toric code at any non-zero temperature in $\TCtime$ time.
    \item If $H$ is a  4-local qubit CLH on 2D lattice and there are  classical qubits, assuming that quantum terms are uniformly correctable, then  $H^{(c)}$ is a constant-local classical Hamiltonian.
\end{itemize}

A further consequence of our work is on the complexity of quantum approximate counting (\textsf{QAC}). Classically, Stockmeyer \cite{stockmeyer1983complexity} showed that classical approximate counting is in $\BPP^{\NP}$, but no analogous quantum result is known. We show that for various CLHs we can place the corresponding \textsf{QAC} in $\BQP^\textsf{CS}$, where \textsf{CS} is an oracle capable of performing arbitrary classical Gibbs sampling.

\end{abstract}

\newpage

\tableofcontents

\newpage

\section{Introduction}

Gibbs state  preparation is a key computational technique used in physics, statistics, and many other scientific fields.
Given a local Hamiltonian $H$ and an inverse temperature $\beta$, the  Gibbs state  $\rho_{\beta H}\sim\exp(-\beta H)$
describes the thermal equilibrium properties of quantum systems at finite temperature, making them essential for
studying the phase diagram, stability of topological quantum memory~\cite{hastings2011topological,landon2013local} as
well as the thermalization process~\cite{riera2012thermalization,muller2015thermalization}. In addition to physics,
Gibbs state  preparation  also has found various applications in optimization~\cite{brandao2019quantum,van2017quantum} and Bayesian Inference~\cite{amin2018quantum,hinton1983optimal}. 
 Various Gibbs state preparation algorithms (or \emph{Gibbs samplers}) have been proposed, including approaches inspired
 by the Davies generator \cite{chen2023quantum,gilyen2024quantum,rall2023thermal,ding2024efficient}, the
 Metropolis-Hasting type method~\cite{jiang2024quantum,temme2011quantum}, and ones based on Grover
 amplification~\cite{poulin2009sampling} and the Quantum Singular Value Transform~\cite{gilyen2019quantum}.

The key requirement for a good Gibbs sampler is \textit{fast mixing}, that is, the algorithm prepares the Gibbs state in \textit{polynomial} time. Gibbs samplers for classical Hamiltonians  have been   studied for decades and fast mixing algorithms have been successfully designed for various scenarios.  In particular, Glauber dynamics yield fast mixing Gibbs samplers
for 1D systems at any constant temperature~\cite{guionnet2003lectures,holley1985rapid,holley1989uniform} and for 2D systems at high temperature~\cite{martinelli1994approach}. On the other hand, for 2D systems like the Ising model, Glauber dynamics-based samplers are known to suffer critical slow downs and are slow mixing at low temperature~\cite{chayes1987exponential,cesi1996two,schonmann1987second}. Instead, the Swendsen-Wang algorithm~\cite{swendsen1987nonuniversal,feng2023swendsen} and approaches based on Barvinok's method~\cite{borgs2020efficient} were proved to achieve fast mixing for low temperature 2D systems.

Recent efforts on developing fast mixing  Gibbs samplers for quantum Hamiltonians have largely focused on commuting local Hamiltonians (CLHs). CLHs are an important subclass of quantum Hamiltonian which exhibits non-trivial quantum phenomenon. Different from classical Hamiltonians whose eigenstates are computational basis,
the eigenstates of a CLH instance can be highly entangled and cannot be prepared by any constant depth quantum  circuit,  as is true for the famous example of Kitaev’s Toric
code~\cite{kitaev2003fault}.  Besides, it was shown that Gibbs sampling of CLHs at constant temperature remains classically hard~\cite{bergamaschi2024quantum,rajakumar2024gibbs}. Nonetheless, several aforementioned fast mixing results for classical Hamiltonians have been successfully generalized to the CLH case.  In particular, multiple results utilize the Davies generator~\cite{davies1976quantum,davies1979generators}, which represents a quantum Markov chain (Lindbladian) for  thermalization process in the weak coupling limit. This has yielded 
fast mixing Gibbs samplers for 1D CLH at any constant temperature ~\cite{bardet2023rapid,kastoryano2016quantum} and 2D CLH at high temperature~\cite{capel2020modified,kastoryano2016quantum}. The fast mixing proofs are obtained  
 by  generalizing classical techniques for analyzing the mixing time of transition matrices~\cite{guionnet2003lectures,holley1985rapid,holley1989uniform,martinelli1994approach,martinelli1994approach} to analyzing the mixing time of the Davies generator. These generalizations are highly non-trivial and very technical since analyzing the spectrum of a quantum operator (in this case the Davies generator) is generally hard. In the low temperature regime, fast mixing Gibbs samplers for CLHs on two or higher dimensions are only known for the standard 2D Toric code~\cite{alicki2009thermalization,Ding2024PolynomialTimePO}.

In this work, we introduce a new approach which does not use the Davies generator. Instead, we design new Gibbs samplers
for various CLHs by giving a reduction from quantum Gibbs state preparation to classical Gibbs sampling. Combined with
the existing fast mixing results for classical Hamiltonians, our Gibbs sampler is able to replicate the state-of-the-art
performances mentioned above. Furthermore, our algorithm can prepare low temperature Gibbs states as long as there
exists a fast mixing Gibbs sampler for the corresponding classical Hamiltonians. Our reductions are summarized in the
following two theorems. More details and comparisons between previous results and the performance of our Gibbs sampler are contained in \Cref{sec:compare}.

Roughly speaking,
we say that there is a \textit{Gibbs sampling reduction}  from a quantum Hamiltonian $H$ to a classical Hamiltonian $H^{(c)}$ if, given the existence of an  algorithm that performs Gibbs sampling for $( H^{(c)},\beta)$ in time $T$, there exists a quantum algorithm preparing the quantum Gibbs state  for $(H,\beta)$ in time $T$ plus a small overhead polynomial in the number of qubits.\footnote{In our case, we will obtain scaling like $T+\mc O(n)$  or $T+\mc O(n^2)$ or $T\times O(n)$.}
First we notice that the Structure Lemma, which is the key technique used in studying the complexity of CLHs~\cite{bravyi2003commutative,irani2023commuting,aharonov2018complexity,schuch2011complexity}, directly gives the desired reduction for 2-local CLHs.
\begin{restatable}[Informal version of Theorem \ref{thm:2local}]{theorem}{twolocalinformal}
\label{thm:intro2local}
There is a Gibbs sampling reduction  from 
      $2$-local qudit commuting Hamiltonians to 2-local qudit classical Hamiltonians.
\end{restatable}

For more physically motivated  4-local CLHs such as the Toric code, the Structure Lemma can no longer  transform   4-local CLHs to classical Hamiltonians. Instead,  by leveraging a symmetry in the eigenspaces,
we demonstrate that an oblivious randomized correction approach  yields the desired reduction for qubit CLHs in 2D, via generalizing a   ground state preparation algorithm~\cite{aharonov2018complexity} to  Gibbs state preparation.  
The locality of the resulting classical Hamiltonian depends on whether there are classical qubits in the CLHs.  Roughly speaking a qubit is classical if by choosing proper basis of this qubit, all terms look like $\ketbra {0}{0}\otimes ... + \ketbra{1} 1\otimes ...$ on this qubit.

\begin{restatable}[Informal version of Theorem  \ref{thm:4local2} and \ref{thm:gibbs_with_classical} ]{theorem}{fourlocalinformal}
\label{thm:intro4local} 
There is a Gibbs sampling reduction from 
 4-local qubit 2D commuting local Hamiltonian 
$H$
 to qudit classical Hamiltonians. In particular, 
\begin{itemize}
    \item If  there are no classical qubits with respect to terms in $H$, then the classical Hamiltonian is a 2-local qudit classical Hamiltonian on a planar graph. 
    \item When there are classical qubits but all quantum terms (terms far away from classical qubits) are uniformly correctable, then the classical Hamiltonian is a $\mc O(1)$-local qudit classical Hamiltonian.
\end{itemize}
\end{restatable}
Note that in the above theorem the quantum Hamiltonian is on qubits while the classical Hamiltonian is on qudits.  
An example of a qubit 2D CLH without classical qubits is the defected Toric code, a generalization of the Toric code with arbitrary complex coefficients. We will give a technical overview based on the defected Toric code in \Cref{sec:defected_toric_code}. 

Our reduction also gives a quantum analogy of Stockmeyer's result~\cite{stockmeyer1983complexity} for the complexity of quantum approximate counting.  In particular, a fundamental result of Stockmeyer~\cite{stockmeyer1983complexity} states that classical approximate counting  (approximating the partition function of a classical Hamiltonian)  is contained in the complexity class $\BPP^{\NP}$. 
It is natural to conjecture that the quantum approximate counting is upper bounded by a complexity class like
$\BQP^{\QMA}$, but few results are known. By the connection between quantum approximate counting and the Gibbs state  preparation~\cite{bravyi2021complexity}\footnote{Lemma 12 in \cite{bravyi2021complexity}, where the k-QMV can be estimated by measuring the Gibbs state .},
our reduction shows that for various CLHs, the corresponding quantum approximate counting problem is contained in $\BQP^{\textsf{CS}}$, where $\textsf{CS}$ is an oracle which can perform arbitrary classical Gibbs sampling.

\subsection{Comparison to previous work.}\label{sec:compare}
Recall that most of the previous work on  Gibbs samplers for CLHs are based on simulating the Davies generator, which is a Lindbladian closely related to the thermalization process.
In this section, we give a detailed comparison between previous results and our result, demonstrating that instead of using the Davies generator, our reduction gives a new  Gibbs sampler for CLHs 
directly utilizing fast mixing Gibbs samplers for classical Hamiltonians. In particular, our reduction is able to replicate state-of-the-art results and also derive new results. Our results are summarized in \Cref{table:comparison}. In this section we discuss related results, and a more thorough discussion on Gibbs state preparation can be found in  
 \cite[Section 3.3]{capel2020modified}.

\begin{remark}
 Due to the relationship between Davies generator and thermalization process, previous works analyzing the Davies
 generator~\cite{bardet2023rapid,kastoryano2016quantum,alicki2009thermalization} also yield insights into how thermal noise influences the quantum systems. Our reduction does not cover this implication. The following comparison is only for the task of preparing Gibbs state.
\end{remark}

\begin{figure}[h!]
\centering
\begingroup\setlength{\fboxsep}{0pt}
\begin{tabularx}{\textwidth}{llp{0.3\textwidth}p{0.3\textwidth}} 
\rowcolor{lightgray!50}
 \hline
 Mixing time & 1D any temp & 2D (2-local)  & 2D (4-local) \\ 
\rowcolor{lightgray!20}
 \hline
 Previous work & $poly(\log n)$\textsuperscript{\cite{capel2020modified}}  & high temp: $poly(\log n)$\textsuperscript{\cite{kochanowski2024rapid,capel2020modified}}   & $poly(n)$ for high temp\textsuperscript{\cite{kastoryano2016quantum}} \\ 
\rowcolor{lightgray!20}
                &                &low temp: unknown          &  $\mc O(n^3)$ mixing time for TC~\textsuperscript{\cite{Ding2024PolynomialTimePO}} \\
\rowcolor{lightgray!20}
 \hline
 Our results & $poly(\log n)$ & {high temp: ${poly(\log n)}$} &  \textbf{*$\bm{poly(n)}$}    \\ 
\rowcolor{lightgray!20}
              &                &  \textbf{low temp:  *$\bm{poly(n)}$}   & \textbf{$\bm{\TCtime}$ for DTC} \\
 \hline
\end{tabularx}
\endgroup
\vspace{-0.5em}
\captionsetup{width=0.9\textwidth, font=small}
\caption{Mixing time of Gibbs samplers for 1D and 2D CLHs at different  temperatures (temp). TC and DTC refer to the standard Toric code and the more general defected Toric code respectively. Improved results are in \textbf{bold}. \emph{The results marked ``*'' indicate we achieve this mixing time, when efficient samplers for the corresponding \emph{classical} Hamiltonians are known.} Our $\TCtime$ result for DTC is the total runtime rather than just mixing time.}
\label{table:comparison}
\end{figure}

\paragraph{Review of Markov chains, Lindbladians and mixing time}  We briefly review some key concepts. Consider an $n$-qubit local Hamiltonian $H$ and an inverse temperature $\beta$. We assume $\beta \in \mathcal O(1)$ unless further specified. We will first assume $H$ is classical and introduce key concepts for classical Gibbs sampling. Then we will generalize to quantum Hamiltonians. 
 
Suppose $H$ is a classical Hamiltonian which diagonalizes in the computational basis, and the task is performing
classical Gibbs sampling for $(H,\beta)$.  The target is the  Gibbs distribution $\pi$  which samples  computational
basis states $\ket{x}$ with probability proportional to $\exp\left(-\beta \langle x|H|x\rangle\right)$. 
 The goal of classical Gibbs sampling is to design  a classical  process which drives any distribution $\nu$ to $\pi$.  The  commonly used method is the classical Metropolis algorithm~\cite{metropolis1953equation}, which is 
   a  discrete-time Markov chain described by  a 
transition matrix $P$,  
  such that $\pi$ is the unique  fixed point of $P$, i.e. $P\pi=\pi$. 
 The mixing time $t(\epsilon)$ is the time needed to get $\epsilon$-close to the invariant distribution $\pi$ with respect to $1$-norm (the total variation distance), that is
 \begin{align}
 t(\epsilon):=\min\{ t: \|P^t\nu-\pi\|_1\leq \epsilon,\forall \text{ distribution } \nu \}.
 \end{align}
  In addition to this discrete-time Markov chain,
one can also design a continuous-time Markov chain, described by a  
generator matrix $G$ such that $\pi$ is the unique invariant distribution of $G$, i.e. $G\pi=0$ or equivalently
$e^{Gt}\pi=\pi$, $\forall t$. Similarly to above,  the mixing time  is defined to be
\begin{align}
	t(\epsilon):=\min\{ t: \|e^{Gt}\nu-\pi\|_1\leq \epsilon,\forall \text{ distribution } \nu \}.
\end{align}
\begin{itemize} 
    \item 
    The Markov chain is \textit{poly-time mixing}, or \emph{fast mixing}, if the the spectral gap of $P$ or $G$ is $\Omega(1/poly(n))$, which implies $t(\epsilon)=poly(n)\times \log \frac{1}{\epsilon}$.
    \item
    The Markov chain is \textit{rapid mixing} if it reaches the invariant distribution in a time which scales logarithmically with the system size, that is $t(\epsilon)= poly(\log n) \times \log\frac{1}{\epsilon}$. Rapid mixing is typically proved by bounding the log-Sobolev constant~\cite{guionnet2003lectures} for the continuous-time chain.
\end{itemize}

When $H$ is a quantum Hamiltonian, we wish to prepare a quantum Gibbs state, defined as
\[
    \rho_{\beta H}:= \rho( H,\beta ):= \frac 1{\text{tr}[\exp(-\beta H)]}\exp(-\beta H).
\]
The goal is to  design a quantum process which drives any quantum state $\sigma$ to $\rho_{\beta H}$. One commonly used method is to  design a Lindbladian $\cL$ such that $\rho_{\beta H}$ is the unique fixed point of $\cL$, i.e. $\cL(\rho_{\beta H})=0$ or, equivalently, $e^{\cL t}(\rho_{\beta H})=\rho_{\beta H}, \forall t$. The Lindbladian is the quantum analogy of a continuous-time Markov chain. The mixing time $t(\epsilon)$ is defined to be
 \begin{align}
 	t(\epsilon):=\min\{ t: \|e^{\cL t}(\sigma)-\rho_{\beta H}\|_1\leq \epsilon,\forall \sigma \}.
 \end{align}
 The notion of poly-time mixing and rapid mixing is defined similarly to the classical setting. One can prepare the quantum Gibbs state  on a quantum computer by Lindbladian simulation techniques~\cite{cleve2016efficient,chen2023quantum}. For simplicity we will assume that $\epsilon=1/poly(n)$ for the remainder of the section. 

\paragraph{Replication of poly-log time mixing algorithm for 1D commuting Hamiltonians.} For  1D classical Hamiltonians, it is well-known that  there is no \textit{computational phase transition}~\cite{guionnet2003lectures,holley1985rapid,holley1989uniform}. As a result, for any constant temperature, Glauber dynamics is rapid mixing for all translation-invariant,  1D classical Hamiltonians with finite-range interactions.
 
A large body of previous work has focused on generalizing the rapid mixing results from classical Hamiltonians to quantum CLHs. In particular, for 1D CLHs, \cite{kastoryano2016quantum} proved that the Davies generator has a constant spectral gap and thus is fast (poly-time) mixing. Then
 Bardet \textit{et.al.}~\cite{bardet2023rapid} strengthened the result to obtain rapid mixing. 
  Specifically, they
proved that for any constant temperature, the Davis generator $\cL$ for any finite-range, translation-invariant 1D CLHs is rapid mixing. 
This is proved by generalizing the classical technique of bounding the log-Sobolev constant
~\cite{guionnet2003lectures,holley1985rapid,holley1989uniform,zegarlinski1990log} to CLHs. Note that this generalization is highly non-trivial since $\cL$ is a quantum operator and analyzing its spectral gap and log-Sobolev constant are complicated. Combining this bound with quantum simulations of the Lindbladian,  Bardet \textit{et.al.} gives a quantum Gibbs sampler with runtime  
 $$poly(\log n) \times f_1,$$ where $f_1$ is the overhead brought by   simulating  the Lindbladian evolution.

In contrast to \cite{bardet2023rapid}, which obtained a rapidly mixing Gibbs sampler by developing sophisticated techniques to bound the log-Sobolev constant of the Davies generator, 
our reduction gives a Gibbs sampler of similar performance by directly using classical results:
\begin{lemma}[Informal version of Lemma \ref{lem:H1D}]
    There is a Gibbs sampling reduction from any finite-range, translation-invariant (TI) qudit 1D CLHs, to
    the finite-range, TI 1D classical Hamiltonians. Combined with the rapid mixing Gibbs sampler for  finite-range, TI 1D classical Hamiltonian at any constant temperature~\cite{guionnet2003lectures,holley1985rapid,holley1989uniform}, we prepare the Gibbs state in time,
    \[
        poly(\log n)\times \mc \red{f_2}  + \mc O(n).  
    \]
\end{lemma}
Here $f_2$  is the overhead incurred by simulating the classical  Markov chain. $\mc O(n)$ is the time needed to prepare
a constant depth quantum circuit arising from the Structure Lemma which implements the quantum-to-classical reduction.

\paragraph{High and low temperature Gibbs state for 2-local CLHs on 2D.}   Recall that $\beta$ is the inverse temperature, thus low temperature corresponds to large $\beta$. We begin with a literature review for classical Hamiltonians. Unlike 1D classical Hamiltonians where the Glauber dynamics
is rapid mixing for any constant temperature, 2-local 2D classical Hamiltonians exhibit a constant-temperature computational phase transition. For example, the ferromagnetic 2D Ising model has a constant \textit{critical inverse temperature} $\beta_c$, such that
\begin{itemize}
    \item For $\beta< \beta_c$   the Glauber  Dynamics is poly-time mixing~\cite{martinelli1994approach,martinelli1994approach}.
    \item For $\beta\geq \beta_c$, the Glauber dynamics meets a critical slow down where the spectral gap of the Glauber Dynamics is smaller than $\exp(-\alpha\sqrt{n})$ for $\alpha>0$~\cite{chayes1987exponential,cesi1996two,schonmann1987second}. 
\end{itemize}
Similar results also hold for the Potts model~\cite{ullrich2012rapid,gheissari2016mixing}. 
To understand this phase transition intuitively, note that
the Glauber dynamics is a Markov chain with local update rules. Intuitively an algorithm using local updates is good at solving a ``local'' problem. In the high temperature region, most spins will interact effectively weakly thus the Gibbs state has little entanglement~\cite{bakshi2024high} and a local optimization suffices. However, in the low temperature region, there are strong correlations in large regions and thus the Gibbs state is highly non-local.\footnote{More precisely there is an  equivalence between the mixing time and the spatial decay of correlation in the Gibbs measure~\cite{dyer2004mixing,cesi2001quasi,stroock1992equivalence,stroock1992logarithmic}. In some classical literature, decay of correlation is referred to as mixing condition~\cite{guionnet2003lectures} or spatial mixing~\cite{dyer2004mixing}.} 
Thus, to prepare low temperature Gibbs state ,
one needs to carefully design Markov chains with \emph{non-local} update rules, such as the cluster updates in the Swendsen-Wang algorithm~\cite{swendsen1987nonuniversal};  
or uses other methods such as the Barvinok's method~\cite{borgs2020efficient}.

In the quantum case, to the best knowledge of the authors, all previous work on Gibbs state  preparation for 2D CLHs has focused on the high temperature region. 
 In particular,
 \begin{itemize}
     \item \cite{kastoryano2016quantum} showed that there is a constant $\beta_1$ such that for $\beta\leq  \beta_1$ the Davies generator is poly-time mixing.
     \item \cite{capel2020modified} showed that  there is a constant $\beta_2$ such that 
for $\beta\leq  \beta_2$\footnote{ 
We did not check that whether  $\beta_1,\beta_2$ and the later mentioned high temperatures are equal.},   the Schmidt generator defined in \cite{capel2020modified}
is rapid mixing.
 \end{itemize}
  
Both the Davies generator and the Schmidt generator for 2D CLHs with respect to local jump operators are local
Lindbladians. An adaption of the classical proofs~\cite{chayes1987exponential,cesi1996two,schonmann1987second}
will show that they are slow mixing for 2D systems at low temperature~\cite{gamarnik2024slow}. That is, there exists a constant $\beta_3$ such that for $\beta\geq \beta_3$, the spectral gap of any $O(\log n)$-local Lindbladian (not necessarily the Davies generator) which fixes the Gibbs state  of 2D Ising model at  inverse temperature $\beta$ has an exponentially-small spectral gap.

Our Gibbs sampler improves on the existing results in two main ways. First, in the high temperature region our reduction
again gives a way to directly utilize classical results~\cite{cesi2001quasi} and obtain a Gibbs sampler of similar
performance as the best prior work \cite{capel2020modified} (i.e., rapid mixing), without involving heavy proofs for 
  analyzing the log-Sobolev constant of the Schmidt generator like \cite{capel2020modified}.

\begin{lemma}
    There is a Gibbs sampling reduction from any \emph{2-local qudit 2D CLHs} to 2-local qudit 2D classical Hamiltonians. Thus for high enough temperature where there exists rapid mixing classical Gibbs samplers for the corresponding classical Hamiltonians (e.g. from \cite{cesi2001quasi} or Chapter 9 of \cite{guionnet2003lectures}), we can prepare the Gibbs state for the 2D CLHs in time 
    \[
        poly(\log n)\times  f_2+ \mc O(n)\,,
    \]
where $f_2$ is the overhead incurred by simulating the classical Markov chain.
\end{lemma}

Our second contribution is in the low temperature regime. Unlike \cite{kastoryano2016quantum,capel2020modified} which
only work for the high temperature region, our reduction allows us to prepare low-temperature Gibbs states by utilizing
classical techniques such as the Swendsen-Wang algorithm. To the best knowledge of the authors, 
prior work has not addressed low-temperature Gibbs samplers for 2-local CLHs.  
As an example, with our reduction we can obtain the following result,
 \begin{lemma}[Informal version of Lemma \ref{lem:H2D}]
     There is a Gibbs sampling reduction from translation-invariant qubit (2-local) 2D CLHs to 2D Ising model with magnetic fields. Then
     one can prepare the Gibbs state  for the corresponding CLH at low temperature in $poly(n)$ time whenever there are poly-time mixing Gibbs sampler for the corresponding Ising model  at low temperature like \cite{feng2023swendsen}. 
 \end{lemma}
 
A key feature of our reduction is that it is agnostic to the underlying classical Gibbs sampler.  At the critical temperature, when applied to qudit 2-local 2D CLHs with \textit{large} constant qudit dimension $d$,
the Swendsen-Wang algorithm will also mix slowly and have   a spectral gap that is exponentially small in the side length of the lattice~\cite{borgs1999torpid}.
However, our reduction allows us to substitute in other samplers, such as the Gibbs sampler for the low temperature Potts model based on Barvinok's method, which remains poly-time for large $d$~\cite{borgs2020efficient}. 

\paragraph{4-local, 2D commuting Hamiltonian.}
The best prior work is due to \cite{kastoryano2016quantum}, who proved that for high enough temperature, the Davies generator is poly-time mixing for 4-local 2D CLHs. Unlike their result, the mixing time of our algorithm is dependent on the classical Hamiltonian produced by the reduction and thus our results are not directly comparable.

For the \textit{standard} Toric code, a concurrent work \cite{Ding2024PolynomialTimePO} showed that for any inverse
temperature $\beta<+\infty$ (not necessarily constant), Lindbladian dynamics with nonlocal jump operators prepares the
Gibbs state efficiently (for very low temperature, the mixing time is approximately $\mc O(n^3)$).  Our Gibbs sampler is
based on different techniques and gives a $\TCtime$-time  Gibbs state preparation algorithm for the general defected Toric code at any non-zero temperature, where the defected Toric code is the Toric code with arbitrary coefficients.   
Our Gibbs sampler is based on generalizing the standard ground state preparation algorithm for the Toric code (which
measures all stabilizers) to the task of Gibbs state preparation via an oblivious randomized correction technique. We will give a technical overview based on the example of defected Toric code   in Section \ref{sec:TC} and \ref{sec:defected_toric_code}.
We remark  that since our Gibbs sampler is not based on Lindbladian, our results do not offer additional insights into
Lindbladian dynamics, unlike in \cite{Ding2024PolynomialTimePO}. Another related work~\cite{gu2024doped} uses classical Monte Carlo methods to simulate the Gibbs states of t-doped stabilizer Hamiltonians, although without discussing convergence guarantees.

In addition to the defected Toric code, Theorem \ref{thm:intro4local} also works for more general families of qubit CLHs
and can prepare the corresponding Gibbs state  as long as there exists an efficient algorithm for the corresponding classical Gibbs sampling task.

\subsection{Technical overview}\label{sec:TC}
Recall that in \Cref{thm:intro2local}, we construct a Gibbs sampling reduction from 2-local qudit CLHs to 2-local qudit classical Hamiltonians,
\twolocalinformal*
The proof is primarily based on the Structure
Lemma~\cite{bravyi2003commutative,irani2023commuting,aharonov2018complexity,schuch2011complexity}. The Structure Lemma
has been the principle tool in studying the complexity of CLHs and, intuitively, says that one can transform a 2-local
qudit CLH $H^{(2)}$ to a 2-local qudit classical Hamiltonian $H^{(2c)}$ via a constant depth quantum circuit $\mc C_H$.
In other words, there is a \emph{one-to-one} correspondence between the computational basis of the classical Hamiltonian
$H^{(2c)}$ and the eigenstates of the quantum Hamiltonian $H^{(2)}$. By this observation, there is a simple procedure to
sample from the Gibbs state of $H^{(2)}$, i.e. sample an eigenstate of $H^{(2)}$ according to the Gibbs distribution. First, we sample a computational basis state $\ket{\psi}$  from the Gibbs distribution of the classical Hamiltonian $H^{(2c)}$
(via a classical Gibbs sampler). Then, applying a constant depth quantum circuit to $\ket \psi$ yields  an eigenstate of $H^{(2)}$, distributed according to the Gibbs state of $H^{(2)}$. 

For Hamiltonians of higher locality, the exact correspondence present in the 2-local case does not hold. Nonetheless, we
show in \Cref{thm:intro4local} that we can extend our techniques beyond 2-local Hamiltonians. We demonstrate a reduction
from Gibbs sampling of 4-local qubit CLHs in 2D to classical Gibbs sampling.
\fourlocalinformal*

The case ``without classical qubits'' is the simpler setting. Still, even in this case, we can no longer
straightforwardly apply the Structure Lemma as is possible for 2-local Hamiltonians. This is not due to a deficiency in
our techniques; rather, 4-local Hamiltonians can exhibit topological order \cite{kitaev2003fault} and there cannot be a
constant depth quantum circuit $\mc C_H$ as in the previous theorem. However, we observe that the eigenspace of qubit
CLH is symmetric in some sense and via an \emph{oblivious randomized correction} technique, we can adapt an algorithm
for preparing ground state (as given in \cite{aharonov2018complexity}) to preparing a Gibbs state. In particular,
\cite{aharonov2018complexity} proves that any 2D qubit CLH without classical qubits is equivalent to a defected Toric
code permitting boundaries. That is, the ``interior'' terms look like Pauli X or Pauli Z terms and terms on the
``boundary'' have more freedom. The presence of boundaries makes it non-trivial to utilize this equivalence to design a Gibbs sampler for general 2D qubit CLH. We will use the defected Toric code as an example to explain our Gibbs sampler in Section \ref{sec:defected_toric_code}.

When the initial Hamiltonian has classical qubits, the situation becomes more complex, as the connection from
\cite{aharonov2018complexity} between 2D qubit CLH and the Toric code 
only applies when there are no classical qubits. This does not pose a problem in \cite{aharonov2018complexity} as they
simply want to verify ground energy, and an $\NP$ prover can provide a \emph{recursive} restriction of classical qubits
consistent with some ground state. This effectively removes all classical qubits and the resulting Hamiltonian can be
translated into a defected Toric code permitting boundaries.

In our case, we would like to recover the \emph{distribution} over eigenstates, and thus we cannot perform the same
recursive restriction. Additionally, we need our reduction to be efficient and should not depend on the power of an
prover. We develop a propagation lemma to characterize the limits of the recursive restriction. Combined with an
assumption that all fully quantum terms\footnote{Intuitively, this is the set of terms which remains quantum under any
recursive restriction of the classical qubits, see \Cref{def:fully_quantum_term}.} are uniformly correctable (see
\Cref{as:uniformly_correctable}), we will argue that the statement of \cite{aharonov2018complexity} can be modified  to
obtain a Gibbs sampling reduction from any 2D qubit CLH to a constant-locality classical Hamiltonian. 

\tikzset{
    pics/qudits/.style n args={3}{code={
        \node[plaquette,draw=black!40] (q2#3) at ($ (#1,#2) + (.2,.8) $) {$q_2$};
        \node[plaquette,draw=black!40] (q4#3) at ($ (#1,#2) + (.2,.2) $) {$q_4$};
        \node[plaquette,draw=black!40] (q3#3) at ($ (#1,#2) + (.8,.8) $) {$q_3$};
        \node[plaquette,draw=red, line width=0.4mm] (q1#3) at ($ (#1,#2) + (.8,.2) $) {$q_1$};
}}}
\begin{figure}[h!]
    \centering
    \begin{subfigure}[t]{0.3\textwidth}
        \centering
        \begin{tikzpicture}[scale=0.75,every label/.append style={font=\scriptsize}]
            \foreach \x in {0,...,4} {
                \foreach \y in {0,...,4} {
                    \ifthenelse{\equal{\intcalcMod{\x+\y}{2}}{1}}{
                        \draw[fill=black!30] (\x,\y) rectangle ++(1,1);
                    }{
                        \draw (\x,\y) rectangle ++(1,1); 
                    }
                }
            }
            \path (1,3) rectangle ++(1,1) node[fitting node] (ztype) {};
            \path (2,3) rectangle ++(1,1) node[fitting node] (xtype) {};
            \node[label={[label distance=-0.25cm]below left:$Z$}] at (ztype.north east) {};
            \node[label={[label distance=-0.25cm]below right:$Z$}] at (ztype.north west) {};
            \node[label={[label distance=-0.25cm]above right:$Z$}] at (ztype.south west) {};
            \node[label={[label distance=-0.25cm]above left:$Z$}] at (ztype.south east) {};
            
            \node[label={[label distance=-0.25cm]below left:$X$}] at (xtype.north east) {};
            \node[label={[label distance=-0.25cm]below right:$X$}] at (xtype.north west) {};
            \node[label={[label distance=-0.25cm]above right:$X$}] at (xtype.south west) {};
            \node[label={[label distance=-0.25cm]above left:$X$}] at (xtype.south east) {};
        \end{tikzpicture}
        \caption{Partition the plaquettes as Black $\cB$ and $\cW$. Put $Z$ terms on white plaquettes and put $X$ terms on black plaquettes then we get the defected Toric code $H_{DT}$.}
        \label{fig:intro_fig_defected}
    \end{subfigure}
    \quad
    \begin{subfigure}[t]{0.3\textwidth}
        \centering
        \begin{tikzpicture}[
            scale=0.75,
            every node/.append style={font=\tiny},
            every label/.append style={font=\scriptsize},
            qubitlabel/.style={inner sep=0, draw, circle, minimum width=0.2cm, fill=white},    
            plaquette/.style={draw, color=blue, rounded corners, rectangle, minimum width=0.6cm, minimum height=0.6cm,inner sep=0mm}
        ]
            \foreach \x in {0,...,4} {
                \foreach \y in {0,...,4} {
                    \ifthenelse{\equal{\intcalcMod{\x+\y}{2}}{1}}{
                        \draw[fill=black!30] (\x,\y) rectangle ++(1,1);
                    }{
                        \ifthenelse{\equal{\intcalcMod{\y}{2}}{0}}{\draw[pattern=north west lines, pattern color=black] (\x,\y) rectangle ++(1,1);}{}
                    }
                }
            }
            \node[plaquette] (p1) at (1.5,3.5) {$p_1$};
            \node[plaquette] (p2) at (1.5,1.5) {$p_2$};
            \node[plaquette] (p3) at (3.5,1.5) {$p_4$};
            \node[plaquette] (p4) at (3.5,3.5) {$p_3$};

            \path (1,3) rectangle ++(1,1) node[fitting node] (ztype) {};
            \node[qubitlabel] at (ztype.north west) {$u$};
            \node[qubitlabel] at (ztype.north east) {$v$};
            \node[qubitlabel] at (ztype.south west) {$\tau$};
            \node[qubitlabel] at (ztype.south east) {$w$};

            \draw[color=blue, line width=0.5mm] (p1.east) -- (p4.west);
            \draw[color=blue, line width=0.5mm] (p1.south) -- (p2.north);
            \draw[color=blue, line width=0.5mm] (p4.south) -- (p3.north);
            \draw[color=blue, line width=0.5mm] (p2.east) -- (p3.west);
        \end{tikzpicture}
        \caption{Remove the white terms on the odd lines, $H_{DT}$ will become $H_{DT}^{(2)}$ which can be viewed as a qudit 2-local CLH.}
        \label{fig:intro_fig_2_local}
    \end{subfigure}
    \quad
    \begin{subfigure}[t]{0.3\textwidth}
        \centering
        \begin{tikzpicture}[
            scale=0.75,
            every node/.style={transform shape},
            every node/.append style={font=\scriptsize},
            every label/.append style={font=\scriptsize},
            qubitlabel/.style={inner sep=0, draw, circle, minimum width=0.2cm, fill=white},    
            vertexp/.style = {draw, fill=Sepia, circle, minimum width=0.1cm, inner sep=0},
            plaquette/.style={draw, color=blue, rounded corners, rectangle, minimum width=0.8cm, minimum height=0.8cm,inner sep=0mm}
        ]
            \foreach \x in {0,...,4} {
                \foreach \y in {0,...,4} {
                    \ifthenelse{\equal{\intcalcMod{\x+\y}{2}}{1}}{
                        \draw[fill=black!30] (\x,\y) rectangle ++(1,1);
                    }{
                        \ifthenelse{\equal{\intcalcMod{\y}{2}}{0}}{\draw[pattern=north west lines, pattern color=black] (\x,\y) rectangle ++(1,1);}{}
                    }
                }
            }
            \node[plaquette] (p1) at (1.5,3.5) {$p_1$};
            \node[plaquette] (p2) at (1.5,1.5) {$p_2$};
            \node[plaquette] (p3) at (3.5,1.5) {$p_4$};
            \node[plaquette] (p4) at (3.5,3.5) {$p_3$};

            \node[draw=none, fill=white, rounded corners=3, rectangle, minimum width=0.4cm, minimum height=0.4cm,inner sep=0mm] at (2.5,2.5) {$h$};

            \draw[color=blue, line width=0.5mm] (p1.east) -- (p4.west);
            \draw[color=blue, line width=0.5mm] (p1.south) -- (p2.north);
            \draw[color=blue, line width=0.5mm] (p4.south) -- (p3.north);
            \draw[color=blue, line width=0.5mm] (p2.east) -- (p3.west);

            \draw[color=red, line width=0.5mm] (2,3) node[vertexp] {} -- (2,4) node[vertexp] {}-- (2,5) node[vertexp] {};
        \end{tikzpicture}
        \caption{A correction operator $L_h = X^{\otimes 3}$ (in red) for the removed white term $h$.}
        \label{fig:intro_fig_correct}
    \end{subfigure}
    \caption{}
    \label{fig:intro_defected}
\end{figure}

\paragraph{Extensions to more general 2D Hamiltonians} For simplicity, the theorems above are proved in the setting when
the underlying Hamiltonian is placed on a $2D$ lattice. However, in \cite{aharonov2018complexity} the authors consider
a more general setting of Hamiltonians on \emph{polygonal complexes}. A straightforward generalization of our proofs works in this setting as well.

\subsection{Case study: the punctured defected Toric code}
\label{sec:defected_toric_code}

To illustrate our Gibbs sampler for 2D qubit CLH, that is \Cref{thm:intro4local}, we consider the restricted setting of the defected Toric code.  
In the remainder of this section, 
we assume  that  the inverse temperature is a constant $\beta <+\infty$. We will formally define the
defected Toric code and first for the \textit{punctured} Toric code (defined later) we will  give a $\mc O(n^2)$-time algorithm to prepare its Gibbs states via an oblivious randomized
correction idea. The algorithm for Toric code on torus  use similar ideas and is of runtime $\TCtime$, which  is put 
in Appendix \ref{appendix:non_punctured}.
Those algorithms are specific to the defected Toric code. Then we describe a slightly different algorithm
which is not  as fast as the first algorithm, but by using the tools from \cite{aharonov2018complexity} it can be extended to prepare Gibbs state for general qubit 2D CLHs.
 
The  defected Toric code $H_{DT}$ is embedded on a 2D, $L \times L$ square lattice, with qubits placed on the vertices.
Terms are grouped into ``black'' terms $\cB$ and ``white'' terms $\cW$, as in \Cref{fig:intro_fig_defected}. Formally, we define
\begin{align}
		&H_{DT} =\sum_{p\in \cB} c_pX^p + \sum_{p\in \cW} c_pZ^p, \label{eq:HDT}
\end{align}
where $X$ and $Z$ are the standard Pauli $X$ and $Z$ operators. For a given term $p$ acting on qubits $q_1, \dots, q_4$, $X^p$ denotes $X_{q_1} \otimes \dots \otimes X_{q_4}$ (and same for $Z^p$). The coefficients $c_p$ can be any real number (whereas $c_p = -1$ in the standard Toric code).

\subsubsection{\texorpdfstring{A $\mc O(n^2)$-time algorithm for the punctured defected Toric code.}{A O(n2)-time algorithm for the defected Toric code.}}\label{sec:TCnn}

First, we explain the $\mc O(n^2)$-time algorithm to prepare the Gibbs states of the punctured defected Toric code. Specifically,  the defected Toric code is \emph{punctured} with a white term $p_w$ and a black term $p_b$ missing from the Hamiltonian. 
Thus,
for any term $p'$,
there is a correction operator $L_{p'}$,
which anti-commutes with $p'$ and commutes with all other terms.
If $p'$ is a white term, $L_{p'}$ is realized as
a string of Pauli $X$ operators, starting in the support of $p'$ and ending adjacent to the missing white term $p_w$.
Similarly if $p'$ is a black term, $L_{p'}$ is a string of $Z$'s from $p'$ to $p_b$. To prepare the Gibbs states for the punctured
Hamiltonian $H'_{DT}$, we initialize our state as the maximally mixed state, then sequentially measure and correct each
plaquette term $p$ in $H'_{DT}$. That is, if we measure the plaquette operator $p$ and get measurement outcome
$\lambda\in \{+c_p,-c_p\}$, then we perform the following oblivious randomized correction: 
\begin{itemize}
     \item  With probability $prob:=\frac{\exp(-\beta \lambda)}{\exp(\beta \lambda)+\exp(-\beta \lambda)}$ we do nothing.
     \item  With probability $1-prob$ we apply the correction operation $L_{p}$.
 \end{itemize}
The algorithm correctly prepares the Gibbs states of $H'_{DT}$ because $L_p$ \textit{bijectively} maps the eigenspace
associated  with measurement outcome $+c_p$ to that of $-c_p$ and vice versa. A more detailed description of this
procedure and proof of its correctness can be found in \Cref{sec:reduction_to_classial}. Furthermore, the above
process only performs the correction \emph{once} for each plaquette term, and thus should not be interpreted as an
iterative, randomized Accept/Reject process as done in the general Metropolis algorithm.

In the case where the Toric code has no punctured terms—such as the standard Toric code defined on a torus in the error-correcting code literature—a slightly different algorithm is required.  In this case,
correction operators are strings of Paulis between \emph{two} measured terms $p_1, p_2$. The above outline does not work
directly for the following reason. Suppose $p_1, p_2$ are measured and we obtain eigenvalues $(\lambda_1, \lambda_2) \in
\{\pm c_p\}^{\otimes 2}$. Although, as above, we can bijectively map between $(c_p, c_p) \rightarrow (-c_p, -c_p)$ and
for $(c_p, -c_p) \rightarrow (-c_p, c_p)$, the same is not possible for $(c_p, c_p) \rightarrow (c_p, -c_p)$ and vice
versa. 
Instead, by adapting similar ideas, we show that there exists a Gibbs sampling reduction from the standard Toric code to 1D Ising model. Further details are provided in \Cref{appendix:non_punctured}.

\subsubsection{A generalizable algorithm}
For general 2D qubit CLH (without classical qubits), not every plaquette term has a correction operator; in fact, the
presence of a correction operator qualitatively characterizes the correctable interior terms and the
non-correctable boundary terms. The proof that the interior terms are correctable uses techniques from
\cite{aharonov2018complexity}, and the exterior terms are handled by a reduction to classical Gibbs sampling. The
details are as below. For simplicity, we assume in this section that $H_{DT}$ is embedded on a plane rather than torus,
 and the initial boundary of the lattice naturally plays the roles of the puncture terms 
 in $H_{DT}$.

\paragraph{Reduction to Classical Hamiltonian}\label{par:intro_classical_reduction} As in \cite{aharonov2018complexity}
the first step is to remove enough terms such that the resulting Hamiltonian can be viewed as $2$-local. In the case of
$H_{DT}$, we can simply remove alternating rows of white terms, as in \Cref{fig:intro_fig_2_local}. Finally, grouping
all qubits on a white term as a single $2^4$-dimensional qudit, we see that the white terms become $1$-local and the
black terms are all $2$-local. In \Cref{fig:intro_fig_2_local}, we group qubits $u, v, w, \tau$ to form the qudit $p_1$.
Similarly we form the qudits $p_2,p_3,p_4$.  Then, the black term to the right of qudit $p_1$  becomes 2-local, acting on $p_1$ and $p_3$. Call this $2$-local Hamiltonian $H_{DT}^{(2)}$. The Structure Lemma of \cite{bravyi2003commutative} gives a way to transform $H_{DT}^{(2)}$ to a 2-local classical Hamiltonian. By working out the details (see \Cref{par:appendix_classical_ham}), it turns out that in this 2-local classical Hamiltonian we obtain three distinct ``types'' of terms:
\begin{itemize}
    \item $h_\text{vert}$, 2-local terms corresponding to black terms acting on vertically arranged white terms (e.g. between $p_1$ and $p_2$),
    \item $h_\text{horiz}$, 2-local terms corresponding to black terms acting on horizontally arranged white terms (e.g. between $p_1$ and $p_3$), and
    \item $h_w$, 1-local terms corresponding to the white terms.
\end{itemize}
The final classical Hamiltonian is then
\[
H^{(2c)}_{DT} = \sum_{\text{vertical} p_i, p_j} (h_\text{vert})_{p_i,p_j} + \sum_{\text{horizontal} p_i,p_j} (h_\text{horiz})_{p_i,p_j} + \sum_p (h_w)_{p}.
\]

\paragraph{Preparation of Quantum Gibbs State.}\label{par:intro_gibbs_prep} 
So far, we've removed terms from $H_{DT}$ to obtain $H^{(2)}_{DT}$, then argued that we can view this as a classical
Hamiltonian $H^{(2c)}_{DT}$. Assume we are able to perform classical Gibbs
sampling at a given temperature on $H^{(2c)}_{DT}$. To obtain a sampler for our original Hamiltonian, we need to reverse
each step of the reduction. First, since the transformation from $H_{DT}^{(2)}$ to $ H_{DT}^{(2c)}$ is via a low-depth
quantum circuit, we can easily obtain a Gibbs state of $H_{DT}^{(2)}$ from the Gibbs state on $H_{DT}^{(2c)}$. The
primary challenge is to correct for the terms we have removed to make $H_{DT}$ $2$-local. 

Suppose we want to correct for a remove white term $p \in \mc W$. We first measure the current state $\psi$ with respect
to the term $p$. If we only need to obtain some state with the correct eigenvalues, whenever we obtain an incorrect outcome,
we could simply perform the correction operator $L_p$ depicted in \Cref{fig:intro_fig_correct}. To obtain $L_p$, we find
a path from a corner of $p$ to the boundary of the lattice, and apply a Pauli $X$ on each qubit along the path. However,
this does not immediately work when we are trying to sample from the Gibbs distribution.

Denote $\Pi^p_{+c_p} \ket  \psi$ be state if we get measurement outcome $+c_p$ when measuring $p$. Similarly for  $\Pi^p_{-c_p} \ket  \psi$. There are two challenges in preparing the Gibbs state. First, we need to maintain the proper distribution over $\Pi^p_{+c_p} \ket \psi$ and $\Pi^p_{-c_p} \ket \psi$. Second, applying the correction $L_p$ \emph{after measuring} $\Pi^p_{+c_p}$ may not yield $\Pi^p_{-c_p} \ket{\psi}$, i.e.
\begin{equation}
    \label{eq:eigenstate_not_equal}
	 L_{p} \Pi^p_{+c_p} \ket{\phi(\bm{y})} \not\propto \Pi^p_{-c_p}  \ket{\phi(\bm{y})} .
\end{equation}
Nonetheless, we show that this can be done via an \emph{oblivious randomized correction} technique. That is, based on the measurement outcome, we apply the correction operation $L_p$ with some probability $\mu$. At a high level, the correctness of the idea comes from the symmetry of the eigenspace; despite \Cref{eq:eigenstate_not_equal}, we \emph{do have} that 
\[
    L_p \Pi^p_{+c_p} \Pi_{\lambda(\psi)} \Pi^p_{+c_p} L_p = \Pi^p_{-c_p} \Pi_{\lambda(\psi)} \Pi^p_{-c_p}
\]
where $\Pi_{\lambda(\psi)}$ is an eigenspace of the \emph{non-removed} operators corresponding to $\ket \psi$. We will
leverage this fact by applying a correction uniformly across this eigenspace, and the correction probability $\mu$ will depend only $\Pi_{\lambda(\psi)}$ rather than $\ket \psi$ itself.

\subsection{Conclusion and future work}

In this manuscript, we give a reduction from Gibbs state preparation for various families of CLHs to the task of Gibbs
sampling for classical Hamiltonians. In particular, based on the Structure Lemma we show that there is a Gibbs sampling
reduction from 2-local qudit CLHs to 2-local qudit classical Hamiltonians. Based on the symmetry in qubit CLH and the
idea of oblivious randomized correction we give a Gibbs sampling reduction from various 2D qubit CLH to qudit classical
Hamiltonians. This approach yields a Gibbs sampler based on techniques very different than those traditionally used,
such as analyzing the Davies generator. We also demonstrate that combined with existing fast mixing results for
classical Hamiltonians, our Gibbs sampler matches the performance of state-of-the-art results in
\cite{capel2020modified,bardet2023rapid,kastoryano2016quantum}. It will be interesting to apply our framework to more examples.

A natural direction to explore is whether our reduction can be generalized to other CLHs, especially those for which the
complexity of proving ground energy is in $\NP$, such as the factorized qudit CLH on 2D
lattice~\cite{irani2023commuting}, the factorized CLH on any geometry~\cite{bravyi2003commutative} and the qutrit CLH on
2D~\cite{irani2023commuting}.

Our work also gives an interesting characterization for the complexity of quantum approximate counting for specific
CLHs. It is well-known that classical approximate counting is in $\BPP^{\NP}$~\cite{stockmeyer1983complexity}. Due to
the connection between quantum approximate counting and the Gibbs state preparation~\cite{bravyi2021complexity}, our
work shows that quantum approximate counting for various CLHs is contained in the complexity class $\BQP^{\textsf{CS}}$,
where $\textsf{CS}$ is an oracle which can perform arbitrary classical Gibbs sampling. It would be interesting to
explore whether there exist other families of quantum Hamiltonians where one can also upper bound the complexity of
quantum approximate counting by $\BQP^{\cO}$ for some oracle $\cO$ which is weaker than quantum approximate counting.

\section{Preliminary}\label{sec:pre}
\subsection{Notation} 
For two operators $h$ and $h'$, we use $[h,h']$ to denote the commutator $hh'-h'h$. We say that
 $h$ and $h'$ commutes if $[h,h']=0$. Given two $n$-qubit quantum states $\rho$ and $\sigma$, we use
 $\|\rho-\sigma\|_1:= \frac{1}{2}\text{tr}(|\rho-\sigma|)$ to denote their trace distance. Given two probability distributions $\cD_1$ and $\cD_2$ over  $\{0,1\}^n$,  we use $\|\cD_1-\cD_2\|_1$ to denote the total variation distance, that is $\|\cD_1-\cD_2\|_1 = \frac{1}{2}\sum_{x\in\{0,1\}^n} |\cD_1(x)-\cD_2(x)|$, where $\cD_i(x)$ is the probability of sampling $x$ in distribution $\cD_i$.
Given a graph $G=(V,E)$, for any vertex $v\in V$, we use $N(v)$ to denote the set of vertices which are adjacent to $v$ (excluding $v$). For $v,w\in V$, we use $\{v,w\}$ and $\langle v, w \rangle$ for unordered set and ordered set respectively. For a positive integer $m\in \bN$, we use 
$[m]$ to denote $\{1,2...,m\}.$

\subsection{Formal problem definitions}
\label{sec:def} 
 
\noindent\textbf{$k$-local Hamiltonians.}
We say an $n$-qudit Hermitian operator $H$ is a $k$-local Hamiltonian, if 
 $H=\sum_{i=1}^{m} h_i$ for $m=poly(n)$, and
 each $h_i$ only acts non-trivially on at most $k$ qudits. We allow different qudits to have different dimensions. 
 
For the special case of    $k=2$, one can  define $H$ via a graph $G=(V,E)$. That is, on each vertex there is a qudit, and on each edge $\{v,w\}$ there is a Hermitian term $h^{vw}$, such that
$$H_G=\sum_{\{v,w\}\in E} h^{vw}.$$
\vspace{0.5em}

\begin{figure}[h!]
    \centering
    \begin{subfigure}[t]{0.27\textwidth}
        \centering
        \begin{tikzpicture}[
            scale=0.75,
            every node/.style={transform shape},
            qubitlabel/.style={inner sep=0, draw, circle, minimum width=0.3cm, fill=white},    
        ]
            \foreach \x in {0,1,2} {
                \foreach \y in {0,1,2} {
                    \ifthenelse{\equal{\intcalcMod{\x+\y}{2}}{0}}{
                        \draw[fill=black!30] (\x,\y) rectangle ++(1,1);
                    }{
                        \draw (\x,\y) rectangle ++(1,1); 
                    }
                }
            }
            \node[qubitlabel] at (2,2) {\scriptsize $q$};
            \node[draw=none,fill=white, rounded corners, rectangle, inner sep=1mm] at (1.5,1.5) {\large $P$};
        \end{tikzpicture}
        \caption{2D Hamiltonian}
        \label{fig:2d_ham}
    \end{subfigure}
    \qquad
    \begin{subfigure}[t]{0.27\textwidth}
        \centering
        \begin{tikzpicture}[
            scale=0.75,
            every node/.style={transform shape},
            qubitlabel/.style={inner sep=0, draw, circle, minimum width=0.3cm, fill=white},    
        ]
            \foreach \x in {0,1,2} {
                \foreach \y in {0,1,2} {
                    \ifthenelse{\equal{\intcalcMod{\x+\y}{2}}{0}}{
                        \draw[fill=black!30] (\x,\y) rectangle ++(1,1) node[fitting node] (n\x\y) {};
                        \node[draw=none,fill=black!10, rounded corners=1, rectangle, inner sep=0.5mm] at ($ (n\x\y.north east) + (-0.25,-0.25) $) {\scriptsize $X$};
                        \node[draw=none,fill=black!10, rounded corners=1, rectangle, inner sep=0.5mm] at ($ (n\x\y.north west) + (0.25,-0.25) $) {\scriptsize $X$};
                        \node[draw=none,fill=black!10, rounded corners=1, rectangle, inner sep=0.5mm] at ($ (n\x\y.south east) + (-0.25,0.25) $) {\scriptsize $X$};
                        \node[draw=none,fill=black!10, rounded corners=1, rectangle, inner sep=0.5mm] at ($ (n\x\y.south west) + (0.25,0.25) $) {\scriptsize $X$};
                    }{
                        \draw (\x,\y) rectangle ++(1,1) node[fitting node] (m\x\y) {}; 
                        \node at ($ (m\x\y.north east) + (-0.25,-0.25) $) {\scriptsize $Z$};
                        \node at ($ (m\x\y.north west) + (0.25,-0.25) $) {\scriptsize $Z$};
                        \node at ($ (m\x\y.south east) + (-0.25,0.25) $) {\scriptsize $Z$};
                        \node at ($ (m\x\y.south west) + (0.25,0.25) $) {\scriptsize $Z$};
                    }
                }
            }
            
        \end{tikzpicture}
        \caption{Toric code}
        \label{fig:toric_code}
    \end{subfigure}
    \qquad
    \begin{subfigure}[t]{0.27\textwidth}
        \centering
        \begin{tikzpicture}[
            scale=0.75,
            every node/.style={transform shape},
            qubitlabel/.style={inner sep=0, draw, circle, minimum width=0.3cm, fill=white},    
        ]
            \foreach \x in {0,...,3} {
                \foreach \y in {0,...,3} {
                    \draw[color=black!70] (0,\y) -- (3,\y);
                    \draw[color=black!70] (\x,0) -- (\x,3);
                }
            }
            \draw (1,1) rectangle ++(1,1) node[fitting node] (c) {}; 
            \node[draw, circle, minimum width=0.2cm, fill, inner sep=0, label={[label distance=-0.1cm]below right:$Z$}] at (c.north west) {};
            \node[draw, circle, minimum width=0.2cm, fill, inner sep=0, label={[label distance=-0.1cm]above right:$Z$}] at (c.south west) {};
        \end{tikzpicture}
        \caption{2D Ising Model}
        \label{fig:ising_model}
    \end{subfigure}
    \caption{}
\end{figure}

\paragraph{2D Hamiltonians.}
Consider a 2D  lattice $G=(V,E)$ as in \Cref{fig:2d_ham}, where qudits are placed on vertices. As above, we can define a $2$-local Hamiltonian on $G$ via
$H_G = \sum_{\{v,w\} \in E} h^{v,w}$. We can also define a $4$-local Hamiltonian over the 2D lattice by associating a
Hermitian term to each plaquette $P$. We use  $v\in P$ to
denote  that vertex $v$ is in the  plaquette $P$. With some abuse of notations, we also use $v$ and $P$ to denote the
corresponding qubit and the Hermitian term. The 2D, 4-local Hamiltonian on the lattice is given by
\[H = \sum_P P\,.\]

  For many of the proofs in this work, it will useful to partition the plaquette terms $\{P\}_P$ into a set of ``black'' terms $\cB$
  and ``white'' terms $\cW$ by viewing the 2D lattice as a chess board, as in \Cref{fig:2d_ham}. 
  

There is another natural notion of a 2D local Hamiltonian where the qudits are placed on the edges and Hermitian terms are corresponding to plaquettes and stars; this is the version which is primarily considered in \cite{aharonov2018complexity}. However, the two settings (where qudits are on vertices or are on edges) are equivalent when the underlying graph is the 2D square lattice (See Appendix C in \cite{irani2023commuting}).

\paragraph{Commuting and classical Hamiltonians.} We say a  $k$-local Hamiltonian $H = \sum_{i=1}^m h_i$ is a
\emph{commuting local Hamiltonian} (CLH) if $[h_i,h_j]=0,\forall i,j$. Whenever we have a Hamiltonian $H=\sum_P P$
defined over the plaquettes of a 2D lattice, we have $[p,p']=0$ $\forall p,p'$. 
We say a  $k$-local Hamiltonian $H = \sum_{i=1}^m h_i$  is \textit{classical}, if each $h_i$ is diagonalized in the computational basis.  

\paragraph{Defected Toric code and 2D Ising model.} Here we give two examples of qubit CLHs on 2D.  For a vertex $v$ in plaquette $p$, we use $Z^p_v,X^p_v$ to denote the Pauli Z and Pauli X operator on the qubit $v$. When $v$ uniquely identifies a vertex we abbreviate $Z^p_v$ as $Z_v$ and similarly for other Pauli operators. For a plaquette term $p$, we define $Z^p:= \otimes_{v\in p} Z_v$ and $X^p$ similarly. 

 As shown in \Cref{fig:toric_code},
 the \textit{defected Toric code} is defined as 
 $$H=\sum_{p\in \cW} c_p Z^p + \sum_{p\in \cB} c_p X^p, $$ where $c_p\in \bR$ can be arbitrary. 
  The \textit{standard Toric code} is a special case when all $c_p=-1$. 
 
  Denote the 2D lattice as $G=(V,E)$. As in \Cref{fig:ising_model},
the  \textit{ferromagnetic 2D Ising model} is a 2-local Hamiltonian 
$$H=\sum_{\{u,v\}\in E} Z_u\otimes Z_v.$$

The 2D Ising model is a classical Hamiltonian whose eigenstates are all computational basis. The defected Toric code is not a classical Hamiltonian. The ground state of the standard Toric code is highly entangled and cannot be prepared by any constant depth quantum circuit.

\paragraph{Gibbs state  preparation.} Given a $k$-local Hamiltonian $H=\sum_i h_i$,  an inverse temperature
$\beta<+\infty$, the \textit{Gibbs state} with respect to $\left(H,\beta\right)$ is defined as 
\begin{align}
    \rho( H,\beta)=\frac 1 {\text{tr}(\exp(-\beta H))}\exp(-\beta H).
\end{align}
 Given $\epsilon> 0$, we say an algorithm $\cA$  prepares $\rho( H,\beta)$ within precision $\epsilon$, if $\cA$ outputs a state $\rho$ such that
\begin{align}
	\|\rho-\rho( H,\beta)\|_1 \leq \epsilon. \label{eq:rho}
\end{align}
When $H$ is classical,  the \textit{classical Gibbs distribution} with respect to $\left(H,\beta\right)$ is denoted as
$\cD_{\beta H}$, which samples a classical string $x\in\{0,1\}^n$ with probability $\exp(-\beta \langle x|H|x\rangle)/
\text{tr}(\exp(-\beta H))$. We say an algorithm $\cA$ performs classical Gibbs sampling $\cD_{\beta H}$ with precision $\epsilon$, if $\cA$ outputs a distribution $\cD$ such that   
\begin{align}
	\|\cD-\cD_{\beta H}\|_1\leq \epsilon. \label{eq:D}
\end{align}
Note that Eq.~(\ref{eq:D}) is equivalent to Eq.~(\ref{eq:rho}) when $\rho$ and $\rho(H, \beta)$ are diagonal matrices.

\section{Reduction for 2-local qudit CLHs}\label{sec:2local}

In this section, we prepare the Gibbs state  for qudit  2-local  CLHs. In particular,  in Section \ref{sec:2local_general} we will prove the Gibbs sampling reduction for general 2-local qudit CLHs in  Theorem \ref{thm:2local}.  Then in Section \ref{sec:2local_specific} we will give several examples, whose proofs are put into Appendix \ref{appendix:specific}.

\subsection{General case}\label{sec:2local_general}
Recall that a  2-local CLHs  $H_G^{(2)}$ is defined on a graph $G=(V,E)$, where 
$$H_G^{(2)}=\sum_{\{v,w\}\in E} h^{vw} .$$ Here $\{h^{vw}\}_{\{v,w\}\in E}$ are Hermitian terms and commute with each
other. The superscript $(2)$ is to emphasize that the Hamiltonian is  2-local. 
In this section, we assume on each vertex there is a \textit{qudit} rather than a qubit, and we allow $G$ to be an
arbitrary graph rather than just a 2D lattice. 

The Gibbs sampling reduction for 2-local CLHs comes from  the  Structure Lemma, which was originally developed by
\cite{bravyi2003commutative} to study the computational complexity of commuting Hamiltonians.  A constructive proof of
the Structure Lemma can be found in Section  7.3 of \cite{gharibian2015quantum}.
Intuitively, the Structure Lemma says that one can decouple all commuting 2-local terms. This will allow us to identify
eigenstates of a 2-local CLH $H_G^{(2)}$ with the computational basis of a classical Hamiltonian $H_G^{(2c)}$ defined
later. In addition, each such eigenstate can be prepared by a constant depth quantum circuit. Thus, to prepare the Gibbs
state of $H_G^{(2)}$, it suffices to first do classical Gibbs sampling for $H_G^{(2c)}$, yielding a distribution over
computational bassi states, then prepare the corresponding eigenstate of $H_G^{(2)}$ indexed by a sampled basis state
via a constant depth quantum circuit.

We first give the formal statement of the Structure Lemma. 

\begin{lemma}[Rephrasing of the Structure Lemma~\cite{bravyi2003commutative}]\label{lem:Cstructure}
 Consider a vertex $v\in V$ and denote the Hilbert space of the qudit on $v$ as $\cH^v$.  Consider the commuting Hermitian terms $\{h^{vw}\}_{w\in N(v)}$. There exists a direct sum decomposition of $\cH^v$, 
 \begin{align}	
 \cH^v=\bigoplus_{j_v=1}^{J_v} \cH^v_{j_v},
\end{align}
such that $\forall j_v$, for any $w\in N(v)$, the term $h^{vw}$ keeps the subspace $\cH^v_{j_v}\otimes\cH^w$ invariant. Furthermore, each $\cH^{v}_{j_v}$ has  a tensor product factorization: 
\begin{align}
	\cH_{j_v}^v = \bigotimes_{w\in N(v)} \cH_{j_v}^{\langle v,w \rangle}, \label{eq:24}
\end{align}
such that for all neighbors $w\in N(v)$, the term $h^{vw}|_{j_v}$  (the restriction of $h^{vw}$ onto $\cH^v_{j_v}\otimes \cH^w$) 
 acts non-trivially only on $\cH_{j_v}^{\langle v,w \rangle}\otimes \cH^w$, i.e.
  \begin{align}
  	h^{vw}|_{j_v} \subseteq 	\left(   \bigotimes_{u\in N(v)/\{w\}} \cI\left(\cH_{j_v}^{\langle v,u\rangle}\right)  \right)
    \otimes  \cL \left(\cH_{j_v}^{\langle v,w\rangle}\otimes \cH^{w}\right),
  \end{align}
  where $\cI(\cH)$ is the identity operator on space $\cH$, and $\cL(\cH)$ is the set of all linear operators on $\cH$. 
\end{lemma}

The Structure Lemma can be understood via \Cref{fig:2localTP}. We can understand \Cref{eq:24} as the following: by
choosing a proper local basis for the Hilbert space $\cH_{j_v}^v$, it is equivalent to the Hilbert space of $|N(v)|$ distinct new qudits. 

\begin{figure}[h!]
    \centering
    \scalebox{1.3}{
    \begin{tikzpicture}[
        every node/.append style={inner sep=0},
    ]
        \node[draw, circle, minimum width=1cm, label={[label distance=0.1cm]above:$j_v$}] (v) at (0,0) {$v$};
        \node[draw, circle, minimum width=1cm, label={[label distance=0.1cm]above:$j_w$}] (w) at (2,0) {$w$};
        \node[circle, fill, minimum width=0.15cm] (vw) at ($ (v.east) + (-0.2,0) $) {};
        \node[circle, fill, minimum width=0.15cm] at ($ (v.north west) + (0.2,-0.2) $) {};
        \node[circle, fill, minimum width=0.15cm] at ($ (v.south) + (0,0.2) $) {};
        \node[circle, fill, minimum width=0.15cm] (wv) at ($ (w.west) + (0.2,0) $) {};
        \node[circle, fill, minimum width=0.15cm] at ($ (w.north) + (0,-0.2) $) {};
        \node[circle, fill, minimum width=0.15cm] at ($ (w.south) + (0,0.2) $) {};
        \draw (vw) -- (wv);
        \node (vwlabel) at ($ (vw) + (0.1,-1) $) {\footnotesize $\mathcal H^{\langle v, w\rangle}_{j_v}$};
        \node at ($ (vwlabel.south) + (0,-0.4)$) {\scriptsize $b^{\langle v,w \rangle}_{j_v}$};
        \draw[->, line width=0.2mm] (vwlabel) -- ($ (vw) + (0,-0.2) $);
        \node (wvlabel) at ($ (wv) + (-0.1,-1) $) {\footnotesize $\mathcal H^{\langle w, v\rangle}_{j_w}$};
        \node at ($ (wvlabel.south) + (0,-0.4)$) {\scriptsize $b^{\langle w,v \rangle}_{j_w}$};
        \draw[->, line width=0.2mm] (wvlabel) -- ($ (wv) + (0,-0.2) $);
        
    \end{tikzpicture}
    }
    \captionsetup{width=0.9\textwidth}
    \caption{ The figure contains two vertex $v$ (on the left) and $w$ (on the right). The Structure Lemma says that after choosing subspace $j_v$ for $\cH^{v}$ and $j_w$ for $\cH^{w}$, all terms  are decoupled, that is they act on different qudits. More specifically,  
    the qudits on $v$ (the big \scalebox{1.5}{$\circ$}) can be interpreted as the tensor product of several qudits of
    smaller dimension (the small $\bullet$). Similarly for the qudit on $w$. The term $h^{vw}$ only acts on the qudits
    correspond to $\{v,w\}$, that is $\cH_{j_v}^{\langle v,w \rangle}$ and  $\cH_{j_w}^{\langle w,v\rangle}$, which are associated with edge $\{v,w\}$ and are not  touched by any other terms. $b^{\langle v,w \rangle}_{j_v}$ and $b^{\langle w,v \rangle}_{j_w}$ are notations for the computational basis states for $\cH^{\langle v,w \rangle}_{j_v}$ and $\cH^{\langle w,v \rangle}_{j_w}$ respectively.
  }
    \label{fig:2localTP}
\end{figure}
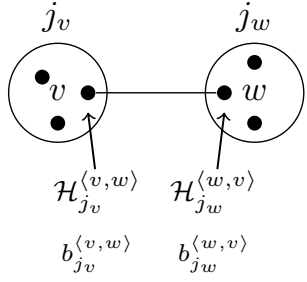

 If for every qudit $v$, one applies 
 \Cref{lem:Cstructure} and chooses an index $j_v\in [J_v]$ and corresponding subspace $\cH^v_{j_v}$, this will decouple
 all terms in $H_G$. Each term $h^{vw}$ restricted to the subspaces $\cH^v_{j_v} \otimes \cH^v_{j_w}$ will act on
 distinct qudits and  
\[
h^{vw} \in \cL\left(\cH_{j_v}^{\langle v,w\rangle}\otimes \cH_{j_w}^{\langle w,v\rangle}\right).
\]

To define the corresponding classical Hamiltonian $H_G^{(2c)}$, we use the indices $\{j_v\}_{v \in V}$ to index the eigenstates of $H_G^{(2)}$. Denote
\[
 h^{vw}|_{j_vj_w}:=  \text{ restriction of $h^{vw}$ onto $\cH_{j_v}^{\langle v, w \rangle} \otimes \cH_{j_w}^{\langle w,v \rangle}.$}
\]
Note that eigenstates of $ h^{vw}|_{j_vj_w}$ might not be computational basis states (and in particular could be
entangled). Nonetheless, we have shown that under the restriction corresponding to $\{j_v\}_{v \in V}$, all terms $h^{vw}|_{j_v,
j_w}$ act on distinct qudits and we can use the computational basis to index the eigenstates. 
  As shown in Figure \ref{fig:2localTP}, let $D^{\langle v, w \rangle}_{j_v} :=  \dim (\cH^{\langle v, w
  \rangle}_{j_v})$ and write the basis of each subspace $\cH_{j_v}^{\langle v, w \rangle}$ as $\ket{b_{j_v}^{\langle v, w \rangle}}$
  where $b_{j_v}^{\langle v, w \rangle}$ ranges over $[D_{j_v}^{\langle v, w \rangle}]$. Then $\bigotimes_{w\in N(v)}
  \ket{b_{j_v}^{\langle v, w \rangle}}$ is a computational basis state of $\cH_{j_v}^{v}$. 
For each edge $\{v,w\}\in E$, the term $h^{vw}|_{j_vj_w}$ is Hermitian and thus can be diagonalized; the computational
basis states $\ket{\bb_{j_vj_w}^{vw}}$ are used to index the eigenstates. Thus, a basis for the full eigenspace of
$h^{v,w}|_{j_v,j_w}$ is given by
    \begin{equation}
        \ket{\bb_{j_vj_w}^{vw}} : = \ket{b_{j_v}^{\langle v, w \rangle},b_{j_w}^{\langle w,v \rangle}},\quad \bb^{v,w}_{j_v
        j_w} \in [D^{\langle v,w \rangle}_{j_v} \times D^{\langle w,v \rangle}_{j,w}]\,.
    \end{equation}
 
    Given an index $\bb^{v,w}_{j_v j_w}$, the corresponding eigenstate is denoted $\psi(\bb_{j_v j_w}^{vw})$ and the
    eigenvalue $\lambda(\bb_{j_v j_w}^{vw})$. The classical Hamiltonian is defined by substituting the eigenstate with its index, that is 
     \begin{align}
        \label{eq:classical_ham}
     	H_G^{(2c)} :=  \sum_{\{v,w\}\in E} \quad  \sum_{j_v,j_w} \quad \sum_{\bb_{j_vj_w}^{vw}}  	\lambda(\bb_{j_vj_w}^{vw}) \ket{\bb_{j_vj_w}^{vw}} \bra{\bb_{j_vj_w}^{vw}}.
     \end{align}  
     Following the usual convention, each term in the summand is implicitly padded with identities as necessary.

 In this way, the eigenstates of $H^{(2c)}_G$ are given by specifying an index $j_v$ for each vertex $v \in V$, then a
 basis state $b^{\langle v, w \rangle}_{j_v j_w}$ for each of the decoupled Hilbert spaces $\cH^{\langle v,
 w\rangle}_{j_v} \otimes \cH^{\langle w,v\rangle}_{j_w}$
     \begin{align}
             &\bj:=\{j_v\}_v,\label{eq:14}\\
             & \bb_\bj:=\{b_{\bj_v, \bj_w}^{\langle v, w \rangle}\}_{\langle v, w \rangle}.
    \end{align}
Thus, by construction, we have the following.
\begin{lemma}
$H_G^{(2c)}$ is 2-local classical Hamiltonian on the graph $G$.	
\end{lemma}
We can also easily map eigenstates of $H^{(2c)}_G$ to eigenstates of $H^{(2)}_G$ via
    \begin{align}
     	&\lambda(\bb_\bj): = \sum_{\{v,w\}\in E} \lambda(\bb_{j_vj_w}^{vw})\\
        & \ket{\psi(\bb_\bj)} : = \bigotimes_{\{v,w\}\in E } \ket{\psi(\bb_{j_vj_w}^{vw})}\,,\label{eq:17}
     \end{align}
     and any classical Gibbs sampling procedure for $H^{(2c)}_G$ yields a Gibbs sampler for $H^{(2)}_G$.
\begin{theorem}\label{thm:2local}
For any inverse temperature $\beta$, if one can do classical Gibbs sampling w.r.t $\left(H_G^{(2c)}, \beta\right)$ within precision $\epsilon$ in classical time $T$, then one can prepare  the quantum Gibbs state  w.r.t. $\left(H_G^{(2)},\beta\right)$ within  precision $\epsilon$ in quantum time $T+ \mc O(m)$, where $m$ is the the number of edges in  graph $G$, by firstly using the  classical Gibbs sampling w.r.t. $\left(H_G^{(2c)},\beta\right)$ to sample the index $\bb_\bj$, then prepare the product state $ \ket{\psi(\bb_\bj)}$ in time $\mc O(m)$.
\end{theorem}
\begin{proof}
	It suffices to notice that by construction, $\bb_\bj$ indexes the  eigenvector of $H_G^{(2)}$ of eigenvalue $\lambda(\bb_\bj)$, that is $ \ket{\psi(\bb_\bj)}$. 
\end{proof}

\subsection{Examples}
\label{sec:2local_specific}

In this section, we write down the Gibbs sampling reduction for some specific Hamiltonians as illustrative examples. All
the proofs are deferred to \Cref{appendix:specific}.

We first consider an $n$-qudit  Hamiltonian on a 1D chain 
\begin{align}
    H=\sum_{i} h_i,\label{eq:hi}
\end{align} 
we say that $H$ is $r$-range if  
$h_i$ only acts non-trivially on qudits $i,i+1,..,i+r-1$. For simplicity we assume $n$ is an integer multiple of $r$ and
the qudit dimension $d$ is a power of $2$. We say that $H$ is finite-range if $r$ is a constant, and $H$ is translation-invariant if all the terms $h_i$ are the same. 

By coarse-graining $H$, we can always assume $H$ is 2-local:
group each consecutive set of $r$ qudits as a new qudit so that each $h_i$ acts non-trivially on at most two (grouped)
qudits. For each pair of new qudits $\{j, j+1\}$, we associate the new term $H_{j,j+1}$, which is a sum of all terms from $H$
acting on the corresponding qudits. Thus
$H$ can be viewed as a 2-local qudit Hamiltonian on 1D written as $H=\sum_j H_{j,j+1}$. Note the terms $H_{j,j+1}$ can
also be made translation-invariant.

\begin{lemma}\label{lem:H1D}
	Consider a finite-range translation-invariant 
 qudit  CLH on 1D chain,  denoted as $H_{1D}$. Then the corresponding classical Hamiltonian $H^{(c)}_{1D}$ can be made as 1D finite-range translation-invariant Ising model.
 
 Combined with the rapid mixing Gibbs sampler for   1D finite-range, translation-invariant Ising model for any constant
 inverse temperature $\beta$~\cite{guionnet2003lectures,holley1985rapid,holley1989uniform} which performs classical Gibbs
 sampling to precision $\epsilon$ in time $T(\beta,\epsilon)$, \Cref{thm:2local} implies that
	 one can prepare the Gibbs state on $(H_{1D},\beta)$ to precision $\epsilon$ in quantum time $T(\beta,\epsilon)+ \mc O(n)$.
 \end{lemma}

We also give another example of a Hamiltonian on a 2D lattice. Recall our first definition of a 2-local Hamiltonian on
2D (see \Cref{sec:def}) 
$$H_{2D}=\sum_{\{v,w\}\in E} h^{vw}.$$
As usual, $H_{2D}$ is translation-invariant if all terms $h^{vw}$ are the same. 
We say that a 2D lattice has periodic boundary condition if it can be embedded onto torus; we will assume a periodic
boundary for simplicity.

\begin{lemma}\label{lem:H2D} 
Consider a  translation-invariant, 2-local 2-dimensional \emph{qubit} CLH $H_{2D}=\sum_{\{v,w\}\in E} h^{vw}$ with a periodic boundary condition. 
 Then the classical Hamiltonian $H^{(c)}_{2D}$ can viewed as a 2D Ising model under a magnetic field.
 
Set the precision to be $1/poly(n)$. If the corresponding 2D Ising model $H^{(c)}_{2D}$ is ferromagnetic with a
consistent field, then there exists poly-time mixing Gibbs sampler using the Swendsen-Wang
dynamics for any constant temperature  (as in \cite{feng2023swendsen}). Via our quantum-to-classical Gibbs sampling reduction, we can  prepare the Gibbs state for the corresponding CLH $H_{2D}$ in quantum polynomial time.
 \end{lemma}

\section{Reduction for 2D 4-local qubit CLH without classical qubits}\label{sec:qubit2D}

In this section we describe how to prepare the Gibbs state of 2D qubit CLHs. In \Cref{sec:carno} we first review the
canonical form of the 2D qubit CLH as developed in \cite{aharonov2018complexity}, who establishes a connection between
2D qubit CLHs and the defected Toric code. Based on this connection and an observation on the symmetry of the ground
space, we use an oblivious randomized correction technique to generalize the Gibbs state preparation algorithm for the
defected Toric code presented in \Cref{sec:TC} to prepare the Gibbs state for the more general family of 2D qubit CLHs
without classical qubits.\footnote{The formal definition of classical qubits is given in Definition
\ref{def:classical}.}

\subsection{A canonical form}\label{sec:carno}
This section is primarily a review of the results in \cite{aharonov2018complexity}. In that work, the authors prove that
2D qubit CLH\footnote{The definition of a ``2D qubit CLH'' in \cite{aharonov2018complexity} is slightly different;
qubits are put on the edges of 2D lattice while we put qubits on the vertices. However, the two settings are equivalent,
as explained in Appendix C of \cite{irani2023commuting}.} without classical qubits (which we will define shortly) is in some sense equivalent to the defected Toric code. 
\cite{aharonov2018complexity} used this connection to show that one can prepare the ground state of 2D qubit CLHs
similar to the way ground states of the defected Toric code are prepared; this is via the measure and correct approach
mentioned in \Cref{sec:TCnn}.

We summarize necessary definitions and theorems which will be used in later sections. Recall that a 2D qubit CLH is
defined as $H=\sum_{p\in P} p$, where $P$ is the set of plaquettes of the lattice.
\begin{definition}[Boundary and interior]\label{def:boundary}
 A qubit is in the boundary of the Hamiltonian, if it is acted trivially by at least one of the four adjacent
 plaquette terms. All other qubits are said to be in the interior. A plaquette term $p$ which acts only on interior qubits is said to be in the interior of the Hamiltonian.
\end{definition}

\begin{definition}[Classical qubit]\label{def:classical}
A qubit is classical if its Hilbert space can be decomposed into a direct sum of 1-dimensional subspace, which are
invariant under all terms $\{p\}_{p \in P}$. We say that there is no classical qubit if and onlf if all qubits in the
system are not classical.
\end{definition}
In other words a qubit $q$ is classical if under some basis for the qubit, all terms look like $\ketbra 0 0_q \otimes
... + \ket 1 1_q \otimes ...$.

\begin{definition}[Access to boundary]\label{def:access}
We say that a plaquette term $p$ has access to the boundary if there exists a path (a sequence of adjacent vertices)
$\gamma_p$ starting from a vertex of $p$ and ending at a vertex correspoding to a boundary 
	qubit. Morever, there should be some choice of local unitary $U_v$ on each vertex of the path such that  the operator 
 $$L_p:=\otimes_{v\in \gamma_p} U_vX_vU_v^\dagger$$ anti-commutes with $p$, and commute with all other terms. Note that
 by construction, $L_p^2 =I$.
\end{definition}

\begin{lemma}[Interior term]\label{lem:Z}
Suppose there is no classical qubit, and a term $p$ is in the interior of the Hamiltonian. Then by choosing a proper basis for each qubit, we have
\begin{align}
	p = a_p \cI +c_p Z^p\label{eq:Z},
	\end{align}
with $a_p,c_p\in \bR, c_p\neq 0.$ Note that replacing $p$ with $p- a_p \cI$ does not change the Gibbs state and thus we
may assume $a_p=0$.
\end{lemma}

Even when all terms are in the interior, Lemma \ref{lem:Z} does not imply all the terms should be a tensor product of Pauli $Z$.
To be written in the form of Eq.~(\ref{eq:Z}), adjacent plaquettes  may need different choices of basis. An example is the Toric code, where plaquettes are alternate $X^{\otimes 4}$ and $Z^{\otimes 4}$. 

Recall the chessboard partition of the 2D lattice into black plaquettes $\cB$ and white plaquettes $\cW$.
\begin{lemma}[Rephrased from Theorem 5.3 and Lemma 6.2~\cite{aharonov2018complexity}]\label{lem:car}
    Consider a 2D qubit CLH $H=\sum_{p \in P} p$. Suppose there are no classical qubits. Then
\begin{itemize}
	\item[(i)] If the set of boundary qubits is not empty, then for any adjacent plaquette terms $p\in \cB, \hat{p}\in
        \cW$ such that $p$ and $\hat p$ are both in the interior, either $p$ or $\hat{p}$ has access to the boundary.
		\item[(ii)] If there are no boundary qubits, then $H$ is equivalent to the defected Toric code on a closed 2D surface without boundary: by a choosing
            proper basis for each qubit, we have $\forall p\in \cB$, $p$ is of the form $a_p \cI +  c_p X^p$  with $c_p\neq 0$ and
            $\forall p\in \cW$, $p$ is of form $a_p \cI +  c_p Z^p$ with $c_p\neq 0$. 
\end{itemize}
\end{lemma}

In \cite{aharonov2018complexity} the authors use Lemma \ref{lem:car} to reduce the task of preparing the ground state of
any (4-local) 2D qubit CLH to the task of preparing the ground state of
a 2-local qudit CLH. This reduction is characterized by the following Corollary.

\begin{corollary}[\cite{aharonov2018complexity}] \label{cor:PA}Consider a 2D qubit CLH $H=\sum_{p \in P} p$.
Suppose there are no classical qubits, and the set of boundary qubits is not empty. Then there exists a partition of all
the terms $\{p\}_p$ as $\cP$ (punctured terms) and $\cR$ (terms with access to the boundary) such that
\begin{itemize}
	\item[(1)] After grouping some qubits into qudits, $H_\cP:=\sum_{p\in\cP} p$ can be viewed as a  2-local \emph{qudit} CLH on a constant-degree planar graph 
 $G=(V,E)$. When viewed as a 2-local Hamiltonian, we also write $H_\cP$ as $H_G^{(2)}=\sum_{\{v,w\}\in E} h^{vw}.$ 
	\item[(2)] All terms in $\cR$ are in the interior of the Hamiltonian and have access to the boundary.
\end{itemize}
\end{corollary}
For instance, in \Cref{fig:intro_defected} from the technical overview, $\cP$ is the set of all black terms $\cB$ and non-removed white terms $\cO$. Then, $H_{\cP}$ is exactly the 2-local Hamiltonian $H_{DT}^{(2)}$.

\subsection{Reduction to a classical Hamiltonian }
\label{sec:reduction_to_classial}

In this section, we describe how to reduce the task of Gibbs state  preparation for qubit CLH without classical qubits
to the task of classical Gibbs sampling. An explicit and canonical example is the Gibbs state preparation for the
defected Toric code, as described in \Cref{appendix:DTC} and \Cref{appendix:non_punctured}, with respect to the punctured defected Toric code (e.g. embedded on a planar lattice)  and on a torus, respectively. For general qubit CLHs without classical qubits, our
result is summarized in the following theorem.

\begin{theorem}\label{thm:4local2}
    Given an 2D $n$-qubit CLH $H=\sum_{p \in P} p$, suppose there are no classical qubits, and the set of boundary
 qubits is not empty. Let $\cP$, $\cR$, $H_\cP:=\sum_{p\in \cP} p$ and  $H_G^{(2)}$ be as defined in \Cref{cor:PA}. 
 
 Let $H_G^{(2c)}$ be the 2-local classical Hamiltonian derived from $H_G^{(2)}$ as in \Cref{sec:2local}. Then for any
 inverse temperature $\beta$ and precision $\epsilon$, if one can perform classical Gibbs sampling on $(
 H_G^{(2c)},\beta)$ to precision $\epsilon$ in classical time $T$, then one can prepare the Gibbs state on $(H,\beta)$ on a quantum computer in time $T + \mc O(n^2)$.

 On the other hand, if the set of boundary qubits is empty, then by \Cref{lem:car} item (ii) the Hamiltonian is equivalent to the defected
 Toric code on a closed 2D surface without boundary\footnote{That is, the defected Toric code  $H=\sum_{p\in \cW} c_p Z^p + \sum_{p\in \cB} c_p X^p$ is embedded on a closed 2D surface and the coefficient $c_p\neq 0$, for all $p$. }, and the Gibbs state can be prepared as described in \Cref{appendix:non_punctured} in time $\TCtime$.
\end{theorem}

We begin with some notation. Recall that the set of all plaquette terms $P$ is partition into black and white terms,
i.e. $P = \cB\cup\cW$. In \Cref{cor:PA} we defined $\cR$ as the set of terms which have access to the boundary. 
Let $ Q\subseteq \cB\cup\cW$ be an arbitrary subset of the plaquette terms. We write $\cR\setminus Q$ as the set
difference of $\cR$ and $Q$. For $p\in \cR$ define the correction operator $L_p$ as in \Cref{def:access}.

Define $\bl_Q:=\{\lambda_p\}_{p\in Q}$ to be a set of real values, where each $\lambda_p$ coresponds to an eigenvalue
of $p\in Q$. Let $\lambda(Q) :=\sum_{p\in Q} \lambda_p$. Recall that all plaquette terms are commuting. Thus, the
terms are simulataneously diagonalizable and the common eigenspace of each $p \in Q$ with eigenvalue
$\lambda_p$ is well defined; we denote this as $\cH^Q_{\bl_Q}$. Formally,
\begin{align}
	\cH^Q_{\bl_Q}:=	\{\ket{\phi} \, | \, p\ket{\phi} = \lambda_p \ket{\phi} , \forall p\in Q\}.
\end{align}
Let $\Pi^Q_{\bl_Q}$ be the projection onto $\cH^Q_{\bl_Q}$. Again by commutation, $\Pi^Q_{\bl_Q}$ is
equal to the product of the individual projectors:
\begin{align}
    \Pi^Q_{\bl_Q} = \prod_{p\in Q} \Pi^p_{\lambda_p}\,.\label{eq:proj}
\end{align}

Note that for a general 2D qubit CLH, a plaquette term $p$ can be any arbitrary 4-qubit operator. For example one can
set $H=p_0 + \cI$ where $p_0$ is an arbitrary
operator on one plaquette. Nonetheless, one can show that all terms in $\cR$ are in a sense quite regular.

\begin{lemma}\label{lem:2eig}
	Each $p\in \cR$ has exactly two eigenvalues $\pm c_p$.
\end{lemma}
\begin{proof}
	By definition of $\cR$, i.e. Corollary \ref{cor:PA} (2), all terms in $\cR$ are in the interior of the Hamiltonian. Then Lemma \ref{lem:2eig} is true by Lemma \ref{lem:Z}.
\end{proof}

Additionally, the eigenspaces corresponding to each eigenvalue of $p$ are symmetric. \Cref{lem:sym} is the key observation which leads to the oblivious randomized correction idea.
\begin{lemma}\label{lem:sym}
	For any subset $Q\subseteq \cB\cup\cW$ and for any $p\in \cR\backslash Q$, we have that 
	\begin{align}
        \label{eq:sym1}
		&L_p  \Pi^p_{+c_p}\Pi^{Q}_{\bl_Q} 	 \Pi^p_{+c_p} L_p^\dagger = \Pi^p_{-c_p}\Pi^{Q}_{\bl_Q} 	 \Pi^p_{-c_p}\\
        \label{eq:sym2}
  &L_p  \Pi^p_{-c_p}\Pi^{Q}_{\bl_Q} 	 \Pi^p_{-c_p} L_p^\dagger = \Pi^p_{+c_p}\Pi^{Q}_{\bl_Q} 	 \Pi^p_{+c_p}
	\end{align}
\end{lemma}
\begin{proof}
    We prove the first equality by moving $L_p$ on the very left of the LHS through $\Pi^p_{+c_p}$ and $\Pi^Q_{\bl_Q}$
    until we can cancel it with $L_p^\dagger$. The second equality follows from the first and the fact that $L_p^2 =  \cI
    $ and $L_p=L_p^\dagger$. By \Cref{lem:2eig}, we have 
\begin{align}
	\Pi^p_{\pm c_p} = \frac{1}{2c_p}\left(\pm p+c_p\cI\right).	
\end{align}
By \Cref{def:access}, we have $L_p$ anti-commutes with $p$. Thus we 
have
\begin{align}
    L_p  \Pi^p_{+c_p} & = L_p  \frac{1}{2c_p}\left(+ p+c_p\cI\right)\\
    &= \frac{1}{2c_p}\left(- p+c_p\cI\right) L_p\\
    &= \Pi^p_{-c_p} L_p,\label{eq:first_comm}
\end{align}
and we can rewrite $L_p  \Pi^p_{+c_p}\Pi^{Q}_{\bl_Q}\Pi^p_{+c_p} L_p^\dagger$ as $\Pi^p_{-c_p} L_p
\Pi^{Q}_{\bl_Q}\Pi^p_{+c_p} L_p^\dagger$. Next, by \Cref{def:access} we have that $L_p$ commutes with all terms in $Q$,
and thus $L_p$ also commutes with each eigenspace projectors $\Pi^{p'}_{\lambda_{p'}}$, $\forall p'\in Q$. Recalling the definition of
$\Pi^Q_{\bl_Q}$ in \Cref{eq:proj}, this means that $L_p$ commutes with $\Pi^{Q}_{\bl_Q}$, and we can move $L_p$ through
$\Pi^Q_{\bl_Q}$. To conclude, we once again use \Cref{eq:first_comm} and the fact that $L_p$ is unitary (i.e.,
$L_pL_p^\dagger = \cI$), obtaining
\begin{align}
	L_p  \Pi^p_{+c_p} \Pi^{Q}_{\bl_Q} 	 \Pi^p_{+c_p} L_p^\dagger = \Pi^p_{-c_p}\Pi^{Q}_{\bl_Q} 	 \Pi^p_{-c_p}\,,
\end{align}
as desired.
\end{proof}

A consequence of this lemma is that the the the eigenspace corresponding to $\bl_Q$ is balanced across any $p$'s $+c_p$
and $-c_p$ eigenspaces.

\begin{lemma}\label{lem:dim}
    $\mathrm{tr}\left( \Pi^p_{+c_p}\Pi^{Q}_{\bl_Q} 	 \Pi^p_{+c_p} \right) = \frac{1}{2} \mathrm{tr}\left( \Pi^{Q}_{\bl_Q}\right)$.
\end{lemma}
\begin{proof}
	Note that 
	\begin{align}
		\mathrm{tr}\left(  \Pi^p_{+c_p}\Pi^{Q}_{\bl_Q} 	 \Pi^p_{+c_p} \, + \,  \Pi^p_{-c_p}\Pi^{Q}_{\bl_Q}\Pi^p_{-c_p} \right) & = 	\text{tr}\left(  (\Pi^p_{+c_p})^2\Pi^{Q}_{\bl_Q} 	 \, + \,  (\Pi^p_{-c_p})^2\Pi^{Q}_{\bl_Q} \right) \\
		& = \text{tr} \left(\Pi^{Q}_{\bl_Q}\right)\, ,
		\end{align}
where the last equality comes is because $\Pi_{+c_p},\Pi_{-c_p}$ are projections, and $\Pi_{+c_p}+\Pi_{-c_p}=\cI$.

Since $L_p$ is a unitary, by \Cref{lem:sym} we have that 
\begin{align}
    \text{tr}(\Pi^p_{+c_p}\Pi^{Q}_{\bl_Q} 	 \Pi^p_{+c_p}) = \Pi^p_{-c_p}\Pi^{Q}_{\bl_Q} 	 \Pi^p_{-c_p}\,.
\end{align}
Thus $\text{tr}( \Pi^p_{+c_p}\Pi^{Q}_{\bl_Q} 	 \Pi^p_{+c_p}) = \frac{1}{2} \text{tr}( \Pi^{Q}_{\bl_Q})$.
\end{proof}

Underyling the \Cref{thm:4local2} is the following algorithm. We will prove \Cref{thm:4local2} by proving correctness via
\Cref{lem:reduction}. Recall that by \Cref{cor:PA} the Hamiltonian $H_G^{(2)}$ is 2-local after we group some qubits into qudits.
    Using the notation from \Cref{sec:2local}, $H_G^{(2)}$ and $H_G^{(2c)}$ denote the $2$-local CLH and the
    corresponding classical
    Hamiltonian, and $\{\bb_\bj\}_{\bb_\bj}$ denote the computational basis of the grouped qudits. Since
    $G$ is planar (\Cref{cor:PA} item (1)) the number of edges $m$ is $\mc O(n)$, with $n$ being the number of
    vertices. As usual, we assume access to a classical Gibbs sampler which obtains a computational basis state
$\ket{\psi(\bb_\bj)}$ with probability $p(\bb_\bj)$ in time $T+ \mc O(m)=T+ \mc O(n)$. That is, we prepare a state
  \begin{align}
      &\rho(\cP):= \sum_{\bb_\bj} p(\bb_\bj)  \ket{\psi(\bb_\bj)}\bra{\psi(\bb_\bj)}\\
      \text{such that }&\left\| \rho(\cP) - \rho(H_G^{(2)},\beta)\right\|_1 \leq \epsilon\,.	\label{eq:gibbs}
  \end{align}
  Here we did not write down the explicit formula for $p(\bb_\bj)$ since we will not use it. In the second step of the
  algorithm, we sequentially measure the current state with respect to each removed term $p \in \cR$ and perform an oblivious
  randomized correction. The details are in \Cref{alg:cA}.  
	
\begin{algorithm}[H]
\caption{Oblivious Randomized Correction}\label{alg:cA}
\begin{algorithmic}[1]
\State Set the current completed set as $Q=\cP$.
\State Set the current state as $\rho(Q)\leftarrow  \ket{\psi(\bb_\bj)}$. \Comment{Prepare $\rho(\cP)$.}
\For{$p\in \cR$}
	\State Measure the current state $\rho(Q)$ w.r.t measurement $p$.  \label{line:rQ}
	\If{the outcome is $\lambda_p\in \{\pm c_p\}$}
		\State w.p. $\frac{\exp(-\beta \lambda_p)}{\exp(-\beta \lambda_p)+\exp(\beta \lambda_p)}$ do nothing and w.p. $\frac{\exp(\beta \lambda_p)}{\exp(-\beta \lambda_p)+\exp(\beta \lambda_p)}$ apply $L_p$ to the measured states.
	\EndIf
	\State Set $Q  \leftarrow Q\cup \{p\}$. Denote the current state as $\rho(Q)$.
\EndFor
\end{algorithmic}
\end{algorithm}

We prove correctness by induction. Define $H_Q:=\sum_{p\in Q} p$. We claim the following.

\begin{lemma}\label{lem:reduction}
    Assume $\bb_\bj$ are sampled from the correct classical Gibbs distribution over $H^{(2c)}_G$. At the end of each
    \textbf{for} iteration in \Cref{alg:cA}, the current state $\rho(Q)$ satisfies
	\begin{align}
			\|\rho(Q)- \rho(H_Q,\beta)\|_1 \leq \epsilon\,. \label{eq:induction}
	\end{align} 
\end{lemma}
\begin{proof}
When $Q=\cP$, note that by definition $H_\cP = H_G^{(2)}$. Thus Lemma \ref{lem:reduction} holds by assumption on the
initial distribution, as given in \Cref{eq:gibbs}.

Suppose Lemma \ref{lem:reduction} holds for a set $Q$.  Denote $\Lambda_Q$ be the set of distinct vectors $\bl_Q$ where
$\Pi^{Q}_{\bl_Q}$ is not 0. Note that
\begin{align}
	\rho(H_Q,\beta) 	=\sum_{\bl_Q\in\Lambda_Q} \frac{\exp(-\beta \lambda(Q))}{Z(Q)}  \cdot \Pi^Q_{\bl_Q},
\end{align}
where $Z(Q)$ is the partition function for $H_Q$ at inverse temperature $\beta$,
\begin{align}
    Z(Q):=\sum_{\bl_Q\in\Lambda_Q}	 \exp(-\beta \lambda(Q)) \cdot  \text{tr}\left(\Pi^Q_{\bl_Q} \right)\,.
\end{align}

 Now, consider the next iteration where we measure some $p\in \cR\backslash Q$.
 Let us first  assume  that in \Cref{alg:cA} line \ref{line:rQ}, we are measuring the exact Gibbs state 
 \[
     \hat{\rho}(Q):=\rho(H_Q,\beta)
\]
 rather than $\rho(Q)$. We may represent the operation performed during each iteration as a quantum channel $\cN_p$.
 Then, the state at the end of the iteration is
\begin{align}
	  \hat{\rho}(Q\cup\{p\}) & := \cN_p(\hat{\rho}(Q))\label{eq:basic} 
	  \\
	   &= \sum_{\bl_Q} \sum_{\lambda_p\in \{\pm c_p\}}  \left[ \frac{\exp(-\beta \lambda(Q))}{Z(Q)} \cdot   \Pi^p
	  _{\lambda_p} \Pi^Q_{\bl_Q} \Pi^p_{\lambda_p} \cdot \frac{\exp(-\beta \lambda_p)}{\exp(-\beta \lambda_p)+\exp(\beta \lambda_p)}\right. \nonumber\\
	   &\quad\left.+ \frac{\exp(-\beta \lambda(Q))}{Z(Q)} \cdot  L_p \Pi^p
	  _{\lambda_p} \Pi^Q_{\bl_Q} \Pi^p_{\lambda_p}  L_p^\dagger \cdot \frac{\exp(\beta \lambda_p)}{\exp(-\beta \lambda_p)+\exp(\beta \lambda_p)}\right]\nonumber\\
	  &= \sum_{\bl_Q} \sum_{\lambda_p\in \{\pm c_p\}}  \left[ \frac{\exp(-\beta \lambda(Q))}{Z(Q)} \cdot   \Pi^p
	  _{\lambda_p} \Pi^Q_{\bl_Q} \Pi^p_{\lambda_p} \cdot \frac{\exp(-\beta \lambda_p)}{\exp(-\beta \lambda_p)+\exp(\beta \lambda_p)}\right.\nonumber\\
	   &\quad+\left. \frac{\exp(-\beta \lambda(Q))}{Z(Q)} \cdot   \Pi^p
	  _{-\lambda_p} \Pi^Q_{\bl_Q} \Pi^p_{-\lambda_p}  \cdot \frac{\exp(\beta \lambda_p)}{\exp(-\beta \lambda_p)+\exp(\beta \lambda_p)}\right]\label{eq:flip}\\
	  &= \sum_{\bl_Q} \sum_{\lambda_p\in \{\pm c_p\}}  \frac{2\exp(-\beta \lambda(Q))}{Z(Q)}   \frac{\exp(-\beta \lambda_p)}{\exp(-\beta \lambda_p)+\exp(\beta \lambda_p)} 
\cdot   \Pi^p
	  _{\lambda_p} \Pi^Q_{\bl_Q} \Pi^p_{\lambda_p}\label{eq:49}
	 \end{align}
where \Cref{eq:flip} comes from Lemma \ref{lem:sym}, and \Cref{eq:49} comes from renaming the $-\lambda_p$ to $\lambda_p$ in the second 
half of \Cref{eq:flip}.

We next argue that $\hat{\rho}(Q\cup\{p\})$ is equal to $\rho(H_{Q\cup\{p\}},\beta)$.
First, notice that since the terms in $Q\cup\{p\}$ are commuting, we have 
\begin{align}
		\Pi^p
	  _{\lambda_p} \Pi^Q_{\bl_Q} \Pi^p_{\lambda_p} = \Pi^{Q\cup \{p\}}_{\bl_{Q\cup\{p\}}}\,. \label{eq:50}
\end{align}
Since $\hat{\rho}(Q \cup \{p\})$ is a positive linear combination of positive operators, we have that
$\hat{\rho}(Q\cup\{p\})\succeq 0$. Moreover, it is correctly normalized:
\begin{align}
    \text{tr}(\hat{\rho}(Q\cup\{p\}))	 &:=  \sum_{\bl_Q} \sum_{\lambda_p\in \{\pm c_p\}} \frac{2\exp(-\beta \lambda(Q))}{Z(Q)}   \frac{\exp(-\beta \lambda_p)}{\exp(-\beta \lambda_p)+\exp(\beta \lambda_p)}
    \cdot \text{tr}(\Pi^p
	  _{\lambda_p} \Pi^Q_{\bl_Q} \Pi^p_{\lambda_p}) \\
	  & = \sum_{\bl_Q} \sum_{\lambda_p\in \{\pm c_p\}} \frac{2\exp(-\beta \lambda(Q))}{Z(Q)}   \frac{\exp(-\beta \lambda_p)}{\exp(-\beta \lambda_p)+\exp(\beta \lambda_p)}
      \cdot \frac{1}{2} \text{tr}(\Pi^Q_{\bl_Q}) \label{eq:half}\\
      & = \sum_{\bl_Q}\frac{\exp(-\beta \lambda(Q))}{Z(Q)}  \text{tr}(\Pi^Q_{\bl_Q})   \sum_{\lambda_p\in \{\pm c_p\}}  \frac{\exp(-\beta \lambda_p)}{\exp(-\beta \lambda_p)+\exp(\beta \lambda_p)}\\
      & = \text{tr}(\rho(H_Q,\beta))\\
&=1\,, \label{eq:48}
\end{align}
where \Cref{eq:half} comes from \Cref{lem:dim}.
 In summary,
 \begin{itemize}
     \item  By \Cref{eq:48}, $\hat{\rho}(H\cup \{p\})$ is a quantum state (this can also be inferred from the definition
         of the algorithm. However, the above calcuations also show this explicitly.)
     \item \Cref{eq:49,eq:50} imply that $\hat{\rho}(H\cup \{p\})$  can be block-diagonalized with respect to the
         projectors $\Pi^{Q\cup \{p\}}_{\bl_{Q\cup\{p\}}}$, as should be true for a genuine Gibbs state. Additionally,
         within a fixed $\bl_Q$, the term $\exp(-\beta\lambda_p)+\exp(\beta\lambda_p)$ in the denominator of the weights
         in \Cref{eq:49} takes the same value for each $\lambda_p \in\{\pm c_p\}$. Thus, the eigenvalues are
         proportional to  ${\exp(-\beta \lambda(Q)-\beta \lambda_p)}$, which shows that $\hat \rho(H \cup \{p\})$ also
         has the correct weights.
 \end{itemize}
  We conclude that
 \begin{align}
 &\cN_p(\hat \rho(Q)) = \hat{\rho}(Q\cup\{p\}) =\rho(H_{Q\cup\{p\}},\beta)\,.
 \end{align}

 Of course, in Algorithm \ref{alg:cA} we start the \textbf{for} loop with the state $\rho(Q)$, which may not be the
 exact Gibbs state $\hat \rho(Q)$. Nonetheless, for the final state $\rho(Q\cup\{p\})$, we derive
\begin{align}
	\|\rho(Q\cup\{p\}) - \rho(H_{Q\cup\{p\}},\beta) \|_1 &= 	\|\cN_p(\rho(Q))- \cN_p(\hat{\rho}(Q)) \|_1\\
	& \leq \|\rho(Q)-\hat{\rho}(Q)\|_1\label{eq:trace}\\
	&\leq \epsilon \label{eq:epsilon}. 
\end{align}
\Cref{eq:trace} comes from the monotonicity of trace distance under quantum channels, and \Cref{eq:epsilon} comes
from the induction hypothesis that \Cref{eq:induction} holds at the beginning of each iteration. 
\end{proof}

\begin{proof}[Proof of \Cref{thm:4local2}.]
Theorem \ref{thm:4local2} is just a corollary of \Cref{lem:reduction}. The runtime of \Cref{alg:cA} is the sum of the time used for preparing the Gibbs state  of $H^{(2)}_G$ and the time used to perform the randomized correction for $\cR$. Recall that since 
$G$ is a constant-degree planar graph, the number of edges is $\mc O(m)=\mc O(n)$. Thus 
 the total runtime of \Cref{alg:cA}  is 
\begin{align}
    T + \mc O(m) +  |\cR|\times \mc O(n) = T + \mc O(n^2),
\end{align}
where the $\mc O(n)$ is the cost for the correction operation $L_p$ which is a tensor product state on at most $n$ qubits. 
$|\cR|=\cO(n)$ is the size of the set $\cR$.

\end{proof}

\section{Reduction for 2D (4-local) qubit CLH with classical qubit}
\label{sec:on2D}
\subsection{\texorpdfstring{A review of the techniques of Aharonov et. al \cite{aharonov2018complexity}}{Review of the techniques of Aharonov et. al}}
\label{sec:aharanov_review}

Recall that \cite{aharonov2018complexity} studied the structure of qubit 2D commuting local Hamiltonians to argue that preparing the ground state can be done in $\textsf{NP}$. The two main technical results of our work (\Cref{thm:4local2,thm:gibbs_with_classical}) require opening up this result so that we are not only able to prepare a single ground state, but sample from the Gibbs state of the Hamiltonian. This requires a good understanding of their correction operators, as well as the role classical qubits play in their proof. As such, we dedicate this section to reviewing the techniques used in their paper. 

\subsubsection{\texorpdfstring{$\mathrm C^\star$-algebras and the Structure Lemma}{C-star algebras and the Structure Lemma}}\label{sec:C}
The primary technical tool used in \cite{aharonov2018complexity} (and in nearly all other works on commuting local
Hamiltonians \cite{bravyi2003commutative, schuch2011complexity, irani2023commuting}) is the Structure Lemma for $\mathrm C^\star$ algebras. 

\begin{definition}[$\mathrm C^\star$-algebra]
    For any Hilbert space $\mc H$, let $\mc L(\mc H)$ be the set of all linear operators over $\mc H$. Then, a $\mathrm
    C^\star$-algebra is any complex algebra $\mathcal{A} \subseteq \mc L(\mc H)$ that is closed under the $\dagger$
    operation (playing the role of complex conjugation) and includes the identity. 
\end{definition}

\begin{definition}[Commuting algebras]
    Let $\mathcal{A}$ and $\mathcal{A}'$ be two $\mathrm C^*$-algebras on $\mc H$.  We say $\mathcal{A}$ and $\mathcal{A}'$ \emph{commute} if $[h, h'] = 0$ for all $h \in \mathcal{A}$ and $h' \in \mathcal{A}'$.
\end{definition}

The connection between local Hamiltonians and algebras is made through the concept of an ``induced algebra''.

\begin{definition}[Induced algebra]
    Let $h$ be a Hermitian operator acting on two qudits $q_1$ and $q_2$.  Consider the decomposition of $h$ into
    \begin{equation*}
        h = \sum_{i, j} (h_{ij})_{q_1} \otimes (\ketbra{i}{j})_{q_2}\,.
    \end{equation*}
    Here $\{\ket{i}_{q_2}\}_{i}$ is an orthogonal basis of $\mc H^{q_2}$, the Hilbert space of $q_2$. Then the induced
    algebra of $h$ on $q_1$ is the $\mathrm C^*$-algebra generated by $\{h_{ij}\}_{ij}$ and $\cI$, denoted as 
    \begin{equation*}
 \mathcal{A}_{q_1}(h):= \left\langle \{ (h_{ij})_{q_1}\}_{ij}, \cI \right\rangle\,.
    \end{equation*}
    If the operator $h$ is clear from context, we will abbreviate  $\mc A_{q_1}(h)$ as $\mc A_{q_1}$.
\end{definition}
The following lemma tells us that the induced algebra is independent of the choice of basis for $q_2$.
\begin{lemma}[Claim B.3 of \cite{aharonov2018complexity}]
    Let $h$ be a Hermitian operator and consider two decompositions of $h$
    \begin{equation*}
        h = \sum_{i, j} (h_{ij})_{q_1} \otimes (g_{ij})_{q_2} = \sum_{i, j} (\hat{h}_{ij})_{q_1} \otimes (\hat{g}_{ij})_{q_2}\,,
    \end{equation*}
    where both the sets $\{g_{ij}\}_{ij}$ and $\{\hat{g}_{ij}\}_{ij}$ are linearly independent.  Then the $\mathrm C^*$-algebra generated by $\{h_{ij}\}_{ij}$ and $\{\hat{h}_{ij}\}_{ij}$ are the same. In particular, the Schmidt decomposition of $h$ is one way to generate an induced algebra.
\end{lemma}

The induced algebra gives us a tool to analyze whether two terms commute.
\begin{lemma}
    Let $(h_1)_{q,q_1}$ and $(h_2)_{q,q_2}$ be two Hermitian operators. Then $[h_1, h_2] = 0$ if and only if $\mathcal{A}_{q}(h_1)$ commutes with $\mathcal A_q(h_2)$.
\end{lemma}

Finally, we recall the Structure Lemma, which was first applied to understanding commuting local Hamiltonians by \cite{bravyi2003commutative}. Since then, it has been the principal tool used by subsequent works on the CLH \cite{aharonov2013commuting, aharonov2011complexity, schuch2011complexity, irani2023commuting}. At a high level, the Structure Lemma says that algebras can be block-diagonalized, and that within each of these blocks, the algebra takes on a tensor product structure which identifies its commutant. For a proof of the lemma, see \cite[Section~7.3]{gharibian2015quantum}.
\begin{lemma}[The Structure Lemma]\label{lem:structure}
    Let $\mathcal{A} \subseteq \mc L(\mc H^q)$ be a $\mathrm C^\star$-algebra on $\mc H^q$. Then there exists a direct sum decomposition $\mc H^q = \bigoplus_{i} \mc H^{q}_i$ and a tensor product structure $\mc H^{q}_i = \mc H^q_{(i,1)} \otimes \mc H^q_{(i,2)}$ such that
    \begin{equation*}
        \mathcal{A} = \bigoplus_i \mathcal{L}(\mc H^q_{(i,1)}) \otimes \mathcal I(\mc H^q_{(i,2)})\,.
    \end{equation*}
\end{lemma}
We remark that \Cref{lem:structure} is equivalent to \Cref{lem:Cstructure} in \Cref{sec:2local}.

A corollary of \Cref{lem:structure}, and the reason why it is so useful for characterizing the properties of commuting local Hamiltonians, is that in order to commute with an algebra, another algebra must live entirely within the $\mc H^q_{(i,2)}$ subspaces.  Formally, we have the following.
\begin{corollary}\label{cor:structure}
    Let $(h)_{q,q_1}$ and $(h')_{q,q_2}$ be two Hermitian operators with $[h,h']=0$. Let $\mathcal{A}_q(h),\mathcal{A}_q(h') \subseteq \mc L(\mc H^q)$ be the induced algebras on $\mc H^q$.
    Suppose $\{\mc H^q_{(i,j)}\}_{i,j}$ is the decomposition induced by \Cref{lem:structure} applied to $\mathcal{A}_q(h)$. Then the following holds:
    \begin{align*}
        \mathcal{A}_q(h) &= \bigoplus_i \mathcal{L}(\mc H^q_{(i,1)}) \otimes \mc I(\mc H^q_{(i,2)})\\
        \mathcal{A}_q(h') &\subseteq \bigoplus_i \mc I(\mc H^q_{(i,1)}) \otimes \mathcal{L}(\mc H^q_{(i,2)})\,.
    \end{align*}
    Crucially, all operators keep the decomposition $\mc H^q = \bigoplus_i \mc H^q_i$ invariant.
\end{corollary}

\subsubsection{Restrictions on 2D commuting local Hamiltonians}
With the Structure Lemma in hand, we can see how \cite{aharonov2018complexity} apply these tools to show that 2D qubit
\textsf{CLH} is in \textsf{NP}. 

We start with some definitions. First, we review the notion of a \emph{classical qubit}, defined originally in \Cref{def:classical}. For a qubit $q$ and a commuting Hamiltonian $H$, let $\mc N(q)$ be the set of terms acting non-trivially on $q$. We
say that $q$ is \emph{classical} if there is a non-trivial decomposition of $\mc H^q$ (i.e. $\mc H^q = \bigoplus_{i \in \ell} \mc H^q_i$, with $\ell > 1$) and each term
$h \in \mc N(q)$ keeps this decomposition invariant. We define $\mc C_0(H)$ as the set of classical qubits of $H$. If the Hamiltonian $H$ is clear from context, we abbreviate $\mc C_0(H)$ as $\mc C_0$.
\begin{definition}[Classical restriction]
    Let $H$ be a commuting local Hamiltonian with classical qubits $\mc C_0$. Then, there exists a unitary $U =
    \cI_{\overline{\mc C_0}} \otimes \bigotimes_{q \in \mc C_0} U_q$ such that $\tilde H := U H U^\dagger$ is block
    diagonal with respect to the computational basis on $\mc C_0$. A classical assignment to $\mc C_0$ corresponds
    to a string $s \in \{0,1\}^{|\mc C_0|}$, with restricted Hamiltonian,
    \[
        H|_s = \Pi_s \tilde H \Pi_s  \quad \text{ where } \Pi_s := \bigotimes_{q \in \mc C_0} \ketbra{s_q}{s_q}_q\,.
    \] 
    Moreover, $H|_s$ is still a commuting Hamiltonian.
\end{definition}
\begin{proof}
    Any classical qubit $q \in \mc C_0$ has a decomposition into $\{\pi_q, \cI -\pi_q\}$ such that each term $h \in \mc
    N(q)$ commutes with $\pi_q = \ketbra \psi \psi$ and $\cI - \pi_q = \ketbra{\psi^\perp}{\psi^\perp}$, with $\braket \psi {\psi^\perp} = 0$. This implies there exists a unitary transformation $U_q$ on $q$ such that
    \begin{equation}
        U_q \ket \psi = \ket 0 \quad \text{and} \quad U_q \ket{\psi^\perp} = \ket 1\,.\label{eq:diag}
    \end{equation}
    Applying this to each qubit $q \in \mc C_0$ yields the desired unitary $U = \bigotimes_{q \in \mc C_0} U_q$.

    To show that $H|_s$ is commuting, we note that each $h$ in the original Hamiltonian is block diagonal with respect
    to \emph{every} $\{\pi_q, \cI - \pi_q\}$, and using \Cref{eq:diag}, we see that $U h U^\dagger$ is block diagonal
    with respect to $\ketbra 0 0$ and $\ketbra 1 1$. Since $\Pi_S$ is also diagonal in the computational basis,
    $U h U^\dagger$ commutes with $\Pi_S = \otimes_q \ketbra{s_q}{s_q}$. Thus, commutation of $\Pi_s U h U^\dagger \Pi_s$ and $\Pi_s U h' U^\dagger \Pi_s$ reduces to the commutation of $h$ and $h'$.
\end{proof}
\begin{remark}
    \label{rem:reduced_support}
    Since each term in $H|_{s}$ now acts as $\ketbra{s_q}{s_q}$ on any classical qudit $q$, we will treat each term $h \in H|_s$ as having support only on $\text{sup}(h) \cap \overline{\mc C_0}$, where $\text{sup}(h)$ is the set of qubits on which $h$ acts non-trivially.
\end{remark}

A key point is that an initial restriction $H|_s$, with $s \in \{0,1\}^{\mc C_0}$ can lead to the creation of new classical qubits in $H|_s$. For instance, in \Cref{fig:classical_propagation}, setting $s_c = \ket{1}$ removes $h$ from the Hamiltonian and causes $q$ to become classical.
\tikzset{anchor/.append code=\let\tikz@auto@anchor\relax}
\begin{figure}[!ht]
\centering
\begin{tikzpicture}[scale=1.5,
        classical/.style={draw,circle,inner sep=0pt,minimum size=5pt},    
        quantum/.style={fill,circle,inner sep=0pt,minimum size=5pt},
        alglabel/.style={font=\tiny}
    ]
    \foreach \i in {0,...,2} {
        \draw [very thin,gray] (\i,0) -- (\i,2);
    }
    \foreach \i in {0,...,2} {
        \draw [very thin,gray] (0,\i) -- (2,\i);
    }

    \draw[fill=gray!20] (0,1) rectangle node[midway,black] {$h$} ++ (1,1) ;
    \draw[fill=gray!20] (1,1) rectangle node[midway,black] {$h_1$} ++ (1,1) ;
    \draw[fill=gray!20] (0,0) rectangle node[midway,black] {$h_2$} ++ (1,1) ;
    \draw[fill=gray!20] (1,0) rectangle node[midway,black] {$h_3$} ++ (1,1) ;

    \node[classical,label={[label distance=-0.1cm]above right:$c$}] (c) at (1,2) {};
    \node[quantum,label={[label distance=.1cm]above:$q$}] (q) at (1,1) {};

    \node[label={[label distance=-0.1cm,anchor=north east]\tiny $\ketbra 0 0$}] at (c) {};
    \node[label={[label distance=-0.1cm,anchor=north west]\tiny $\ketbra 1 1$}] at (c) {};
    \node[label={[label distance=-0.1cm,anchor=south west]\tiny $\alg Z$}] at (q) {};
    \node[label={[label distance=-0.1cm,anchor=north west]\tiny $\alg \Id$}] at (q) {};
    \node[label={[label distance=-0.1cm,anchor=south east]\tiny $\alg X$}] at (q) {};
    \node[label={[label distance=-0.1cm,anchor=north east]\tiny $\alg Z$}] at (q) {};

\end{tikzpicture}
\captionsetup{width=.8\linewidth}
\caption{$c$ is an initial classical qubit. If $c$ is set to $\ket 1$, then all terms acting on $q$ have local algebras $\subseteq \alg Z$, rendering $q$ classical.}
\label{fig:classical_propagation}
\end{figure}
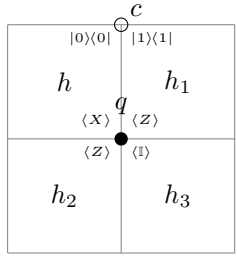

\begin{definition}[Propagated classical qubits]
    For a commuting Hamiltonian $H$, let $\mc C_0$ be the set of qubits which are classical with respect to $H$. Then given an assignment $s_0$ to $\mc C_0$, we write $\mc C_1(s_0)$ to denote the classical qubits of the restricted Hamiltonian $H|_{s_0}$. In general, we write $\mc C_i(s_0,\dots,s_{i-1})$ to refer to the classical qubits of
    \[
        H|_{(s_0,\dots,s_{i-1})} := (((H|_{s_0})|_{s_1})\dots)|_{s_{i-1}}\,.
    \]
\end{definition}

\begin{definition}[Valid restriction]
    A \emph{valid} restriction of $H$ is a sequence $\bs s = (s_0, \dots, s_i)$ and Hamiltonian $H|_{\bs s}$ such that
    $s_i$ is supported on the classical qubits from the prior restrictions $s_0, \dots, s_{i-1}$, i.e. $s_i \in
    \{0,1\}^{\mc C_i(s_0,\dots, s_{i-1})}$. We say that a valid restriction is \emph{terminating} if $H|_{\bs s}$ has no
    classical qubits and $\mc C_{i+1}(s_0,\dots, s_i) = \emptyset$.
\end{definition}
When we refer to a terminating restriction, that restriction is implictly assumed to be valid. 

\begin{definition}[Possible Classical Qubits] \label{def:possibly_classical}
    For a commuting Hamiltonian $H$, we write $\mc C$ to be the set of qubit $q$ which 
    are classical with respect to any valid restriction $\bs s$ and the Hamiltonian $H_{\bs s}$. 
\end{definition}
In other words, $\mc C$ is the set of qubits which might become classical when we sequentially choose an assignment for the current classical qubits. For example in Figure \ref{fig:classical_propagation}  $q$ is a ``possible classical qubit'' and both $c,q\in \mc C$.

Now that we have characterized all the possible classical qubits, we define the notion of a ``fully quantum'' qubit (and term).
\begin{definition}[Fully quantum qubit]
    Given a commuting Hamiltonian $H$, we say that a qubit $q$ is fully quantum if $q$ is not a possible classical qubit (i.e. for any valid restriction $\bs s$, the qubit $q$ remains non-classical in $H|_{\bs s}$).
\end{definition}
\begin{definition}[Fully quantum term]
    \label{def:fully_quantum_term}
    We say that a term $h$ of $H$ is fully quantum if all qubits $q$ on which \old{$H$} \new{$h$} acts non-trivially are fully quantum.
\end{definition}

With this notation, we can understand the first step of the algorithm of \cite{aharonov2018complexity} as applying a valid, terminating restriction $\bs s = (s_0,\dots, s_i)$ to the Hamiltonian.

\subsection{Classical qudits }\label{sec:classical}

In this section we show  a reduction from a CLH instance $H$ to a classical Hamiltonian with constant locality, albeit with some (fully quantum) terms removed, but without needing to first remove classical qubits. In general, the choice of classical qubits can affect the correction operators for a removed term. To deal with this issue, we make the following assumption.

\begin{assumption}[Fully quantum terms are Uniformly Correctable]
\label{as:uniformly_correctable}
    Let $H$ be an instance of \textsf{CLH} where the set of classical qubits is non-empty. Suppose $h$ is a fully quantum term (\Cref{def:fully_quantum_term}). Consider a terminating restriction of the Hamiltonian $H|_s$, such that $H|_s$ has no classical qubits. Suppose that in $H|_s$, $h$ is correctable via the correction operator $L_h$.  Then we assume that $h$ is correctable via the \emph{same} correction operator $L_h$ for any other terminating restriction of the Hamiltonian, $H|_t$, $t \neq s$.
\end{assumption}

Formally, our result is stated as follows.
\begin{theorem}
\label{thm:gibbs_with_classical}
    Let $H$ be an instance of CLH with some classical qubits. Then, by removing a set of fully quantum terms
    $\mc R$, the resulting Hamiltonian can be converted to a $2+\mathcal O(1)$-local classical Hamiltonian $H^{(c)}$ via a constant depth quantum circuit. Moreover, if
    \begin{itemize}
        \item we assume \Cref{as:uniformly_correctable} and,
        \item there is an algorithm to prepare the Gibbs state $\rho'$ of $H^{(c)}$ in time $T$ within precision
            $\epsilon$,
    \end{itemize}
    then there is a quantum algorithm to prepare the Gibbs state $\rho$ of the original Hamiltonian
    $H$ in time $T +  \mc O(n^2)$ within precision $\epsilon$.
\end{theorem}

We prove this theorem in two steps. First, in \Cref{sec:modification} we show how
to modify the proof of \cite{aharonov2018complexity} so that classical qubits do not need to be removed when preparing a
single ground state. Then in \Cref{sec:proof_for_classical_gibbs} we show how to extend this idea to Gibbs state
sampling and get a proof for \Cref{thm:gibbs_with_classical}.

\subsubsection{Characterization of classical qubits}
Proving \Cref{thm:gibbs_with_classical} will require characterizing the set of possible classical qubits $\mc C$
(\Cref{def:possibly_classical}).  We begin with a lemma about algebra on qubits in the interior (recall \Cref{def:boundary}).

\begin{lemma}[Propagation of Classical Qubits]
   Let $c \in \mc C_0$ be a classical qubit in the original Hamiltonian $H$, which is acted on non-trivially by a term
   $A$. Suppose $B$ is another term which interacts non-trivially with $A$ on two other qubits $q,q' \neq c$.
   Furthermore, suppose that $B$ is an interior term and is supported entirely outside
   of $\mc C_0$. Then any projection $\pi_q$ and $\pi_{q'}$ respecting the local algebras $\mc A_q(B)$ and $\mc
   A_{q'}(B)$ respectively satisfies
   \[
       \mc A_r(B) = \mc A_r(\pi_q \pi_{q'} B \pi_{q'} \pi_q) \quad \text{and} \quad \mc A_{r'}(B) = \mc A_{r'}(\pi_q
       \pi_{q'} B \pi_{q'} \pi_q)\,,
   \]
   where $r,r'$ are the other qubits in the support of $B$. In particular, ``classical-ness'' does not propagate from
   $c$ through interior terms.
\end{lemma}
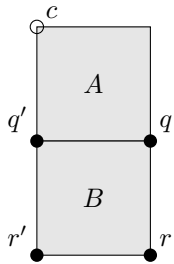
\begin{figure}[!ht]
\centering
\begin{tikzpicture}[scale=1.5,
        classical/.style={draw,circle,inner sep=0pt,minimum size=5pt},    
        quantum/.style={fill,circle,inner sep=0pt,minimum size=5pt},
        alglabel/.style={font=\tiny}
    ]

    \draw[fill=gray!20] (0,1) rectangle  ++ (1,1)node[midway] {$A$} ;
    \draw[fill=gray!20] (0,0) rectangle  ++ (1,1)node[midway] {$B$} ;

    \node[classical,label={[label distance=-0.1cm]above right:$c$}] (c) at (0,2) {};
    \node[quantum,label={[label distance=-0.1cm]above right:$q$}] (q) at (1,1) {};
    \node[quantum,label={[label distance=-0.1cm]above left:$q'$}] (qp) at (0,1) {};
    \node[quantum,label={[label distance=-0.1cm]above right:$r$}] (r) at (1,0) {};
    \node[quantum,label={[label distance=-0.1cm]above left:$r'$}] (rp) at (0,0) {};
\end{tikzpicture}
\caption{The terms $A$ and $B$ and relevant qubits.}
\label{fig:classical_prop}
\end{figure} 
\begin{proof}
    First, we claim that the assumption that $B$ is in the interior and supported entirely on non-classical qubits
    implies that $\mc A(B) = \alg{Z^{\otimes 4}}$, under some change of basis. The proof is essentially by Theorem 5.3
    of \cite{aharonov2018complexity}, except that they prove the claim for \emph{each} interior term in the entire
    Hamiltonian, with the assumption that all classical qubits have been removed. In our case, we only apply their proof
    for a single term which is not supported on any classical qubits. For completeness, we outline the required
    components.

    Consider the sequence of qubits $(q,q',r',r)$ acted on by $B$. Since each of these qubits is in the interior, this
    implies that the two ``star'' and two ``plaquette'' terms adjacent to each qubit act non-trivially on it, and thus
    induce a $2$-dimensional algebra (Claim F.2 of \cite{aharonov2018complexity}). This implies (by Lemma F.5) that $\mc
    A_{p,p'}(B) = \alg{ Z \otimes Z}$ for every length two subsequence of $(q,q',r',r)$ (i.e. every edge of $B$). Now
    consider length three subsequences $(p,p',p'')$. Lemma B.4 of \cite{aharonov2018complexity} tells us that
    \[
        \mc A_{p,p',p''}(B) \subseteq \mc A_{p,p'}(B) \otimes \mc A_{p''}(B) = \alg{Z \otimes Z} \otimes \alg Z.
    \]
    Similarly,
    \[
        \mc A_{p,p',p''}(B) \subseteq \mc A_p(B) \otimes \mc A_{p',p''}(B) = \alg Z \otimes \alg{ Z \otimes Z}.
    \]
    By writing down the permissible expressions for $h \in \mc A_{p,p',p''}(B)$ subject to these conditions, we find
    that $\mc A_{p,p',p''}(B) = \alg{Z \otimes Z \otimes Z}$ (details can be found in \cite{aharonov2018complexity}).
    Extending this argument once more to length-$4$ sequences completes the claim.

    In conclusion $\mc A(B) = \alg{Z^{\otimes 4}}$ and thus $B = \alpha \cI + \beta Z^{\otimes 4}$ with $\beta \neq 0$.
    Consider the effect of applying projector $\pi_q$ and $\pi_{q'}$ on qubits $q$ and $q'$. Since $\mc A_q(B) = \mc
    A_{q'}(B) = \alg Z$, the fact that $\pi_q$ and $\pi_{q'}$ respect the local algebras implies that $\pi_q$ and
    $\pi_{q'}$ are in the $Z$-basis. If both projectors are trivial, the algebra of $B$ on $r,r'$ is unaffected. We give
    the proof for the case when both are non-trivial and dimension $1$; the case when only one is dimension 1 and the
    other is trivial is similar. The result of applying $\pi_q$ and $\pi_{q'}$ is
    \[
        \pi_q \pi_{q'} B \pi_{q'} \pi_{q} = \alpha \pi_q \otimes \pi_{q'} \otimes \cI + \beta ((-1)^{b_1} \pi_q) \otimes
        ((-1)^{b_2} \pi_{q'}) \otimes Z^{\otimes 2}\,,
    \]
    where $b_1, b_2 \in \{0,1\}$ indicate the possible phase (corresponding to whether $\pi_q, \pi_{q'}$ correspond to
    the $+1$ or $-1$ eigenspaces of $Z$). We case on the values of $(b_1,b_2)$.
    \begin{itemize}
        \item If $b_1 = b_2$ then we obtain the Schmidt decomposition
            \[
                B = \pi_q \otimes \pi_q \otimes (\alpha \cI + \beta Z^{\otimes 2})\,,
            \]
            and $\mc A_{r,r'}(\pi_q \pi_{q'} B \pi_{q'} \pi_q)$ remains $\alg{Z^{\otimes 2}}$.
        \item If $b_1 \neq b_2$ then $\pi_q \otimes \pi_{q'}$ and $(-1)^{b_1}\pi_q \otimes (-1)^{b_2}\pi_{q'}$ are
            orthogonal, and we have the Schmidt decomposition
            \[
                B = \pi_q \otimes \pi_q \otimes (\alpha \cI) + (-1)^{b_1}\pi_q \otimes (-1)^{b_2}\pi_{q'} \otimes \beta
                Z^{\otimes 2}\,.
            \]
            Again, this leaves the local algebra unchanged.
    \end{itemize}
    Therefore, no choice of (consistent) projectors on $q, q'$ change $B$'s local algebra on $\{r,r'\}$.
\end{proof}

This lemma provides strong restrictions on the possible set of classical qubits after ``propagating'' the initial set of
classical qubits.
\begin{lemma}[Characterization of Classical Qubits]
    \label{lem:classical_characterization}
    Let $\mc C$ be the set of all possibly classical qubits (as in \Cref{def:possibly_classical}). Then, each classical
    qubit $c \in \mc C$ is either
    \begin{itemize}
        \item a classical qubit in the original Hamiltonian $H$, or
        \item supported on a boundary term or a term supported on $\mc C_0$ (i.e. adjacent to an originally classical
            qubit).
    \end{itemize}
\end{lemma}
\begin{proof}
    Assume $c \in \mc C$ is not originally a classical qubit. Then, there must have been some set of projections
    $\pi_{c_1}, \dots, \pi_{c_k}$ such that in the Hamiltonian $\pi_{c_k} \dots \pi_{c_1} H \pi_{c_1} \dots \pi_{c_k}$,
    the qubit $c$ is a classical. Suppose $c$ is acted on by $N(c) = \{h_1, h_2, h_3, h_4\}$. Certainly, if no projection
    $\pi_{c_i}$ is applied to a qubit in the support of any $h \in N(c)$, it does not render $c$ a classical qubit.
    Otherwise, \Cref{lem:classical_characterization} implies that if each $h \in N(c)$ is supported outside of
    $\mc C_0$ and in the interior of the system, no projection applied to a qudit of $h_i$ renders $c$ classical. Thus,
    some term $h \in N(c)$ is either a boundary term or supported on $\mc C_0$.
\end{proof}

\subsubsection{Uniform preparation of a ground state}
\label{sec:modification}
In this section, we describe how to avoid restricting classical qubits in the proof of \cite{aharonov2018complexity}. 
\begin{remark}[Comparison to the proof of Aharonov et.\ al \cite{aharonov2018complexity}]
We emphasize that this new proof does not qualitatively improve on their result as the initial Hamiltonian is
block-diagonal with respect to $\mc C$ and thus for the task of \emph{finding} a ground state, one can always assume the
first step is to perform a classical restriction. Moreover, this new proof is \emph{worse} in the sense that
\cite{aharonov2018complexity} obtains a $2$-local classical Hamiltonian, whereas, in general, we need to assume  Assumption \ref{as:uniformly_correctable} and our Hamiltonian can be
$k$-local, for some $k \in \mc O(1)$ depending on the structure of the Hamiltonian. However,  the ideas used here will be
useful for the full proof of \Cref{thm:gibbs_with_classical} in \Cref{sec:proof_for_classical_gibbs}, which allows us to
prepare the Gibbs state instead of only a single ground state.
\end{remark}

Let us first recall the high level proof of \cite{aharonov2018complexity}, which  can be boiled
down to:
\begin{enumerate}
    \item The prover provides a terminating restriction $\bs s$ such that $H|_{\bs s}$ contains a ground state of $H$.
        This removes all classical qubits.
    \item Remove a set of terms $\mc R$ of $H|_{\bs s}$ so that the Hamiltonian becomes two local. \label{item:removal}
    \item Prepare a ground state of the two-local Hamiltonian via \cite{bravyi2003commutative} (with the help of the
        prover). \label{item:prepare_gs}
    \item Correct the removed operators in $\mc R$ to get a ground state for $H$. \label{item:correct}
\end{enumerate}

Our first observation is that \Cref{item:removal} can be performed without first removing classical qubits, as this step
depends only on the geometric structure of the Hamiltonian. However, the issue comes in \Cref{item:correct}; depending
on the choice of classical qubits, the set of ``correctable'' terms may be different. Therefore, if we perform
\Cref{item:removal} without considering the classical restriction and the resulting set of correctable terms, we may not
be able to correct all terms in $\mc R$.

More precisely, \cite{aharonov2018complexity} characterize the set of correctable terms via paths to the boundary.
\begin{lemma}[Access to the Boundary \cite{aharonov2018complexity}]
    \label{lem:original_access_to_boundary}
    Let $H$ be a CLH instance without any classical qubits, and let $h,h'$ be terms in the interior of the system, such
    that $h$ and $h'$ share a single edge. Then \emph{either} $h$ or $h'$ has a path to the boundary. If $h$ (or $h'$)
    has a path to the boundary, we say that $h$ is \emph{correctable}.
\end{lemma}
These ``paths to the boundary'' yield correction operators, and the choice to remove $h$ or $h'$ depends on which term
is correctable. For instance in \Cref{fig:choice_of_classical}, we see that our choice for the classical qubit $c$
changes which of the two boxed terms may be removed.

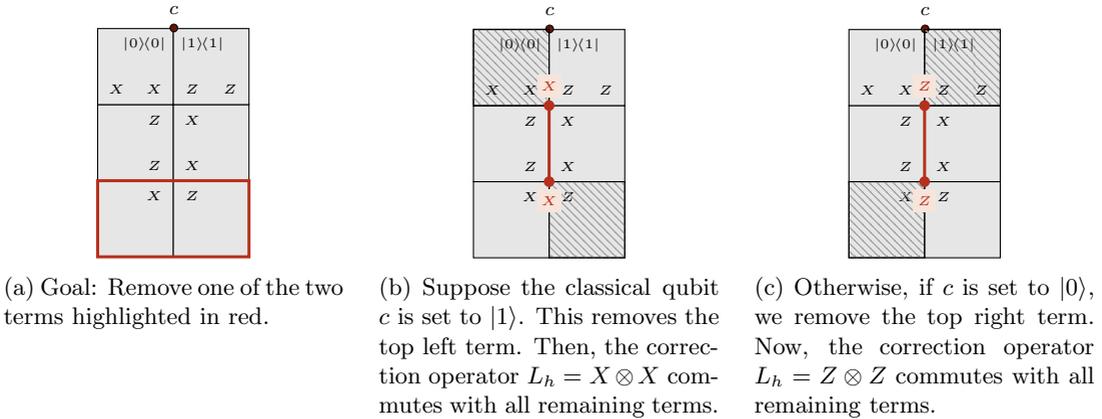
\begin{figure}
    \centering
    \begin{subfigure}[t]{0.3\textwidth}
        \centering
        \begin{tikzpicture}[
            every node/.append style={font=\tiny},
            alglabel/.style = {inner sep=0.1cm,minimum width=0.5cm, minimum height=0.4cm},
            vertexp/.style = {draw, fill=Sepia, circle, minimum width=0.1cm, inner sep=0}
        ]
            \foreach \y in {0,1,2} {
                \foreach \x in {0,1} {
                    \draw[fill=black!10] (\x,\y) rectangle ++(1,1) node[fitting node] (n\x\y) {};
                }
            } 
            \node[vertexp, label={above:\scriptsize $c$}] at (n02.north east) {};
            \node[alglabel, anchor=north east] at (n02.north east) {$\ketbra 0 0$};
            \node[alglabel, anchor=north west] at (n12.north west) {$\ketbra 1 1$};
            \node[alglabel, anchor=south east] at (n02.south east) {$X$};
            \node[alglabel, anchor=south west] at (n02.south west) {$X$};
            \node[alglabel, anchor=south east] at (n12.south east) {$Z$};
            \node[alglabel, anchor=south west] at (n12.south west) {$Z$};
            \node[alglabel, anchor=north east] at (n01.north east) {$Z$};
            \node[alglabel, anchor=south east] at (n01.south east) {$Z$};
            \node[alglabel, anchor=north west] at (n11.north west) {$X$};
            \node[alglabel, anchor=south west] at (n11.south west) {$X$};
            \node[alglabel, anchor=north west] at (n10.north west) {$Z$};
            \node[alglabel, anchor=north east] at (n00.north east) {$X$};

            \draw[line width=0.4mm, fill=none, BrickRed] (0,0) rectangle ++ (2,1);
        \end{tikzpicture}
        \caption{Goal: Remove one of the two terms highlighted in red.}
    \end{subfigure}
    \quad
    \begin{subfigure}[t]{0.3\textwidth}
        \centering
        \begin{tikzpicture}[
            every node/.append style={font=\tiny},
            alglabel/.style = {inner sep=0.1cm,minimum width=0.5cm, minimum height=0.4cm},
            vertexp/.style = {draw, fill=Sepia, circle, minimum width=0.1cm, inner sep=0},
            correction/.style = {draw, fill=BrickRed, circle, minimum width=0.1cm, inner sep=0}
        ]
            \foreach \y in {0,1,2} {
                \foreach \x in {0,1} {
                    \draw[fill=black!10] (\x,\y) rectangle ++(1,1) node[fitting node] (n\x\y) {};
                }
            } 
            \node[vertexp, label={above:\scriptsize $c$}] at (n02.north east) {};
            \draw[pattern=north west lines, pattern color=black!40] (0,2) rectangle +(1,1);
            \node[alglabel, anchor=north east] at (n02.north east) {$\ketbra 0 0$};
            \node[alglabel, anchor=north west] at (n12.north west) {$\ketbra 1 1$};
            \node[alglabel, anchor=south east] at (n02.south east) {$X$};
            \node[alglabel, anchor=south west] at (n02.south west) {$X$};
            \node[alglabel, anchor=south west] at (n12.south west) {$Z$};
            \node[alglabel, anchor=south east] at (n12.south east) {$Z$};
            \node[alglabel, anchor=north east] at (n01.north east) {$Z$};
            \node[alglabel, anchor=south east] at (n01.south east) {$Z$};
            \node[alglabel, anchor=north west] at (n11.north west) {$X$};
            \node[alglabel, anchor=south west] at (n11.south west) {$X$};
            \node[alglabel, anchor=north west] at (n10.north west) {$Z$};
            \draw[pattern=north west lines, pattern color=black!40] (1,0) rectangle +(1,1);
            \node[alglabel, anchor=north east] at (n00.north east) {$X$};
    
            \draw[line width=0.4mm, BrickRed] (1,1) node [correction] (cor1) {} --  (1,2) node [correction] (cor2) {};
            \node[draw=none,fill=BrickRed!10, inner sep=0.05cm, minimum width=0.3cm, minimum height=0.3cm, anchor=south,yshift=0.1cm,text=BrickRed] at (cor2) {$X$};
            \node[draw=none,fill=BrickRed!10, inner sep=0.05cm, minimum width=0.3cm, minimum height=0.3cm, anchor=north,yshift=-0.1cm,text=BrickRed] at (cor1) {$X$};
        \end{tikzpicture}
        \caption{Suppose the classical qubit $c$ is set to $\ket 1$. This removes the top left term. Then, the correction operator $L_h = X \otimes X$ commutes with all remaining terms.}
    \end{subfigure}
    \quad
    \begin{subfigure}[t]{0.3\textwidth}
        \centering
        \begin{tikzpicture}[
            every node/.append style={font=\tiny},
            alglabel/.style = {inner sep=0.1cm,minimum width=0.5cm, minimum height=0.4cm},
            vertexp/.style = {draw, fill=Sepia, circle, minimum width=0.1cm, inner sep=0},
            correction/.style = {draw, fill=BrickRed, circle, minimum width=0.1cm, inner sep=0}
        ]
            \foreach \y in {0,1,2} {
                \foreach \x in {0,1} {
                    \draw[fill=black!10] (\x,\y) rectangle ++(1,1) node[fitting node] (n\x\y) {};
                }
            } 
            \node[vertexp, label={above:\scriptsize $c$}] at (n02.north east) {};
            \draw[pattern=north west lines, pattern color=black!40] (1,2) rectangle +(1,1);
            \node[alglabel, anchor=north east] at (n02.north east) {$\ketbra 0 0$};
            \node[alglabel, anchor=north west] at (n12.north west) {$\ketbra 1 1$};
            \node[alglabel, anchor=south east] at (n02.south east) {$X$};
            \node[alglabel, anchor=south west] at (n02.south west) {$X$};
            \node[alglabel, anchor=south west] at (n12.south west) {$Z$};
            \node[alglabel, anchor=south east] at (n12.south east) {$Z$};
            \node[alglabel, anchor=north east] at (n01.north east) {$Z$};
            \node[alglabel, anchor=south east] at (n01.south east) {$Z$};
            \node[alglabel, anchor=north west] at (n11.north west) {$X$};
            \node[alglabel, anchor=south west] at (n11.south west) {$X$};
            \node[alglabel, anchor=north west] at (n10.north west) {$Z$};
            \node[alglabel, anchor=north east] at (n00.north east) {$X$};
            \draw[pattern=north west lines, pattern color=black!40] (0,0) rectangle +(1,1);
    
            \draw[line width=0.4mm, BrickRed] (1,1) node [correction] (cor1) {} --  (1,2) node [correction] (cor2) {};
            \node[draw=none,fill=BrickRed!10, inner sep=0.05cm, minimum width=0.3cm, minimum height=0.3cm, anchor=south,yshift=0.1cm,text=BrickRed] at (cor2) {$Z$};
            \node[draw=none,fill=BrickRed!10, inner sep=0.05cm, minimum width=0.3cm, minimum height=0.3cm, anchor=north,yshift=-0.1cm,text=BrickRed] at (cor1) {$Z$};
        \end{tikzpicture}
        \caption{Otherwise, if $c$ is set to $\ket 0$, we remove the top right term. Now, the correction operator $L_h = Z \otimes Z$ commutes with all remaining terms.}
    \end{subfigure}
    \caption{The choice of correction operator may depend on the choice of classical qubits.}
    \label{fig:choice_of_classical}
\end{figure}

To avoid this issue, we take \Cref{as:uniformly_correctable}, which states that if some $h$ is correctable under some
valid, terminating restriction $H|_{\bs s}$, then it is correctable under \emph{any} other valid, terminating restriction
$H|_{\bs t}$. This implies that the classical restriction can be deferred until after \Cref{item:prepare_gs} is
performed. We will now formalize this idea.

\tikzset{pics/tile1/.style n args={3}{code={
    \draw (#1,#2) -- ($ (2,0) + (#1,#2) $) -- ($ (1,1.72) + (#1,#2) $) -- cycle ;
    \node (c#1#3) at ($ (#1, #2) + (1,0.86)$) {};
}}}
\tikzset{pics/tile2/.style n args={3}{code={
    \draw ($ (2,1.72) + (#1,#2) $) -- ($ (1,0) + (#1,#2) $) -- ($ (0,1.72) + (#1,#2) $) -- cycle;
    \node (c#1#3) at ($ (#1, #2) + (1,0.86)$) {};
}}}
\begin{figure}
    \centering
    \begin{tikzpicture}
        \begin{scope}
        \clip (0,0) rectangle ++(5,5);
        \newcommand\delt{0.25}
        \foreach \x in {0,...,29} {
            \foreach \y in {0,...,29} {
                \draw[black!20,line width=0.2mm] ($ \delt*(\x,\y) + (0.1,0) $) rectangle ++(\delt,\delt);
            }
        }
        \foreach \x in {-3,-2,0,2,4} {
            \foreach \i in {-2,-1,0,1,2,3} {
                \pgfmathsetmacro{\y}{\i*1.72}
                \pgfmathtruncatemacro{\xd}{\x+1}
                \ifthenelse{\equal{\intcalcMod{\i}{2}}{0}}{
                    \pic {tile1={\x}{\y}{\i}};
                    \pic {tile2={\xd}{\y}{\i}};
                }{
                    \pic {tile2={\x}{\y}{\i}};
                    \pic {tile1={\xd}{\y}{\i}};
                }
            } 
        }
        \foreach \x in {-2,-1,...,4,5} {
            \foreach \i in {-1,0,...,2,3} {
                \node[fill,circle, minimum width=0.1cm,inner sep=0] at (c\x\i) {};
                \pgfmathtruncatemacro{\xs}{\x-1}
                \pgfmathtruncatemacro{\is}{\i-1}
                \draw[Maroon, line width=0.5mm] (c\x\i.center) -- (c\xs\i.center);
                \ifthenelse{\equal{\intcalcMod{\x+\i}{2}}{0}}{
                        \draw[Maroon, line width=0.5mm] (c\x\i.center) -- (c\x\is.center);
                }{}
            }
        }
        \end{scope}
    \end{tikzpicture}
    \captionsetup{width=0.8\linewidth}
    \caption{An example of a triangulation, followed by a co-triangulation. Within each triangle, a central point is identified. Each point is connected via one of the sides of the triangle to a center of an adjacent triangle. This yields ``tiles'', demarcated by the dark red lines. The set of qubits within a tile are identified as a single qubit. Within this grouping, the only original Hamiltonian  terms that are more than $2$-local are those at the centers of the triangles.}
    \label{fig:triangulation}
\end{figure}
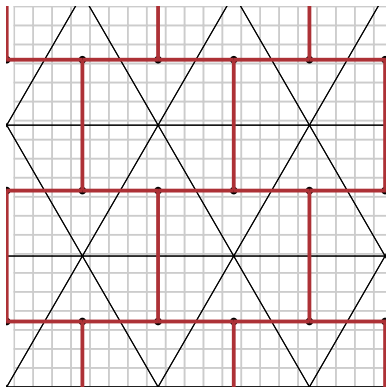
\paragraph{Choice of Removed Terms.} In the proof of
\cite{aharonov2018complexity}, the authors triangulate the 2D complex, then construct a ``co-triangulation'', dividing the
surface in tiles $T \in \mc T$ such that each qubit within a single tile $T$ becomes a new qudit $q_T$ in the
transformed Hamiltonian (see \Cref{fig:triangulation}). Under this transformation, terms in the original Hamiltonian are
either 1) internal to a single tile $T$ (and
are now 1-local), 2) cross between two adjacent tiles $T, T'$ (and are now 2-local), or 3) on the corner of three tiles
(and are now 3-local). Since the  goal is to produce a $2$-local Hamiltonian, the natural strategy is to simply set the
set of removed terms $\mc R$ to be precisely these corner terms. However, we need to be a bit careful, since any
given term is not necessarily correctable; all we know from \Cref{lem:original_access_to_boundary} is that either the
corner term $h$ or its neighbor $h'$ is correctable.\footnote{Recall that in \cite{aharonov2018complexity} classical
qubits have already been removed so $h$ and $h'$ are automatically quantum.} This is easy to handle. As long as the
original triangulation has sufficient girth, the co-triangulation can be ``shifted'' so that its corners lie entirely in
correctable terms.

In our case, we no longer have the ambiguity of which terms is correctable, but we instead need to deal with classical
terms carefully. Consider a triangulation and corresponding tiling $\mc T$. Let $h$ be an arbitrary term containing a
corner of the tiling $\mc T$. We consider the following set of cases:
\begin{itemize}
    \item \textbf{(Case 1)} $h$ is a \emph{fully quantum} term (\Cref{def:fully_quantum_term}).
        \begin{itemize}
            \item \textbf{(Case 1a)} $h$ is an interior term.
            \item \textbf{(Case 1b)} $h$ is a boundary term. 
        \end{itemize}
    \item \textbf{(Case 2)} $h$ is supported on an originally classical qubit $q \in \mc C_0$.
    \item \textbf{(Case 3)} $h$ is supported on a qubit $q \in \mc C \setminus \mc C_0$ (i.e. a qubit which is classical
        under some restriction but which is not classical in the original Hamiltonian).
\end{itemize}
\begin{figure}[!ht]
\centering
\begin{subfigure}{0.4\textwidth}
\centering
\begin{tikzpicture}[scale=1.5,
        classical/.style={draw,circle,inner sep=0pt,minimum size=5pt},    
        quantum/.style={fill,circle,inner sep=0pt,minimum size=5pt},
        pnt/.style={fill,circle,inner sep=0pt,minimum size=2pt},
        alglabel/.style={font=\tiny}
    ]

    \draw[fill=gray!20] (0,1) rectangle node[midway] {$h$} ++ (1,1) ;
    \draw[fill=gray!20] (0,0) rectangle node[midway] {$h_1$} ++ (1,1) ;
    \draw[fill=gray!20] (1,0) rectangle node[midway] {$h_2$} ++ (1,1) ;
    \draw[fill=gray!20] (1,1) rectangle node[midway] {$h_3$} ++ (1,1) ;

    \node[quantum,label={[label distance=-0.1cm]above right:$q$}] (q) at (1,1) {};

    \node[label={[label distance=-0.1cm]below left:$c$}, pnt] (c) at ($(q)!0.3!(0,2)$) {};

    \draw[line width=0.25mm] plot [smooth] coordinates {(c) (0.5,0.5) (-0.5,0.25)};
    \draw[dashed, line width=0.25mm] plot [smooth] coordinates {(-0.5,0.25) (-1,0) (-1,-0.5)};

    \draw[line width=0.25mm] plot [smooth] coordinates {(c) ($(q) + (0.5,0.8)$) (2,1.8) (2.3, 2)};
    \draw[dashed, line width=0.25mm] plot [smooth] coordinates {(2.3,2) (2.5,2.3) (2.5,2.5)};

    \draw[line width=0.25mm] plot [smooth] coordinates {(c) (0.5,1.9) (0,1.8) (-0.5,1.8)};
    \draw[dashed, line width=0.25mm] plot [smooth] coordinates {(-0.5,1.8) (-1,2)};
\end{tikzpicture}
\end{subfigure}
\qquad
\begin{subfigure}{0.4\textwidth}
\centering
\begin{tikzpicture}[scale=1.5,
        classical/.style={draw,circle,inner sep=0pt,minimum size=5pt},    
        quantum/.style={fill,circle,inner sep=0pt,minimum size=5pt},
        pnt/.style={fill,circle,inner sep=0pt,minimum size=2pt},
        alglabel/.style={font=\tiny}
    ]

    \draw[dashed, fill=gray!20] (0,1) -- (0,2) -- (1,2)  -- cycle;
    \draw[] (0,1) rectangle node[midway,anchor=south east] {$h$} ++ (1,1) ;
    \draw[dashed, fill=gray!20] (0,1) -- (1,0) -- (0,0)  -- cycle;
    \draw[] (0,0) rectangle node[midway,anchor=north east] {$h_1$} ++ (1,1) ;
    \draw[dashed, fill=gray!20] (1,0) -- (2,1) -- (2,0)  -- cycle;
    \draw[] (1,0) rectangle node[midway,anchor=north west] {$h_2$} ++ (1,1) ;
    \draw[dashed, fill=gray!20] (1,2) -- (2,2) -- (2,1)  -- cycle;
    \draw[] (1,1) rectangle node[midway,anchor=south west] {$h_3$} ++ (1,1) ;

    \node[quantum,label={[label distance=-0.1cm]above right:$q$}] (q) at (1,1) {};

    \node[label={[label distance=-0.1cm]below left:$c$}, pnt] (c) at ($(q)!0.3!(0,2)$) {};

    \draw[line width=0.25mm] plot [smooth] coordinates {(c) (0.5,0.5) (-0.5,0.25)};
    \draw[dashed, line width=0.25mm] plot [smooth] coordinates {(-0.5,0.25) (-1,0) (-1,-0.5)};

    \draw[line width=0.25mm] plot [smooth] coordinates {(c) ($(q) + (0.5,0.8)$) (2,1.8) (2.3, 2)};
    \draw[dashed, line width=0.25mm] plot [smooth] coordinates {(2.3,2) (2.5,2.3) (2.5,2.5)};

    \draw[line width=0.25mm] plot [smooth] coordinates {(c) (0.5,1.9) (0,1.8) (-0.5,1.8)};
    \draw[dashed, line width=0.25mm] plot [smooth] coordinates {(-0.5,1.8) (-1,2)};
\end{tikzpicture}
\end{subfigure}
\captionsetup{width=.8\linewidth}
\caption{\textbf{(Case 2)} When a center $c$ is placed on a term supported on a classical qubit, any setting of the classical qubit induces a hole in the 2D structure.}
\label{fig:remove_classical}
\end{figure}
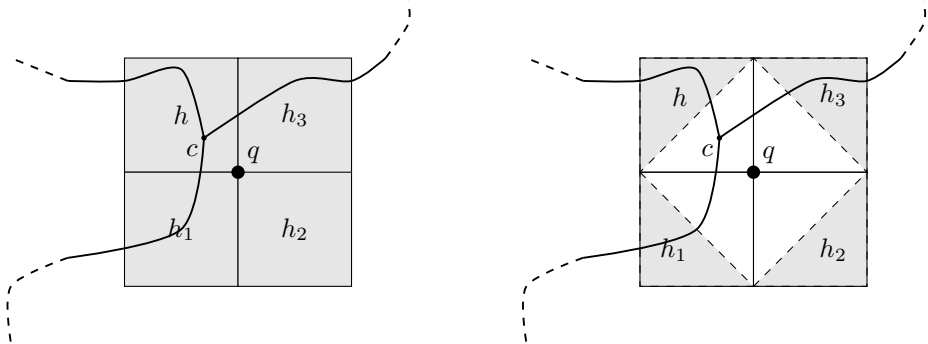
In the first case, we apply the same logic as \cite{aharonov2018complexity}: for \textbf{Case 1a}, $h$ is an interior
term, it will be correctable (in our case by \Cref{as:uniformly_correctable}) and we added it to the removed set $\mc
R$; in \textbf{Case 1b}, $h$ is a boundary term which implies a neighbor acts trivially on one of its qubits. This
neighbor results in a ``hole'' in the original surface. Place the corner of the co-triangulation in the hole.

In \textbf{Case 2}, we construct the co-triangulation such that one tile contains only the classical qubits of $h$; this
will ensure that any assignment to the classical qubit renders this term $2$-local as well. For \textbf{Case 3}, we
appeal to \Cref{lem:classical_characterization}; since $q$ is not originally classical but becomes classical, it must be
on a boundary term or a term with a classical qubit. In the boundary case, we apply the same argument as in \textbf{Case
1b}. In the classical case, apply the argument from \textbf{Case 2}. Finally, we denote the Hamiltonian with terms in
$\mc R$ removed as $\tilde H$.

\paragraph{Construction of the Classical Hamiltonian.} Now that we have fixed the tiling $\mc T$ and chosen the set
of removed terms $\mc R$, it remains to describe how to translate this into a classical Hamiltonian. To do so, we
introduce some notation. For each tile $T \in \mc T$, we denote $Q(T)$ to refer to the qubits contained entirely within
$T$ and $H(T)$ as the set of original Hamiltonian terms in $\tilde H$ overlapping $T$ (i.e. acting non-trivially on some
$q \in Q(T)$). In addition to the qubits internal to $T$, the qubits ``nearby'' $T$ are also important. Define
\[
    \overline{Q}(T) = \bigcup_{h \in H(T)} Q(T),
\]
Additionally, let $\mc N(T)$ be the tiles neighboring $T$. Then, we define the Hilbert space of the tile qudit $q_T$ as
$\mc H_T = \otimes_{q \in T} \mc H_q$. To define the terms in the grouped Hamiltonian, we define the following sets:
\begin{itemize}
    \item $S_T$ is the set of Hamiltonian terms which act non-trivially on $Q(T) \setminus \mc C_0$ (i.e. on a qudit in
        $T$ which is not originally classical).
    \item $S_{T,T'} := S_T \cap S_{T'}$ are the terms acting on both $T$ and $T'$.
    \item $h_{T,T'} = \sum_{h \in S_{T,T'}} h$ will be the 2-local terms in the grouped Hamiltonian.
    \item $h_T = \sum_{\substack{h \in S_T\\\forall T' \neq T, h \not\in S_{T'}}} h$ are the 1-local terms.
\end{itemize}
Then, the grouped Hamiltonian is denoted as $\tilde H_{\mc T} = \sum_{T} h_T + \sum_{T' \neq T} h_{T,T'}$. Again, note that
$h_{T,T'}$ is technically \emph{not} $2$-local; it only becomes $2$-local once a classical restriction to $\mc C_0$ is
fixed. But given such a restriction, we have the following simple consequence of the Structure Lemma
\cite{bravyi2003commutative}.
\begin{lemma}
    Any restriction $s \in \{0,1\}^{|\mc C_0|}$ to the classical qubits in $\mc C_0$ induces a $2$-local structure in
    $H_{\mc T}$. Therefore, for each $q_T$, the terms $h_{T,T'}|_s, T' \in \mc N(T)$ are mutually commuting and induce a
    decomposition of the Hilbert space of $q_T$ as
    \[
        \mc H_T = \bigoplus_{i = 1}^{\ell_T} \mc H_T^{(i)} = \bigoplus_{i=1}^\ell \bigotimes_{T' \in \mc N(T)} \mc H_T^{(i,T')},
    \]
    where within each subspace $\mc H_T^{(i)}$, the term $h_{T,T'}|_s$ acts non-trivially only on $\mc H_T^{(i,T')}$. To
    make the dependence on an assignment $s$ to $\mc C_0$ explicit, we refer to this decomposition as $\mc D^s_T$.
\end{lemma}
We can say something slightly stronger than this; for a tile $T$, the above decomposition only depends on terms
$h_{T,T'}$ which intersect $T$ (and the internal terms as well). Therefore, if we consider a restriction $s'$ where a
bit \emph{outside} of $\overline Q(T)$ is flipped, this induces the same decomposition, i.e. $\mc D^{s'}_T =
\mc D^s_T$. As a result, when specifying a restriction in the context of a triangle, we implicitly imagine $s$ is
defined over $\{0,1\}^{|Q(T) \cap \mc C_0|}$.

Recall our goal is to argue that $h_{T,T'}|_s$ is classical. Recalling \Cref{sec:2local}, we see that the classical
terms constructed for a term $h_{v,w}$ in \Cref{eq:classical_ham} are only a function of the decomposition on vertices $v$ and $w$ (in our
case $q_T$ and $q_{T'}$). In particular, for a fixed assignment $s$ and term $h_{T,T'}|_s$, we obtain the classical term,
\begin{equation}
    \label{eq:dep_classical_ham}
    h_{T,T'}|_s := \sum_{\substack{\bs j = (j_1,j_2)\\ \in [\ell_T] \times [\ell_{T'}]}} \sum_{\bs b_{\bs j}^{T,T'}}
    \lambda(\bs b^{T,T'}_{\bs j}) \ketbra {\psi(\bs b^{T,T'}_{\bs j})}{\psi(\bs b^{T,T'}_{\bs j})},
\end{equation}
where $\bs b^{T,T'}_{\bs j}$ and $\lambda(\bs b^{T,T'}_{\bs j})$ are defined as in \Crefrange{eq:14}{eq:17}, except that
we use $T,T'$ to refer to qudits, rather than $v,w$. Then, we have that the following $2$-local Hamiltonian is
equivalent to $\tilde H|_s$,
\[
    \tilde H^{(c,s)} := \sum_{T,T'} h_{T,T'}|_s + \sum_T h_T|_s.
\]
Now, we have that by construction, any ground state $\ket \psi$ to $\tilde H^{(c,s)}$ is (after some 1-local unitary
transformations) equal to a ground state $\ket \phi$ of $H|_s$; in turn, we can obtain a ground state of the
original Hamiltonian $\tilde H$ as $\ket s \otimes \ket \phi$. In particular, we can write down a Hamiltonian equivalent
to $\tilde H$ as
\[
\sum_{s \in \{0,1\}^{|\mc C_0|}} \ketbra s s \otimes \tilde H^{(c,s)} = \sum_{T,T'}\sum_{s \in \{0,1\}^{|\mc C_0|}}
\ketbra s s \otimes h_{T,T'}|_s + \sum_T\sum_{s \in \{0,1\}^{|\mc C_0|}} \ketbra s s \otimes  h_T|_s.
\]
Unfortunately, since $|\mc C_0|$ can be as large as $\textsf{poly}(n)$, this yields an exponentially-large Hamiltonian.
But we use our earlier observation, which is that the decompositions on $T$ and $T$ only depend on $\overline Q(T) \cap
\mc C_0$. Since the local decompositions are the same, so too are classical terms obtained in
\Cref{eq:dep_classical_ham}. This means that for any $T,T'$,
\[
\sum_{s \in \{0,1\}^{|\mc C_0|}}  \ketbra s s \otimes h^{(c,s)}_{T,T'} = \sum_{s' \in \{0,1\}^{|\mc C_0 \cap (S_T \cup S_{T'})|}} \ketbra{s'}{s'} \otimes h^{(c,s')}_{T,T'},
\]
and the resulting Hamiltonian is
\begin{equation}
    \label{eq:final_classical_ham}
    \tilde H^{(c)} = \sum_{T,T'} \sum_{s' \in \{0,1\}^{|\mc C_T \cup \mc C_{T'}|}} \ketbra{s'}{s'} \otimes h^{(c,s')}_{T,T'} + \sum_T \sum_{s' \in \{0,1\}^{|\mc C_T |}} \ketbra{s'}{s'} \otimes h^{(c,s')}_{T}.
\end{equation}
Since each tile $T$ is of constant size, each set $\mc C_T \cup \mc C_{T'}$ has constant size as well. This implies we
only get a constant blow-up in the number of terms of $H$, yielding a polynomial-size instance. The locality of the
instance depends on the max size of the set $\mc C_T \cup \mc C_{T'}$. Assuming this is bounded by $k$, we obtain an $2
+ k$-local classical Hamiltonian.

\vspace{1em}
\noindent \textbf{Correcting Removed Terms $\mc R$}. As in the original proof, a ground state $\ket \psi$ for $\tilde
H^{(c)}$ can be constructed in \textsf{NP}. To obtain a ground state for the original Hamiltonian $H$, we need to
correct the terms removed terms $\mc R$. The first step is to remove the classical qubits. We do this by successively
measuring the classical qubits to obtain a terminating restriction $\bs s = (s_0, \dots, s_\ell)$. 
\begin{lemma}
    \label{lem:ground_state_restriction}
    Suppose the above operation yields a series of assignments $\bs s = (s_0, \dots, s_\ell)$, taking $\ket \psi$ to $\ket{\phi} =
    \ket{s_0, \dots, s_\ell} \ket{\phi'}$. Then, $\ket{\phi}$ is a ground state for the Hamiltonian
    $\tilde H|_{\bs s}$.
\end{lemma}
\begin{proof}
    We show by induction. By assumption $\ket \psi$ is a ground state of $\tilde H^{(c)}$. Suppose this holds for
    $s_0,\dots, s_i$, with corresponding ground state $\ket{\psi'}$ of $H' := \tilde H|_{s_0, \dots, s_i}$. By
    definition $\mc C' := \mc C_{i}(s_0,\dots, s_{i-1})$ are the classical qubits of $H'$. This implies that (after some
    single-qubit unitary transformation) each term $h \in H'$ can be written as $h  =\sum_{s \in \{0,1\}^{|\mc C'|}}
    \ketbra s s_{\mc C'} \otimes \tilde h_s$, and thus $H'$ takes the form
    \[
        \sum_{s \in \{0,1\}^{|\mc C'|}} \ketbra s s_{\mc C'} \otimes \left(\sum_{h \in H'} \tilde h_s\right) = 
        \sum_{s \in \{0,1\}^{|\mc C'|}} \Pi_s \otimes \left(\sum_{h \in H'} \tilde h_s\right)
    \]
    i.e., $H'$ is \emph{block diagonal} w.r.t. the qubits in $\mc C'$. Thus, if $\ket{\psi'}$ is a ground state of $H'$,
    then,
    \[
        0 = \bra{\psi'} H' \ket{\psi'} = \sum_{s \in \{0,1\}^{|\mc C'|}} \bra{\psi'}  \Pi_s H' \Pi_s \ket{\psi'}
        \stackrel{\Pi_s H' \Pi_s \succeq 0}{\vphantom{\otimes}\iff} \forall s \bra{\psi'} \Pi_s H' \Pi_s \ket{\psi'} = 0.
    \]
    This implies that $\Pi_s \ket{\psi'}$ is a ground state of $H'|_{s}$ for any measurement outcome $s$.
\end{proof}

Therefore, the above procedure yields a Hamiltonian $H' = \tilde H|_{\bs s}$ and a ground state $\ket\psi$ of
$H'$. Now, recall that the removed term $h \in \mc R$ is an interior, fully quantum term of the Hamiltonian $H$.
\Cref{as:uniformly_correctable} implies that there is a correction unitary operator $L_h$ anti-commuting with
$h$ and commuting with every other term in $H'$. Thus, we may perform the measurement $\mc M = \{\tfrac 1 2(\cI + h),
\tfrac 1 2 (\cI -h)\}$ and we can apply $L_h$ to flip the measurement outcome if required.

\subsubsection{Gibbs state sampling}
\label{sec:proof_for_classical_gibbs}
In this section, we extend the previous proof so that rather than obtaining a single ground state of $H$, we obtain a
sample from the Gibbs distribution. Initially, we apply exactly the same steps as in the previous section, until we
obtain a classical Hamiltonian
\begin{equation}
    \tilde H^{(c)} = \sum_{T,T'} \sum_{s' \in \{0,1\}^{|\mc C_T \cup \mc C_{T'}|}} \ketbra{s'}{s'} \otimes
    h^{(c,s')}_{T,T'} + \sum_T \sum_{s' \in \{0,1\}^{|\mc C_T |}} \ketbra{s'}{s'} \otimes h^{(c,s')}_{T}
\end{equation}
equivalent to the original Hamiltonian with terms $\mc R$ removed. Now, rather than obtaining a \emph{ground state} of
$\tilde H^{(c)}$, we assume that we can sample from the Gibbs distribution of $\tilde H^{(c)}$, and then argue that we
can recover the Gibbs distribution of the original Hamiltonian, by correcting for the terms $\mc R$. We will show how to
correct these terms inductively. Let $Q$ be the set of terms that were either not removed or already corrected, and let
$H^Q = \sum_{h\in Q} h$. Assume we have the Gibbs state  $\rho(Q)$ of $H^Q$
\[
    \rho(Q) := \frac 1 Z \sum_{s, \lambda^Q_s} e^{-\beta \lambda^Q_s} \ketbra s s \otimes \Pi^Q_{\lambda_s},
\]
where we use the fact that the Hamiltonian $H^Q$ is diagonal w.r.t. $s \in \{0,1\}^{|\mc C_0|}$. To correct terms $h \in \mc R$, we use the same observation from \Cref{lem:ground_state_restriction} that given some classical restriction $s$, the resulting Hamiltonian $H|_s$ is block-diagonal w.r.t. the qubits in $\mc C_1(s)$. Thus, we may assume we see the state
\[
    \rho^{(s,t)}(Q) = \sum_{\lambda^Q_{s,t}} e^{-\beta \lambda^Q_{s,t}} \ketbra{s,t}{s,t}\Pi^Q_{\lambda_{s,t}}
\]
with probability proportional to $\sum_{\lambda^Q_{s,t}} e^{-\beta \lambda^Q_{s,t}}$. Here, $\Pi^Q_{\lambda_{s,t}}$ is
the projector onto the $\lambda^Q_{s,t}$-eigenspace of $H^Q|_{s,t}$. The idea is to apply the proof of
\Cref{lem:reduction} within the subspace $\ketbra{s,t}{s,t} \otimes \cI$. Specifically, we make the following observations.
\begin{itemize}
    \item \Cref{lem:2eig} holds since each $h \in \mc R$ is a fully quantum term (for any classical restriction, $h$ is in the interior).
    \item By \Cref{as:uniformly_correctable} there is a correction operator $L_h$ anti-commuting with $h$ and commuting with all other terms of $H^Q|_{s,t}$, for any choice of $s,t$. This assumption means that $L_h$ itself is only defined over non-classical qubits and thus for a particular choice of $s,t$, we can think of this operator as being $\ketbra{s,t}{s,t} \otimes L_h$. Then, \Cref{lem:sym} follows by tensoring both sides of \Cref{eq:sym1,eq:sym2} with $\ketbra{s,t}{s,t}$.
    \item Similarly, \Cref{lem:dim} holds, again by tensoring both matrices with $\ketbra{s,t}{s,t}$.
\end{itemize}
These observations plus the proof of \Cref{lem:reduction} implies that we can prepare the state,
\[
    \hat \rho(Q \cup \{h\}) = \frac 1 {Z(Q \cup \{h\})}\sum_{\lambda^Q_{s,t}} \sum_{\lambda_h \in \{\pm c_p\}} \exp(-\beta \lambda^Q_{s,t}) \cdot \exp(-\beta \lambda_h) \ketbra{s,t}{s,t}\Pi^{Q \cup \{h\}}_{\lambda^Q_{s,t} + \lambda_h}
\]
It's easy to see that \emph{within} the subspace $\ketbra{s,t}{s,t}$ this is the correct Gibbs state. Moreover, across different $(s,t), (s',t')$ pairs, we retain the correct distribution. This is because correcting additional fully quantum terms $h$ only applies an identical multiplicative factor across each of the $\ketbra{s,t}{s,t}$-subspaces.

\section{Acknowledgement}
We thank Dominik Hangleiter, Sandy Irani,  Anurag Anshu, and Yongtao Zhan for the helpful discussion.

YH is supported by the National Science Foundation Graduate Research Fellowship under Grant No. 2140743; YH is also
supported by NSF grant CCF-2430375. Any opinion, findings, and conclusions or recommendations expressed in this material are those of the authors(s) and do not necessarily reflect the views of the National Science Foundation. Jiaqing Jiang is supported by MURI Grant FA9550-18-1-0161 and the IQIM, an     NSF Physics Frontiers Center (NSF Grant PHY-1125565).

\appendix
\section{Gibbs sampling reduction for punctured defected Toric code}\label{appendix:DTC}

As described in the proof overview section (Section \ref{sec:defected_toric_code}),  our   Gibbs sampler gives an $\mc O(n^2)$-time algorithm  for preparing  the Gibbs states of the punctured defected Toric code $H_{DT}$. In this section we instead describe a different and slower   Gibbs sampler for $H_{DT}$, which  achieves the following Claim \ref{claim:TC}
and captures the key ideas for our  Gibbs sampler for  general qubit 2D CLHs.

\tikzset{
    pics/qudits/.style n args={3}{code={
        \node[plaquette,draw=black!40] (q2#3) at ($ (#1,#2) + (.2,.8) $) {$q_2$};
        \node[plaquette,draw=black!40] (q4#3) at ($ (#1,#2) + (.2,.2) $) {$q_4$};
        \node[plaquette,draw=black!40] (q3#3) at ($ (#1,#2) + (.8,.8) $) {$q_3$};
        \node[plaquette,draw=red, line width=0.4mm] (q1#3) at ($ (#1,#2) + (.8,.2) $) {$q_1$};
}}}
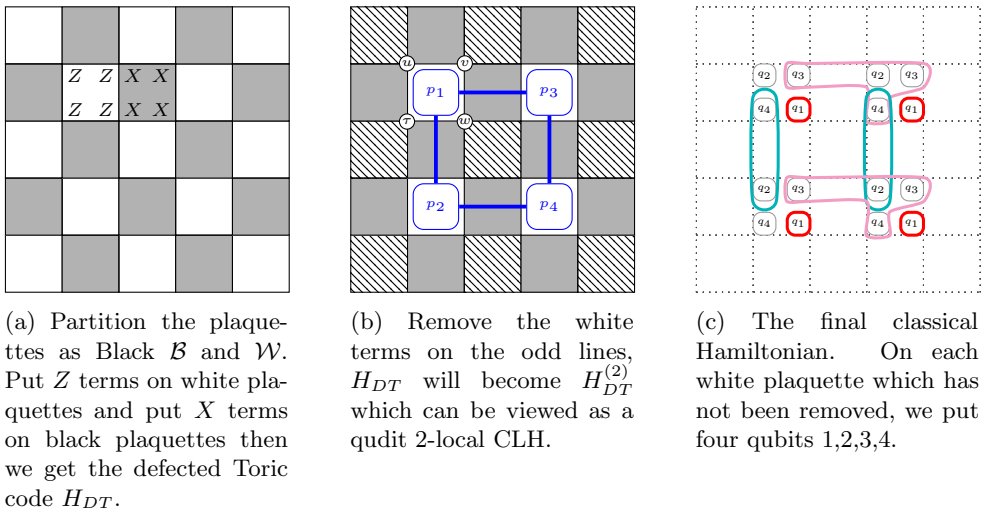
\begin{figure}[h!]
    \centering
    \begin{subfigure}[t]{0.25\textwidth}
        \centering
        \begin{tikzpicture}[scale=0.75,every label/.append style={font=\scriptsize}]
            \foreach \x in {0,...,4} {
                \foreach \y in {0,...,4} {
                    \ifthenelse{\equal{\intcalcMod{\x+\y}{2}}{1}}{
                        \draw[fill=black!30] (\x,\y) rectangle ++(1,1);
                    }{
                        \draw (\x,\y) rectangle ++(1,1); 
                    }
                }
            }
            \path (1,3) rectangle ++(1,1) node[fitting node] (ztype) {};
            \path (2,3) rectangle ++(1,1) node[fitting node] (xtype) {};
            \node[label={[label distance=-0.25cm]below left:$Z$}] at (ztype.north east) {};
            \node[label={[label distance=-0.25cm]below right:$Z$}] at (ztype.north west) {};
            \node[label={[label distance=-0.25cm]above right:$Z$}] at (ztype.south west) {};
            \node[label={[label distance=-0.25cm]above left:$Z$}] at (ztype.south east) {};
            
            \node[label={[label distance=-0.25cm]below left:$X$}] at (xtype.north east) {};
            \node[label={[label distance=-0.25cm]below right:$X$}] at (xtype.north west) {};
            \node[label={[label distance=-0.25cm]above right:$X$}] at (xtype.south west) {};
            \node[label={[label distance=-0.25cm]above left:$X$}] at (xtype.south east) {};
        \end{tikzpicture}
        \caption{Partition the plaquettes as Black $\cB$ and $\cW$. Put $Z$ terms on white plaquettes and put $X$ terms on black plaquettes then we get the defected Toric code $H_{DT}$.}
        \label{fig:fig2a}
    \end{subfigure}
    \qquad
    \begin{subfigure}[t]{0.25\textwidth}
        \centering
        \begin{tikzpicture}[
            scale=0.75,
            every node/.append style={font=\tiny},
            every label/.append style={font=\scriptsize},
            qubitlabel/.style={inner sep=0, draw, circle, minimum width=0.2cm, fill=white},    
            plaquette/.style={draw, color=blue, rounded corners, rectangle, minimum width=0.6cm, minimum height=0.6cm,inner sep=0mm}
        ]
            \foreach \x in {0,...,4} {
                \foreach \y in {0,...,4} {
                    \ifthenelse{\equal{\intcalcMod{\x+\y}{2}}{1}}{
                        \draw[fill=black!30] (\x,\y) rectangle ++(1,1);
                    }{
                        \ifthenelse{\equal{\intcalcMod{\y}{2}}{0}}{\draw[pattern=north west lines, pattern color=black] (\x,\y) rectangle ++(1,1);}{}
                    }
                }
            }
            \node[plaquette] (p1) at (1.5,3.5) {$p_1$};
            \node[plaquette] (p2) at (1.5,1.5) {$p_2$};
            \node[plaquette] (p3) at (3.5,1.5) {$p_4$};
            \node[plaquette] (p4) at (3.5,3.5) {$p_3$};

            \path (1,3) rectangle ++(1,1) node[fitting node] (ztype) {};
            \node[qubitlabel] at (ztype.north west) {$u$};
            \node[qubitlabel] at (ztype.north east) {$v$};
            \node[qubitlabel] at (ztype.south west) {$\tau$};
            \node[qubitlabel] at (ztype.south east) {$w$};

            \draw[color=blue, line width=0.5mm] (p1.east) -- (p4.west);
            \draw[color=blue, line width=0.5mm] (p1.south) -- (p2.north);
            \draw[color=blue, line width=0.5mm] (p4.south) -- (p3.north);
            \draw[color=blue, line width=0.5mm] (p2.east) -- (p3.west);
        \end{tikzpicture}
        \caption{Remove the white terms on the odd lines, $H_{DT}$ will become $H_{DT}^{(2)}$ which can be viewed as a qudit 2-local CLH.}
        \label{fig:fig2b}
    \end{subfigure}
    \qquad
    \begin{subfigure}[t]{0.25\textwidth}
        \centering
        \begin{tikzpicture}[
            scale=0.75,
            every node/.style={transform shape},
            every node/.append style={font=\tiny},
            every label/.append style={font=\scriptsize},
            qubitlabel/.style={inner sep=0, draw, circle, minimum width=0.2cm, fill=white},    
            plaquette/.style={draw, draw=blue, fill=white, rounded corners=3, rectangle, minimum width=0.4cm, minimum height=0.4cm,inner sep=0mm}
        ]
            \foreach \x in {0,...,5} {
                \foreach \y in {0,...,5} {
                    \draw[dotted, color=black!70] (0,\y) -- (5,\y);
                    \draw[dotted, color=black!70] (\x,0) -- (\x,5);
                }
            }
            
            \pic {qudits={1}{3}{1}};
            \pic {qudits={3}{3}{2}};
            \pic {qudits={3}{1}{3}};
            \pic {qudits={1}{1}{4}};
            
            \draw [Lavender, line width=0.4mm] plot [smooth cycle,tension=0.4] coordinates {(q31.north west) (q31.north east) (q32.north east) (q32.south east) (q42.north east) (q42.south east) (q42.south west) (q22.south west) (q31.south east) (q31.south west)};
            \draw [BlueGreen, line width=0.4mm] plot [smooth cycle,tension=0.4] coordinates {(q42.north west) (q42.north east) (q23.south east) (q23.south west)};
            \draw [BlueGreen, line width=0.4mm] plot [smooth cycle,tension=0.4] coordinates {(q41.north west) (q41.north east) (q24.south east) (q24.south west)};
            \draw [Lavender, line width=0.4mm] plot [smooth cycle,tension=0.4] coordinates {(q34.north west) (q34.north east) (q33.north east) (q33.south east) (q43.north east) (q43.south east) (q43.south west) (q23.south west) (q34.south east) (q34.south west)};
        \end{tikzpicture}
        \caption{The final classical Hamiltonian. On each white plaquette which has not been removed, we put four qubits 1,2,3,4. }
        \label{fig:fig2c}
    \end{subfigure}
    \caption{\Cref{fig:intro_defected} is recreated here for convenience, with more detail on the resulting classical Hamiltonian $H^{(2c)}_{DT}$.}
\end{figure}

We recall the defected Toric code, $H_{DT}$. As shown in \Cref{fig:fig2a}, imagine  partitioning the plaquettes in the 2D lattice as Black  $\cB$  and White $\cW$.
The \textit{defected Toric code} is putting $Z$ and $X$ terms in white and black plaquetes respectively:
\begin{align}
		&H_{DT} =\sum_{p\in \cB} c_pX^p + \sum_{p\in \cW} c_pZ^p, 
\end{align}
where $c_p$ can be an arbitrary real number.  In the standard Toric code $c_p=-1$ for any $p$.  One can check that $H_{DT}$ is a qubit 4-local CLH on 2D. 
 
\begin{claim}\label{claim:TC} 
 	For any inverse temperature $\beta$, if one can do classical Gibbs sampling  with respect to $H_{DT}^{(2c)},\beta$ within precision $\epsilon$ in classical time $T$, then one can prepare the quantum Gibbs state  with respect to $H_{DT},\beta$ within precision $\epsilon$ in quantum time $T + \mc O(n^2)$.
\end{claim}
In this section we describe the algorithm and the reduction in details, while as the correctness proof is omitted, as is the same as in Section \ref{sec:qubit2D}. The algorithm will require removing a set $\mc R \subseteq \mc W$ of terms from $H_{DT}$. In particular, if we remove alternating rows of white terms (as in \Cref{fig:fig2b}) then treat the $4$ qubits in the support of each remaining white term as a \emph{single grouped qudit}, then the resulting Hamiltonian $H^{(2)}_{DT}$ is $2$-local. We'll denote the remaining white terms as $\mc O = \mc W \setminus \mc R$.

\paragraph{Notation.}
For simplicity, here we assume the 2D lattice is a square $L\times L$ lattice embedded on a plain which has boundaries. The number of qubits is $n=L\times L$. For better illustration, here we denote the computational basis as $\ket{x}\in\{\pm 1\}^n$ rather than the conventional notation $\ket{x}\in \{0,1\}^n$, where the Pauli $X$ and $Z$ operators act as
\begin{align}
    &Z \ket{1}=\ket{1},  Z \ket{-1}=-\ket{-1}\\
    & X \ket{1} = \ket{-1},  X \ket{-1} = \ket{1}.
\end{align}

This notion is only used in this section. We say $\ket{x}$ is of even Hamming weight if there are even 1s in $x$.

\paragraph{The classical Hamiltonian $H_{DT}^{(2c)}$.}\label{par:appendix_classical_ham} 
To map $H^{(2)}_{DT}$ to  the classical Hamiltonian $H_{DT}^{(2c)}$, we change the basis of the qudit. More specifically, for any term $p\in \cO$, consider the 4 qubits on $p$, as shown in \Cref{fig:fig2b}, name them as $u,v,w,\tau$ . A natural basis for the 4 qubit Hilbert space is the computational basis, which are common eigenvalues 
of $\{Z_u\}_{u\in p}$.  An alternative labeling is choosing another set of 4 independent stabilizers  
\begin{align}
	S^p_1:=Z^p, S^p_2:=X_u\otimes X_v, S^p_3:=X_v\otimes X_w, S^p_4 :=  X_w\otimes X_\tau,
\end{align}
 The 4 new stabilizers will specify  a new basis for the 4-qubit Hilbert space
 indexed by $s\in\{\pm 1\}^4$, 
 which can be viewed as the computational basis for 4 virtual qubits.
The 4 new stabilizers act as Pauli Z on the 4 virtual qubits. That is if we use $(Z_{*i})^p$ to denote Pauli Z on virtual qubit $i$ in plaquette $p$, then
$$
(Z_{*i})^p = S^p_i.
$$
Each $s\in\{\pm 1\}^4$
 indicates the common eigenvector of $S^p_i$ w.r.t eigenvalues $s_i$. 
 Denote the corresponding vector  as $\ket{s^p}^*$ where $*$ means $\ket{s^p}^*$ is the computational basis for the virtual qubits rather than the original qubits.
As an example one can verify that $\ket{\psi^p}$ in Eq.~(\ref{eq:psip}) corresponds to $\ket{1111^p}^*$.  
One can check that in this virtual qubit basis, 
 we have that the Hermitian terms\footnote{With some abuse of notations, we use $p$ to denote both the plaquette and the Hermitian term on the plaquette.} can be re-phrase as
\begin{align}
	 &\forall p\in \cO, p = c_p\cdot  (Z_{*1})^p,
\end{align}
Besides, note that $X_{u}\otimes X_{\tau}=  S_2^p\times S_3^p\times S_4^p $, and 
$X^{\otimes 4} = X^{\otimes 2}\otimes X^{\otimes 2}$,  we have
\begin{align}
	&\text{For horizontal $p$ connecting $p_1,p_3$, }p = c_p \cdot (Z_{*3})^{p_1} \otimes \left(Z_{*2}\otimes  Z_{*3} \otimes Z_{*4} \right)^{p_3}.\\
	&\text{For vertical $p$ connecting $p_1,p_2$, }p = c_p \cdot (Z_{*4})^{p_1} \otimes (Z_{*2})^{p_2}.
\end{align}
Thus the resulting Hamiltonian is exactly  $H_{DT}^{(2c)}$ as shown in \Cref{fig:fig2c} and Eq.~(\ref{eq:classical}), where we omitting the subscript * for simplicity, 
\begin{align}
 	H_{DT}^{(2c)} : =\sum_{\text{horizontal $(p,p_i,p_j)$}} c_p \cdot (Z_3)^{p_i}\otimes (Z_2Z_3Z_4)^{p_j}  +  \sum_{\text{vertical $(p,p_i,p_j)$}} c_p \cdot (Z_4)^{p_i}\otimes (Z_2)^{p_j} +\sum_{p\in O} c_p Z_1^p
 	\label{eq:classical}
 \end{align}
In summary,  by choosing another basis each qudit can be viewed as four virtual qubits, and the 2-local Hamiltonian $H_{DT}^{(2)}$ becomes a classical Hamiltonian $H_{DT}^{(2c)}$ w.r.t those virtual qubits. 
The computational basis of the virtual qubits $\ket{\bm{y}}^*$ 
corresponds to an eigenvector of $H_{DT}^{(2)}$,
 denoted as $\ket{\phi(\bm{y})}$. Denote the  eigenvalue w.r.t $\ket{\bm{y}}^*$ and $p\in \cB\cup\cO$ as $\lambda_p(\bm{y})$, and define 
 $$\lambda_{\cB\cup\cO}(\bm{y}) :=\sum_{p\in \cB\cup\cO} \lambda_p(\bm{y}).$$

\subsection{Preparation of a ground state}
Before describing our Gibbs state sampler, we first review a folklore quantum algorithm for preparing  the ground state of the standard Toric code.

For ease of presentation, for any subset $Q\subseteq \cW$ or $Q\subseteq\cB$, we use ``terms in $Q$" to denote the operators $(-X^p)$ if $p\in \cB$, and $(-Z^p)$ if $p\in \cW$. Note that the ground state of the standard Toric code is the common $(-1)$-eigenvector of all terms in $\cW$ and $\cB$. The first step in preparing the ground state of $H_{DT}$ is preparing the ground state of $H^{(2)}_{DT}$, where
\[
    H_{DT}^{(2)}:=\sum_{p\in \cB\cup\cO} p
\]
is two local.

The ground state of $H_{DT}^{(2)}$ is easy to describe. For each term  $p\in  \cO$, denote $\ket{\psi^p}$ as the uniform superposition of basis states with even Hamming weights, that is 
 	\begin{align}
 		\ket{\psi^p} =\frac{1}{2}\sum_{x\in \{\pm 1\}^p, |x| \text{ even}} \ket{x^p},\label{eq:psip}
 	\end{align}
 where $\{\pm 1\}^p$ is the computational basis of the four qubits of $p$.   One can check that 
 \begin{align}
     \ket{\psi^{\cB\cup \cO}}:= \otimes_{p\in \cO} \ket{\psi^p},
 \end{align}
  is the common $(-1)$-eigenstate of all terms in $\cB\cup\cO$.
\begin{remark}
    The underlying reason that the common $(-1)$-eigenstates of plaquette terms in $\cB\cup \cO$ can be chosen as a product of constant-qudit state, lies in the fact that  the Hamiltonian
    $H_{DT}^{(2)}$  can be viewed as a 2-local qudit commuting Hamiltonian after grouping  the four qubits in $p\in \cO$ as a qudit. Thus by the Structure Lemma, one can prepare the ground state by constant depth quantum circuit.
\end{remark}
Next, we need to correct for the removed terms $p \in \mc R$. Starting with $\ket{\psi^{\cB\cup\cO}}$, we sequentially measure the current state w.r.t the measurement $-Z^p$, for each $p \in \cR$.
\begin{itemize}
  \item[(a)]  If we get $- 1$, we know the current state is the common $(-1)$-eigenstate for all terms in $\cB\cup\cO\cup\{p\}$ and we move to the next $p\in\cR$. 
  \item[(b)] Otherwise we perform a deterministic correction: as shown in \Cref{fig:intro_fig_correct}, we can connect term $p$ to the boundary by a path $\gamma_p$. The deterministic correction is done by applying a sequence of $X$ operator, denoted as 
  \begin{align}
      L_{p}=\otimes_{v\in \gamma} X_v. \label{eq:correction}
  \end{align} 
\end{itemize}
where $X_v$ is the Pauli X on qubit $v$. Note that $L_p$ anti-commutes with $p$, and commutes with all other terms. Thus the final state will again be the common $(-1)$-eigenstate for all terms in $\cB\cup\cO\cup\{p\}$. 

\begin{remark}
\label{rem:corrected_state}
For a state $\ket \psi$ before the measurement for $p \in \mc R$ is applied, we end up with a $-1$-eigenstate whether we fall in the first or second case above. However, note that the exact states we obtain in each case may not be equal. This contributes to the difficulty of obtaining a Gibbs sampler.
 \end{remark}
 
In summary, to prepare the ground state, one first prepares the ground state of a 2-local Hamiltonian 
$\sum_{p\in \cB} -X^p + \sum_{p\in \cO} -Z^p.$
Then performs  a \textit{deterministic correction} to modify the ground state of the 2-local Hamiltonian to the ground state of the Toric code. The ground state preparation for the defected Toric code is similar.

\subsection{Preparation of the Gibbs state}
We now describe how to generalize the ground state preparation algorithm to prepare Gibbs state  for the defected Toric code. At a high level, the idea is to:
\begin{itemize}
    \item Step 1: First prepare the \textit{Gibbs state } of the  Hamiltonian 
$$H_{DT}^{(2)}:=\sum_{p\in \cB\cup\cO} p.$$
by mapping it to the \textit{classical 2-local} Hamiltonian $H_{DT}^{(2c)}$ defined in Eq.~(\ref{eq:classical}). 
    \item  Step 2:  Perform a \textit{randomized correction}, to modify the current Gibbs state for $H_{DT}^{(2)}$ to the final Gibbs state  for $H_{DT}$.
\end{itemize}

\paragraph{Step 1: Gibbs state for $H_{DT}^{(2)}$.} As in \Cref{fig:fig2b}, $H^{(2)}_{DT}$ is 2-local in the sense that:  if we group every 4 qubits in a plaquette $p\in \cO$ as one qudit, then every $p\in \cO$ only acts on one qudit, every $p\in \cB$ acts on two qudits. 
  
Moreover, since we only change local basis for every qudit, $\ket{\phi(\bm{y})}$ is in fact a tensor product of single-qudit state. Thus we can prepare the quantum Gibbs state  of $H_{DT}^{(2)}$ by a simple algorithm: do classical Gibbs sampling for $H_{DT}^{(2c)}$, get a string $\bm{y}$, and prepare the tensor product state $\ket{\phi(\bm{y})}$. 

\paragraph{Step 2: Correcting removed terms.} Step 2 is more tricky. By Step 1, we assume that we can sample the string $\bm{y}$ with probability  $\exp(-\beta \lambda_{\cB\cup\cO}(\bm{y}))/Z_{\cB\cup\cO}$, where $Z_{\cB\cup\cO}$ is the normalization factor.

Then we
take a term in the removed set $p\in \cR$,  and try to prepare the Gibbs state  w.r.t. $\left(H_{\cB\cup\cO\cup\{p\}},\beta\right)$. Consider measuring $\ket{\phi(\bm{y})}$ w.r.t $p$. Denote the measurement outcome as $\lambda \in\{\pm c_p\}$. Denote the projectors to be 
\begin{align}
	\Pi^{p}_{+c_p} : = \frac 1 2\left(I+\frac{p}{c_p}\right), \quad
 \Pi^{p}_{-c_p} : = \frac 1 2 \left(I-\frac{p}{c_p}\right).
\end{align}

We have 
\begin{align}
 	 Pr\left( \text{outcome is } \lambda \right) = 	\exp(-\beta \lambda_{\cB\cup\cO}(\bm{y}))/Z \cdot \langle \phi(\bm{y})| \Pi^{p}_{\lambda} | \phi(\bm{y})\rangle. \label{eq:prac} 
\end{align}
 Recall that at this moment, the ideal distribution we want is 
 \begin{align}
 	Pr_{ideal} \left(  \text{outcome is }\lambda\right) \text{ is proportional to } \exp(-\beta \lambda_{\cB\cup\cO}(\bm{y})) \cdot \exp(-\beta \lambda).\label{eq:ideal}
 \end{align}
Compared to the preparation of \textit{ground state}, which only needs to correct the eigenvalue, in the task of Gibbs state  preparation, to get Eq.~(\ref{eq:ideal}) from Eq.~(\ref{eq:prac}), it seems that one needs to correct 
\begin{itemize}
    \item The probability incurred by measurement, that is $\langle \phi(\bm{y})| \Pi^{p}_{\lambda} | \phi(\bm{y})\rangle$.
    \item  The probability incurred by the new energy $\exp(-\beta \lambda)$.
\end{itemize}
 This first correction is  tricky because it depends on $\phi(\bm{y})$, thus the probability to be corrected is different for every $\bm{y}$. We circumvent this problem by an intuitive oblivious \textit{randomized correction}. That is if we get outcome  $\lambda$  then
 \begin{itemize}
     \item  With probability $prob:=\frac{\exp(-\beta \lambda)}{\exp(\beta \lambda)+\exp(-\beta \lambda)}$ we do nothing.
     \item  With probability $1-prob$ we apply the correction operation $L_{p}$ defined in Eq.~(\ref{eq:correction}).
 \end{itemize}
  
One may doubt whether the above algorithm successfully prepares the Gibbs state  or not, since the correction only brings us back to the right eigenspace, but does not really help us get back to the right state, as noted in \Cref{rem:corrected_state}. That is 
\begin{align}
	& L_{p} \Pi^p_{+c_p}   \ket{\phi(\bm{y})} \not\propto \Pi^p_{-c_p}  \ket{\phi(\bm{y})},
\end{align}
where here $\not\propto$ means the two vectors are not proportional to each other. In fact, $L_p \Pi_{+c_p}   \ket{\phi(\bm{y})}$ might not even be (proportional to) one of the eigenstates $ \{ \Pi^p_{+c_p}\ket{\phi(\bm{y})}\}_{\bm{y}}\cup  \{ \Pi^p_{-c_p}\ket{\phi(\bm{y})}\}_{\bm{y}}$. 
The key fact that makes this oblivious randomized correction work, is the symmetry of the eigenspace and the fact that $L_p$ is a unitary. More specifically, let $\Pi^{\cB\cup\cO}_{\bm{y}}$ be the common eigenspace of $p'\in \cB\cup\cO$ w.r.t eigenvalue $\lambda_{p'}(\bm{y})$. Then
 \begin{align}
     L_p \Pi^{p}_{+c_p} \Pi^{\cB\cup\cO}_{\bm{y}} \Pi^{p}_{+c_p} L_p = \Pi^{p}_{-c_p} \Pi^{\cB\cup\cO}_{\bm{y}}\Pi^{p}_{-c_p}.
 \end{align}
Finally, note that all $p\in \cR$ can be corrected independently. Thus it suffices to perform the same measure and randomized correction procedure sequentially.

For general qubit 2D CLH $H$ without  classical qubits,  denote the removable terms $\cR$ as the set of terms whose measurement value can be corrected without changing the value of other terms. We prepare the Gibbs state  of $H$ in a similar way as for the defected Toric code.  That is we first remove  terms in $\cR$ and  transform the remaining 2-local Hamiltonian $H^{(2)}$ to a classical Hamiltonian $H^{(2c)}$. Then we first do classical Gibbs sampling for $H^{(2c)}$, then prepare the Gibbs state  of $H^{(2c)}$, then perform measurement and randomized correction to prepare the Gibbs state  of $H$.

\section{Gibbs Sampler for defected Toric Code embedded on a closed 2D surface}
\label{appendix:non_punctured}
As described in \Cref{sec:TCnn}, the oblivious correction algorithm works only when there are punctured terms that act
as a boundary in the  defected Toric code.
If the defected Toric code is defined on a torus—as in the standard formulation in the error-correcting code literature—then correction operators correspond to strings of Pauli operators connecting two measured terms of the Hamiltonian, and the oblivious correction approach does not applies directly. 
Nonetheless, we show that there is a simple algorithm which prepares the Gibbs
state of such defected Toric code on torus. 
One can similarly prepare the Gibbs state for the defected Toric code embedded on other closed 2D surface, such as a sphere, or for the models defined in \cite{aharonov2018complexity}.

As before, we partition the terms into ``white'' terms $\mc W$ and
``black'' terms $\mc B$, which correspond to $Z$ and $X$-type terms respectively. 

The algorithm proceeds in two phases.
In the first phase, we prepare the Gibbs state for the $Z$-type terms, i.e.,
\begin{equation}\label{eq:z_gibbs}\rho(H_{\mc W}, \beta) = \sum_{\lambda \in \Lambda_{\mc W}} \mu_{\mc W}(\lambda)
\Pi_{\lambda}.\end{equation}
Here $H_{\mc W} = \sum_{p \in \mc W} p$, and $\Lambda_{\mc W}$ denote the set of length-$|\mc W|$ vectors  with entries $\pm 1$ whose product equals 1, that is, $\lambda=(...,\lambda_p,...)_{p\in W}$ where $\lambda_p\in\{+1,-1\}$ and $\prod_{p\in \mc W} \lambda_p =1$.  The vector $\lambda$ is used to label the eigenspaces of $H_{\mc W}$, denoted as $\Pi_\lambda$, which satisfies $\forall p\in \mc W,  p \Pi_\lambda =\lambda_p \Pi_\lambda.$  The $\mu_{\mc W}(\lambda)$ is  the corresponding Boltzmann weights, 
     $$\mu_{\mc W}(\lambda)= \prod_{p\in \mc W} exp(-\beta \lambda_p c_p)/Z_{\mc W} \text{ where  $Z_{\mc W}= \sum_{\lambda \in \Lambda_{\mc W}} \prod_{p\in \mc W} exp(-\beta \lambda_p c_p)$}.$$
In the second phase, begin with $\rho(H_{\mc W}, \beta)$ we further consider the $X$-type terms, 
prepares the full states for $H=H_{\mc W}+H_{\mc B}$, that is
\begin{equation}\rho(H, \beta) = \sum_{\lambda \in \Lambda_{\mc W}} \mu_{\mc W}(\lambda)
\Pi_{\lambda} \sum_{\lambda'\in    \Lambda_{\mc B}}  \mu_{\mc W}(\lambda') \Pi_{\lambda'}\,, \label{eq:Hbeta}\end{equation}
Note that \cref{eq:Hbeta} holds since the correction operators for $X$ terms  commute with the correction operators for $Z$ terms, thus the correction operators commute with each $\Pi_\lambda$.

More specifically, for both phases, we use 1D Ising chains to guide the correction operations.  
For the first phase,
we construct the clasiscal Ising chain 
\begin{align}
H_\text{Ising} = \sum_{i=1}^k c_i h_i,\label{eq:Toric_Ising}
\end{align}
with $h_i = Z_i Z_{i+1}$, $k = |\mc W|$, $k+1$
identified with $1$, and the coefficients $c_i$ will be  specified later.   Start with the uniform distribution over bit strings. In each step of Glauber dynamics for a state
$\ket x$, a site $i \in [k]$ is picked randomly, and the change in energy 
$\Delta E = \text{Tr}\left[H_{\mathrm{Ising}} (\ketbra x x - X_i \ketbra x x X_i)\right]$
is computed.
The bit flip $X_i$ is
accepted with probability $\min\{1,exp(-\beta \Delta E)\}$.  To map this Gibbs sampling for Ising model onto Gibbs states preparation for $H_{\mc W}$, we order the terms $p \in \mc W$ in any arbitary
order $p_1, \dots, p_i, \cdots, p_{k}$, with $k = |\mc W|$. For convenience, we also write the corresponding coefficient $c_p$ in Eq.~(\ref{eq:HDT}) as $c_i$.
Note that each $p_i\in \mc W$ in the  Toric code will correspond to $h_i$ in the aforementioned Ising model, i.e. Eq.~(\ref{eq:Toric_Ising}).  Imagine the strings $x \in \{0,1\}^k$ as labeling the
eigenvalues of $(p_1, \dots, p_k)$ via
\begin{equation}
    \label{eq:eig_mapping}
    \lambda_i = (-1)^{x_i \oplus x_{i+1}}\,.
\end{equation}
We argue that this mapping is valid. $H_{\mc W}$'s eigenspaces are labeled by $\lambda \in \{\pm 1\}^k$ with the parity
constraint $\prod_{i \in [k]} \lambda_i = 1$. This constraint is obeyed by \Cref{eq:eig_mapping}:
\begin{align}
    \prod_{i=1}^k \lambda_i = (-1)^{x_1 \oplus x_2} \cdot (-1)^{x_2 \oplus x_3} \dots (-1)^{x_{k-1} \oplus x_k} \cdot
    (-1)^{x_k \oplus x_1} = (-1)^{x_1 \oplus x_k} \cdot (-1)^{x_k \oplus x_1} = 1\,.\label{eq:TC_prod1}
\end{align}
Moreover, the update rule of Glauber dynamics maps cleanly to $H_{\mc W}$. The classical bit flip $X_i$ is mapped to the
string operator $X^s$, where $s$ is a path from $p_i$ to $p_{i+1}$ in the 2D Toric code lattice. This operator is defined so that it anti-commutes
with $p_{i-1}$ and $p_{i}$ and commutes with all other terms $p \in \mc W$, and commutes with all terms in $\mc B$. This retains the equivalency in
\Cref{eq:eig_mapping} since applying $X_i$ to the $i$-th bit toggles the parity, and the operation takes $\lambda_i
\rightarrow -\lambda_i$ and $\lambda_{i-1} \rightarrow -\lambda_{i-1}$. This is exactly the effects of applying the correction operators $X^s$ on the Toric code. One can see from \Cref{eq:eig_mapping} that the Gibbs distribution of $H_\text{Ising}$ corresponds to
\Cref{eq:z_gibbs} since 
\begin{align}
    \rho(H_\text{Ising}, \beta) &= \frac 1 Z \sum_{x \in \{0,1\}^k} \exp(-\beta \text{Tr}[H_\text{Ising} \ketbra x x]) \ketbra x
    x\\
     &= \frac 1 Z \sum_{x \in \{0,1\}^k} \exp\left(-\beta \sum_i  c_i (-1)^{x_i \oplus x_{i+1}}\right) \ketbra
                                    x x\\
     &= \frac 1 Z \sum_{\lambda \in \{\pm 1\}^k} \exp\left(-\beta c_i \lambda_i \right) \ketbra
                                    x x\\
\end{align}
where each $x$  label the eigenspace projector  $\Pi_\lambda$ in $\rho(H_{\mc W},\beta)$.

Finally, to prepare the full Gibbs state for \( H = H_{\mathcal{W}} + H_{\mathcal{B}} \), we begin with the state \( \rho(H_{\mathcal{W}}, \beta) \), and repeat the above mapping—this time with respect to the \( X \)-type terms instead of the \( Z \)-type terms. Specifically, starting from \( \rho(H_{\mathcal{W}}, \beta) \), we run a classical Glauber dynamics corresponding to the terms in \( \mathcal{B} \), and use this  Ising dynamics to guide the correction procedure for the \( X \)-type terms on the state \( \rho(H_{\mathcal{W}}, \beta) \).

Note that Gibbs sampling for 1D Ising model is known to be rapid mixing and the corresponding Gibbs samplers have runtime $O(n \,poly(\log n))$ ~\cite{guionnet2003lectures,holley1985rapid,holley1989uniform,zegarlinski1990log}. Thus the 
 runtime of our algorithm is $O(npoly(log n) \times n)=\TCtime$, where $n$ is the cost for each correction operation.

Finally, note that the algorithm described above for preparing the Gibbs state of the defected Toric code on a torus can be directly generalized to defected Toric code defined on  any closed 2D  surfaces. This is because the algorithm does not rely on any special properties of the torus, but only on the constraint that all syndromes of the same type multiply to one, and on the fact that in the defected toric code the value of a syndrome can be flipped by string operators.

\section{Proofs of reductions for specific Hamiltonians}\label{appendix:specific}

\begin{proof}[Proof of Lemma \ref{lem:H1D}]

By the argument in Section \ref{sec:2local_specific}, we can always assume that $H_{1D}$ is 2-local. To make the notation consistent with Section \ref{sec:2local} we rename $H_{1D}$ as $H_{1D}^{(2)}$.
Then according to Section \ref{sec:2local} we can construct a 1D qudit classical 2-local Hamiltonian $H_{1D}^{(2c)}$. 

By assumption, we assumed that the dimension of qudit is $d=2^k$ for some $k$, then every qudit can be viewed as $k$ qubits and the computational basis can be written as $\{0,1\}^k$.
Note that with an arbitrary ordering of the $k$ qubits, we can view the $n$ qudits on 1D chain as $nk$ qubits on 1D chain. 
 Let $Z$ be the Pauli Z operator, note that
\begin{align}
    \ketbra 0 0 = \tfrac 1 2 (Z+I), \quad\ketbra 1 1 =\tfrac 1 2 (I-Z).
\end{align}
Thus we can rewrite classical Hamiltonian $H_{1D}^{(2c)}$ as sum of tensors of $Z$. that is
\begin{align}
    H_{1D}^{(2c)} = \sum_{S} a_S Z^S,
\end{align}
where $S$ are subset of qubits,   and $a_S=0$ if $diam(S)\geq 2k$, where $diam$ is the diamater of $S$, that is the farthest distance between qubits in $S$ w.r.t the 1D chain. 
$Z^S$ is the tensor product of Pauli $Z$ on qubits in S.

Besides, since $H_{1D}$ is translation-invariant, so does $ H_{1D}^{(2c)}$. Combined with the rapid mixing Gibbs sampler for   1D finite-range, translation-invariant Ising model~\cite{guionnet2003lectures,holley1985rapid,holley1989uniform}, we complete the proof.
\end{proof}

The proof of  Lemma \ref{lem:H2D} involves the notations of induced algebra, which should be read after reading section \ref{sec:C}.
\begin{proof}[Proof of Lemma \ref{lem:H2D}]

    Let $h$ be the translation-invariant term in the 2D Hamiltonian which acts on two systems (qubits) $a,b$. Let $\cA^a_h$ and $\cA^b_h$ as the induced algebra of $h$ on systems $a$ and $b$ respectively.
 \begin{figure}[h!]
\centering
    \begin{tikzpicture}[
        every label/.append style={font=\footnotesize},
    ]
        \node[draw, rectangle, rounded corners=1, minimum width=0.3cm, minimum height=0.3cm, inner sep = 0] (q) at (0,0) {\footnotesize $q$};
        \node[draw, fill, circle, minimum width=0.1cm, inner sep = 0, label={above right:$a$}] (eo) at ($ (q.east) + (0.2, 0) $) {};
        \node[draw, fill, circle, minimum width=0.1cm, inner sep = 0, label={above:$b$}] (et) at ($ (q.east) + (1, 0) $) {};
        \draw (eo) -- (et);
        \node[draw, fill, circle, minimum width=0.1cm, inner sep = 0, label={above:$b$}] (wo) at ($ (q.west) + (-0.2, 0) $) {};
        \node[draw, fill, circle, minimum width=0.1cm, inner sep = 0, label={above:$a$}] (wt) at ($ (q.west) + (-1, 0) $) {};
        \draw (wo) -- (wt);
        \node[draw, fill, circle, minimum width=0.1cm, inner sep = 0, label={right:$b$}] (no) at ($ (q.north) + (0,0.2) $) {};
        \node[draw, fill, circle, minimum width=0.1cm, inner sep = 0, label={right:$a$}] (nt) at ($ (q.north) + (0,1) $) {};
        \draw (no) -- (nt);
        \node[draw, fill, circle, minimum width=0.1cm, inner sep = 0, label={right:$a$}] (so) at ($ (q.south) + (0,-0.2) $) {};
        \node[draw, fill, circle, minimum width=0.1cm, inner sep = 0, label={right:$b$}] (st) at ($ (q.south) + (0,-1) $) {};
        \draw (so) -- (st);
    \end{tikzpicture}
    \caption{By translational invariance, each term acting on qubit $q$ is equivalent to $h^{a,b}$.}
    \label{fig:hab}
\end{figure}
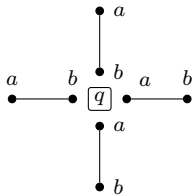
    
    As pictured in Figure \ref{fig:hab}, consider a qubit $q$ on the 2D lattice, whose Hilbert space is denoted by $\cH$.
    Since by assumption all terms are commuting, we have 
    $$
[\cA_h^a,\cA_h^b]=0.
    $$
By the Structure Lemma\footnote{Note that there are constructive proofs for the structure Lemma, as  in Section  7.3 of \cite{gharibian2015quantum}.}, that is Lemma \ref{lem:structure}, the algebra $\cA_h^a$ induces a decomposition
of $\cH$,
$$ \cH = \bigoplus_{i=1}^m \cH_i, \text{ where }
\cH_i = \cH_{(i,1)} \otimes \cH_{(i,2)}, \text{ and } \cA_h^a = \bigoplus_{i=1}^m \cL(\cH_{(i,1)}) \otimes \cI(\cH_{(i,2)}).
$$
Since $q$ is a qubit thus $\dim(\cH)=2$, we know either of the following cases holds:
\begin{itemize}
    \item[(a)] $m=1$ and $\dim(\cH_{(1,1)})=1$. Then $\cA_h^a$ acts trivially on $\cH$. In other words, $h$ is a single-qubit term acts on system $b$. By definition of local Hamiltonian, $h$ is a Hermitian term. We denote (give a name to) the eignevalues and eigenstates as $\lambda_0$,$\ket{0}$ and $\lambda_1$,$\ket{1}$. One can check that 
    $$
    H_{2D} =  2\sum_{q \text { on 2D}} \lambda_0 \ketbra{0}{0}_q + \lambda_1 \ketbra{1}{1}_q.
    $$
\item[(b)] $m=1$ and $\dim(\cH_{(1,1)})=2$. By the Structure Lemma we know $\cA_h^b$ acts trivially on $\cH$. 
This case can be handled in the same way as (a).
\item[(c)] $m=2$. In this case $\dim(\cH_i)=1$ for $i=1,2.$ We denote the basis for $\cH_1$, $\cH_2$ as $\ket{0},\ket{1}$. Then by the Structure Lemma,  $\cA^a_h$ keeps $\ket{0},\ket{1}$ invariant. Since $\cA^a_h,\cA^b_h$ commute, by Corollary \ref{cor:structure} $\cA^b_h$ also keeps $\ket{0},\ket{1}$ invariant.
Thus $h$  keeps the computational basis $\{0,1\}^2$ invariant. In other words, $h$ diagonalize in the computational basis, and thus is a linear combination of terms $$\ketbra{00}{00},\ketbra{01}{01},\ketbra{10}{10},\ketbra{11}{11}.$$
Note that
$$
    \ketbra{0} 0 = \tfrac 1 2 (Z+I)\text{ and } \ketbra 1 1=\tfrac 1 2 (I-Z).
$$
\end{itemize}
 We rewrite $H$ as the 2D Ising model with magnetic fields 
\begin{align}
    H_{2D} = \sum_{q \text{ on 2D}} \alpha_I I_q  + \alpha_Z Z_q +  \sum_{q,q' \text{ adjacent}} \beta Z_q\otimes Z_{q'},
\end{align}
where we implicitly use the fact that $H_{2D}$ is translation-invariant to get the same coefficient $\alpha_I,\alpha_Z,\beta$ for different $q,q'$.
    
\end{proof}

\section*{Author Contribution}

J.J. conceived the study. J.J and Y.H. performed the theoretical analysis.  All authors discussed the results and contributed to writing the manuscript.

\bibliographystyle{quantum}
\bibliography{ref.bib}
\end{document}